\documentclass{vldb}

\usepackage{balance}  % for  \balance command ON LAST PAGE  (only there!)

\usepackage{booktabs} % For formal tables

\usepackage{amsfonts}
\usepackage{balance}  % for  \balance command ON LAST PAGE  (only there!)
\usepackage{xcolor, graphicx, hyperref} %MnSymbol was removed due to error
\usepackage{multirow, framed, float, tabularx}
\usepackage{aliascnt}
\usepackage{relsize}
\usepackage{listings}
\usepackage{colortbl}
\usepackage{tikz,pgf}
\usepackage{mathrsfs}
\usetikzlibrary{arrows}
\usepackage{algorithm}
\usepackage{algorithmicx}
\usepackage{algpseudocode}
\usepackage{fancyvrb}
\usepackage[caption=false,font=footnotesize]{subfig}
\usepackage{enumitem}

\usepackage[bottom]{footmisc}

\usepackage{times}

\renewcommand{\footnotesize}{\small}
\newcommand{\oursystem}{{our framework}}

\newcommand{\attr}{{\mathcal{A}}}

\newcommand{\allclusters}{{\mathcal{C}}}
\newcommand{\soln}{{\mathcal {O}}}

\newcommand{\maxval}{{\tt Max{-}Avg}}
\newcommand{\minsize}{{\tt Min{-}Size}}

\newcommand{\val}{{\mathtt{val}}}
\newcommand{\avg}{{\mathtt{avg}}}
\newcommand{\cov}{{\mathtt{cov}}}

\newcommand{\bottomup}{{\tt Bottom-Up}}
\newcommand{\bottomupindiv}{{\tt Bottom-Up-Indiv}}
\newcommand{\fixedorder}{{\tt Fixed-Order}}
\newcommand{\randomorder}{{\tt Random-Order}}
\newcommand{\levelbased}{{\tt Bottom-Up-D}}
\newcommand{\hybrid}{{\tt Hybrid}}
\newcommand{\rhybrid}{{\tt R-Hybrid}}

\newcommand{\avgmax}{{\tt AvgMax}}
\newcommand{\deltajudgment}{{\tt Delta-Judgment}}

\definecolor{LightCyan}{rgb}{0.88,1,1}
\definecolor{Gray}{gray}{0.9}

\newcommand{\greedyavg}{{\tt Level-Greedy-Avg}}
\newcommand{\greedyavgL}{{\tt Level-Greedy-Avg-L}}
\newcommand{\mergeindiv}{{\tt Merge-Greedy-Indiv}}
\newcommand{\mergeall}{{\tt Merge-Greedy-All}}
\newcommand{\merge}{{\tt Merge}}
\newcommand{\updatesolution}{{\tt UpdateSolution}}

\newcommand{\yuhao}[1]{{\texttt {\color{magenta} Yuhao: [{#1}]}}}
\newcommand{\sudeepa}[1]{{{\tt \color{blue} Sudeepa: [{#1}]}}}

\newcommand{\red}[1]{#1} %{{{{\color{red} {#1}}}}}

\newcommand{\FindClusters}{{\tt FindClusters}}

\newcommand{\proj}[1]{{\Pi}}
\newcommand{\sel}[1]{{\sigma}}

\newcommand{\cut}[1]{}
\newcommand{\cutfull}[1]{}

\newcommand{\commentresolved}[1]{}

\newcommand{\ie}{{\it i.e.}} %\xspace}
\newcommand{\eg}{{\it e.g.}} %\xspace}
\newtheorem{theorem}{Theorem}[section]          	% Theorem environment.
\newaliascnt{lemma}{theorem}				% 1 alias counter
\aliascntresetthe{lemma}  					% 3 set
\newaliascnt{conjecture}{theorem}			% 1 alias counter
\aliascntresetthe{conjecture}  				% 3 set
\newaliascnt{remark}{theorem}				% 1 alias counter
              
\aliascntresetthe{remark}  					% 3 set
\newaliascnt{corollary}{theorem}			% 1 alias counter
\aliascntresetthe{corollary}  				% 3 set
\newaliascnt{definition}{theorem}			% 1 alias counter
\newtheorem{definition}[definition]{Definition}    % Definition environment.
\aliascntresetthe{definition}  				% 3 set
\newaliascnt{proposition}{theorem}			% 1 alias counter
\newtheorem{proposition}[proposition]{Proposition}  % proposition environment.
\aliascntresetthe{proposition}  				% 3 set
\newaliascnt{example}{theorem}			% 1 alias counter
\newtheorem{example}[example]{Example}  	% 2 environment.
\aliascntresetthe{example}  				% 3 set
\newaliascnt{observation}{theorem}			% 1 alias counter
\aliascntresetthe{observation}  				% 3 set
\cut{

}

\newcommand{\reva}[1]{#1}%{{\color{red}{#1}}}
\newcommand{\revb}[1]{#1}%{{\color{blue}{#1}}}
\newcommand{\revc}[1]{#1}%{{\color{magenta}{#1}}}
\newcommand{\ansa}[1]{#1}%{{\color{red}{#1}}}
\newcommand{\ansb}[1]{#1}%{{\color{blue}{#1}}}
\newcommand{\ansc}[1]{#1}%{{\color{magenta}{#1}}}
\begin{document}

\clearpage

% ****************** TITLE ****************************************

\title{Interactive Summarization and Exploration of\\ Top Aggregate Query Answers}

% possible, but not really needed or used for PVLDB:
%\subtitle{[Extended Abstract]
%\titlenote{A full version of this paper is available as\textit{Author's Guide to Preparing ACM SIG Proceedings Using \LaTeX$2_\epsilon$\ and BibTeX} at \texttt{www.acm.org/eaddress.htm}}}

% ****************** AUTHORS **************************************


\numberofauthors{1}

%  in this sample file, there are a *total*

% 1st. author

% \author{Yuhao Wen, Xiaodan Zhu, Sudeepa Roy, Jun Yang}
% \affaddr{Department of Computer Science, Duke University}
% %  \city{Durham} 
%   %\state{NC} 
% \email{{ywen, xdzhu, sudeepa, junyang}@cs.duke.edu}
% 2nd. author

\author{
\alignauthor
Yuhao Wen, Xiaodan Zhu, Sudeepa Roy, Jun Yang\\
       \affaddr{Department of Computer Science, Duke University}\\
       \email{{\{ywen, xdzhu, sudeepa, junyang\}}@cs.duke.edu}
}

% There's nothing stopping you putting the seventh, eighth, etc.
% author on the opening page (as the 'third row') but we ask,
% for aesthetic reasons that you place these 'additional authors'
% in the \additional authors block, viz.

% \additionalauthors{Additional authors: John Smith (The Th{\o}rv\"{a}ld Group, {\texttt{jsmith@affiliation.org}}), Julius P.~Kumquat
% (The \raggedright{Kumquat} Consortium, {\small \texttt{jpkumquat@consortium.net}}), and Ahmet Sacan (Drexel University, {\small \texttt{ahmetdevel@gmail.com}})}
% \date{30 July 1999}

% Just remember to make sure that the TOTAL number of authors
% is the number that will appear on the first page PLUS the
% number that will appear in the \additionalauthors section.

\maketitle

\begin{abstract}
%Recently, making results more meaningful has attracted a lot of attention since it can improve the quality of results retrieved by user queries. Previous research considered this problem from diversity, relevance and coverage. However, they only focused only one or two of these three aspects. In this paper, we propose a new, intuitive definition of ranking aggregate query result with relevance, diversity and coverage. Instead of outputing top-k tuples of the result, we output top-k clusters that each cluster covers some tuples. Then we show a efficient algorithm to solve this problem. A demo system ClusterDB is built to illustrate how users can interactively retrieve query result that contains more information than the original result.

We present a system for summarization and interactive exploration of high-valued aggregate query answers to make a large set of possible answers more informative to the user.
Our system outputs a set of clusters on the high-valued query answers showing their common properties such that the clusters are diverse as much as possible to avoid repeating information, and cover a certain number of top original answers as indicated by the user. Further, the system facilitates interactive exploration of the query answers by helping the user (i) choose combinations of parameters for clustering, (ii) inspect the clusters as well as the elements they contain, and (iii) visualize how changes in parameters affect clustering. We define optimization problems, study their complexity, explore properties of the solutions investigating the semi-lattice structure on the clusters, and propose efficient algorithms and optimizations to achieve these goals. We evaluate our techniques experimentally and discuss our prototype with a graphical user interface that facilitates this interactive exploration. \ansa{A user study is conducted to evaluate the usability of our approach.}
%constraints along with comparison with another algorithm.}
\end{abstract}

\cut{
Making top-ranked query answers more meaningful to users has attracted much attention recently. Previous research considered three aspects of answers returned to the user -- relevance, diversity, and coverage, although focused only on one or two of them at the same time. In this work, we propose a novel framework for summarizing aggregate query answers by clustering result tuples with identical values for subsets of the attributes, while considers all three aspects simultaneously. In particular, given a distance parameter D and coverage parameter L, instead of returning the top-k original result tuples ranked by their values, our framework outputs up to k clusters such that (a) these clusters are at least at distance D from each other (diversity), (b) together they cover the top-L original result tuples (coverage), and (c) the value of the clusters is maximized (relevance). We explore properties of the solutions investigating the semi-lattice structure on the clusters, present complexity results for the optimization problems, propose efficient algorithms to find the clusters, and evaluate them with experiments on real data. We have also built a prototype with an intuitive interface for users to explore query answer in a layered fashion (with clusters and the original result tuples they contain) and to adjust the input parameters interactively.

}

\setcounter{page}{1}
\begin{sloppypar}
% ****************** INTRODUCTION **********************************
\begin{figure*}[t]
\centering
\begin{minipage}[t]{.48\textwidth}
  \centering
 %\vspace*{\fill}
\subfloat[Top-8 and bottom-8 tuples with values as score]{{\scriptsize
\begin{tabular}{|c|c|c|c|c|c|}
\hline
{\tt Rank} & {\tt hdec} & {\tt agegrp} & {\tt gender} & {\tt occupation} & {\tt val}   \\\hline
%& {\tt decade} & {\tt range} & & & \\\hline
1 & 1975 	& 20s 	&  M 	& Student 	& 4.24 \\\hline
2 & 1980 	& 20s 	& M  	&  Programmer 	& 4.13 \\\hline
3 & 1980 	& 10s 	& M 	& Student  	&  3.96 \\\hline
4 & 1980 	& 20s		 & M 	& Student  	&  3.91  \\\hline
5 & 1985 	& 20s 	& M 	& Programmer  	&  3.86  \\\hline
6 & 1980 	& 20s 	& M & Engineer  	&  3.83  \\\hline
7 & 1985 	& 10s 	& M 	& Student 	& 3.77  \\\hline
8 & 1985 	& 20s 	& M 	& Student 	& 3.76  \\\hline
%9 &1995	& 30s 	& F 	& Educator  	&  370  \\\hline
%10&1985 	& 20s 	& M 	& Engineer  	& 3.65  \\\hline
\multicolumn{6}{|c|}{$\ldots\ldots\ldots\ldots\ldots\ldots$}\\\hline
%41 & 1995 & 60s  &  M  &  Retired  &  3.09 \\\hline
%42 & 1995 & 30s  &  M  &  Scientist  &  3.07  \\\hline
43 & 1995 & 30s  &  M  &  Marketing  &  3.02  \\\hline
44 & 1995 & 20s  &  M  &  Technician  &  2.92  \\\hline
45 & 1995 & 30s  &  M  &  Entertainment  &  2.91  \\\hline
46 & 1995 & 20s  &  M  &  Executive  &  2.91  \\\hline
47 & 1995 & 30s  &  F  &  Librarian  &  2.84  \\\hline
48 & 1995 & 30s  &  M  &  Student  &  2.81  \\\hline
49 & 1995 & 20s  &  M  &  Writer  &  2.51  \\\hline
50 & 1995 & 20s  &  F  &  Healthcare  &  1.98 \\\hline
\end{tabular}
}\label{fig:eg-intro-top-bottom}}
\end{minipage}~~
\begin{minipage}[t]{.48\textwidth}
 \vspace*{\fill}
\subfloat[Clusters with average score]{{\scriptsize
\begin{tabular}{|c|c|c|c||c|c|}
\hline

 {\tt hdec} & {\tt agegrp} & {\tt gender} & {\tt occupation} & {\tt avg val}  &  \\\hline
%{\tt decade} & {\tt range} & & &  \\\hline 
\rowcolor{Gray}
{\bf 1975}    & {\bf 20s}    & {\bf M}    & {\bf  Student}  &   {\bf 4.24 }   & $\blacktriangledown$  \\\hline
\rowcolor{Gray}
{\bf 1980}    & {\bf $*$}    & {\bf M}    & {\bf $*$}   & {\bf  3.96}  & $\blacktriangledown$   \\\hline
\rowcolor{Gray}
{\bf 1985}  &  {\bf 20s}  &   {\bf  M }   & {\bf Programmer}   & {\bf  3.86}  & $\blacktriangledown$   \\\hline
\rowcolor{Gray}
{\bf 1985}   & {\bf $*$}  &   {\bf M} &    {\bf Student} &   {\bf  3.76  }   & $\blacktriangledown$\\\hline
\end{tabular}
}
\label{fig:eg-intro-clusters}
}
\vfill
\subfloat[Original result tuples (with ranks) in the clusters]{{\scriptsize
%{\scriptsize
\begin{tabular}{|c|c|c|c||c|c|}
\hline
 {\tt hdec} & {\tt agegrp} & {\tt gender} & {\tt occupation} & {\tt avg val} & {\tt rank}  \\\hline
%{\tt decade} & {\tt range} & & & & \\\hline 
\rowcolor{Gray}
{\bf 1975}    & {\bf 20s}    & {\bf M}    & {\bf  Student}  &   {\bf 4.24 } & $\blacktriangledown$  \\\hline
1975 	& 20s 	&  M 	& Student 	& 4.24 & 1 \\\hline
\rowcolor{Gray}
{\bf 1980}    & {\bf $*$}    & {\bf M}    & {\bf $*$}   & {\bf  3.96}  &  $\blacktriangledown$ \\\hline
1980 	& 20s 	& M  	&  Programmer 	& 4.13  & 2\\\hline
1980 	& 10s 	& M 	& Student  	&  3.96 & 3\\\hline
1980 	& 20s		 & M 	& Student  	&  3.91 & 4 \\\hline
1980 	& 20s 	& M & Engineer  	&  3.83 & 6 \\\hline
\rowcolor{Gray}
{\bf 1985}  &  {\bf 20s}  &   {\bf  M }   & {\bf Programmer}   & {\bf  3.86} &   $\blacktriangledown$\\\hline
1985 	& 20s 	& M 	& Programmer 	& 3.86 & 5  \\\hline
\rowcolor{Gray}
{\bf 1985}   & {\bf *}  &   {\bf M} &    {\bf Student} &   {\bf  3.76  }  & $\blacktriangledown$\\\hline
1985 	& 10s 	& M 	& Student  	& 3.77 & 7  \\\hline
1985 	& 20s 	& M 	& Student 	& 3.76 & 8  \\\hline
\end{tabular}
%}
}
\label{fig:eg-intro-two-layer}
}
\end{minipage}
\vspace{-0.3cm}
\caption{Illustrating our framework for $k = 4$, $L = 8$, $D = 2$. In general, original result tuples outside top-$L$ can belong to the output clusters, although here we happen to have a solution that covers just the top-$L$ result tuples.}
\label{fig_sim}
\vspace{-0.5cm}
\end{figure*}

\section{Introduction}\label{sec:introduction}

Summarization and diversification of query results have recently drawn significant attention in databases and other applications 
such as keyword search, recommendation systems, and online shopping. The goal of both result summarization and result diversification is to make a large set of possible answers more informative to the user, since the user is likely not to view results beyond a small number. 
This brings the need to make the \emph{top-$k$} results displayed to the user \emph{summarized} (the results should be grouped and summarized to reveal high-level patterns among answers), \emph{relevant} (the results should have high \emph{value} or \emph{score} with respect to user's query or a database query), \emph{diverse} (the results should avoid repeating information), and also providing \emph{coverage} (the results should cover top answers from the original non-summarized result set).
In this paper, we present a framework to summarize and explore high valued aggregate query answers to understand their common properties easily and efficiently while meeting the above competing goals simultaneously. We illustrate the challenges and our contributions using the following example:
% for aggregate query answers, but also helps the user understand the high valued answers   It summarizes the high-valued answers in its top layer, while allowing further investigation of the original answer tuples in its second layer. We illustrate this view with an example.
\begin{example}\label{eg:intro}
%\red{need a new example with the new objective==========================}
%(movielens example goes here)
Suppose an analyst is using the movie ratings data from the \emph{MovieLens} website \cite{movielensdata} to investigate average ratings of different genres of movies by different groups of users over different time periods. So the analyst first joins several relations from this dataset (information about movies, ratings, users, and their occupations) to one relation $R$, extracts some additional features from the original attributes (age group, decade, half-decade), and then runs the following SQL aggregate query on  $R$ (the join is omitted for simplicity). In this query, {\tt hdec} denotes disjoint five-year windows of half-decades, \eg, 1990 (=1990-94), 1995 (=1995-99), etc.; {\tt agegrp} denotes age groups of the users in their teens or 10s (\ie, 10-19), 20s (\ie, 20-29), etc. %30s (\ie, 30-39) etc.

\vspace*{-0.5ex}
{\small
\begin{verbatim} 
SELECT hdec, agegrp, gender, occupation, avg(rating) as val
FROM R
GROUP BY hdec, agegrp, gender, occupation 
WHERE genres_adventure = 1
HAVING count(*) > 50
ORDER BY val DESC
\end{verbatim}
}
\vspace*{-0.5ex}
The  top 8 and bottom 8 results from this query are shown in Figure~\ref{fig:eg-intro-top-bottom}. {To have a quick summary of these 50 result tuples, The data analyst is interested in seeing the summary in at most four rows to have an idea of the %characteristics of the
 viewers and time periods with a high rating for the adventure genre. }
\end{example}
One straightforward option is to output the top 4 result tuples from Figure~\ref{fig:eg-intro-top-bottom}, but they do not summarize the common properties of the intended viewers/times periods. In addition,  despite having high scores, they have attribute values that are close to each other (\eg, male students in their 20s) leading to repetition of information and sub-optimal use of the designated space of $k = 4$ rows. More importantly, \emph{the top $k$ original tuples may give a wrong impression on the common properties of high-valued tuples even if they all share those properties}. For instance, three out of top four tuples share the property {\tt (20s, M)}, but it is misleading, since a closer look at Figure~\ref{fig:eg-intro-top-bottom} reveals that many tuples with low values (49th, 46th, 44th) share this property too, suggesting that male viewers in their 20s %(or male student viewers) 
may or may not give high ratings to the adventure genre. \reva{Therefore, we aim to achieve a summarization with the following desiderata: (i) it should be simple and memorable (\eg, male students or {\tt (M, Students)}), (ii) it should be diverse (\eg, {\tt (1975, 20s, M, Student)} and {\tt (1980, 20s, M, Student)} might be too similar), and (iii) it should be discriminative  (\eg, the properties like {\tt (20s, M)} covering both high and low valued tuples should be avoided). Furthermore, it should be achieved at an interactive speed and displayed using a user-friendly interface.}
 % Therefore, \emph{running a clustering algorithm on top-$k$ original tuples does not work as a meaningful summary}. 
\par
In recent years, work has been done to \emph{diversify} a set of result tuples by selecting a subset of them (discussed further in Section~\ref{sec:related}), \eg, \emph{diversified top-$k$} \cite{QinYC12} takes account of diversity and relevance while selecting top result tuples; \emph{DisC diversity} \cite{drosou2012disc}  takes into account similarity with the tuples that have not been selected, and diversity and relevance in the selected ones. In contrast, we intend to output summarized information on the result tuples by displaying the common attribute values in each cluster to give the user a holistic view of the result tuples with high value. In this direction, the \emph{smart drill-down} \cite{JoglekarGP16} framework helps the user explore summarized  ``interesting'' tuples in a database, but it does not focus on aggregate answers, or helping the user choose input parameters and understand consecutive solutions, which are two key features of our framework. \red{As discussed in Section~\ref{sec:related} and observed in experiments in our initial exploration, standard clustering or classification approaches do not give a meaningful summary of high-valued results as well.} In particular, we support summarization and interactive exploration of aggregate answers in the following ways each posing its own technical challenges.\\

%\vspace{0.5ex}
\noindent\textbf{(1) Summarizing Aggregate Answers with Relevance, Diversity, and Coverage.~~}
 \red{To meet the desiderata of a good summarization, the basic operation of our framework involves %\emph{two layers}: the \emph{top layer} shows a set of 
generating a set of \emph{clusters}  summarizing the common properties or common attribute values of high-valued answers (Section~\ref{sec:prelim}). If all elements in a cluster do not share the same value for an attribute, then the value of that attribute is replaced with a `$*$' \footnote{\red{Our framework and algorithms can be extended to more fine-grained generalizations of values beyond $*$ (by introducing a concept hierarchy over the domain). Details in Appendix~\ref{app:range}. %; more discussion can be found in the full version \cite{fullversion}.
}}.
%Range values (e.g., [10,20]) can be supported by our existing system by expanding ``$*$'' values to range values and assigning elements based on ranges. Details are present in the 
%full version\cite{fullversion}.}}  
%Details are presented in Section~\ref{sec:prelim}.} 
The %\emph{bottom layer} shows the \emph{elements} (the original answer tuples) \red{the clusters} contain. 
clusters can be expanded to show the elements contained in them to the user.} To compute the clusters, our framework can take (up to) three parameters as input: (i) \emph{size constraint} $k$ denotes the number of  rows or clusters to be displayed %in the \emph{top layer} of the framework 
($k$ = 4 in Example~\ref{eg:intro}), (ii) \emph{coverage parameter $L$},  requiring that the top-$L$ tuples in the original ranking must be covered by the $k$ chosen clusters, and (iii) \emph{distance parameter $D$}, requiring that the summaries should be at least distance $D$ from each other to avoid repeating similar information. 
\cut{
Once the results are summarized given the values of $k, L, D$, the user can see the original result tuples in the \emph{second layer} of the framework by expanding the summaries or \emph{clusters} shown in the top layer. We  illustrate the concepts here using the running example (further details in Section~\ref{sec:framework}):
}
\begin{example}\label{eg:intro-clusters}
Suppose we run \oursystem\ for the query in Example~\ref{eg:intro} with parameters $k = 4, L = 8$, and $D = 2$, \ie, the user would like to see at most 4 clusters, % in the top layer, 
these clusters should cover top $8$ tuples from Figure~\ref{fig:eg-intro-top-bottom}, and any two clusters should not have identical values for more than two attributes.
%(We defer the formal definition of distance to Sections~\ref{sec:prelim}.) 
Our framework first displays the four clusters shown in Figure~\ref{fig:eg-intro-clusters} %as the top layer of the solution 
along with 
%, where the score displayed for a cluster is the 
the average scores of result tuples contained in them.

\cut{
\begin{figure}[h]
{\centering
{\scriptsize
\begin{tabular}{|c|c|c|c||c|c|}
\hline
 {\tt hdec} & {\tt agegrp} & {\tt gender} & {\tt occupation} & {\tt avg val}  &  \\\hline
 %%%Setting 3: L=8, k =4, D = 2, Bottom-up(101)
%{\tt decade} & {\tt range} & & &  \\\hline 
{\bf 1975}    & {\bf 20s}    & {\bf M}    & {\bf  Student}  &   {\bf 4.24 }   & $\blacktriangledown$  \\\hline
{\bf 1980}    & {\bf $*$}    & {\bf M}    & {\bf $*$}   & {\bf  3.96}  & $\blacktriangledown$   \\\hline
{\bf 1985}  &  {\bf 20}  &   {\bf  M }   & {\bf Programmer}   & {\bf  3.86}  & $\blacktriangledown$   \\\hline
{\bf 1985}   & {\bf $*$}  &   {\bf M} &    {\bf Student} &   {\bf  3.76  }   & $\blacktriangledown$\\\hline
\end{tabular}
%\begin{tabular}{|c|c|c|c||c|c|}
%\hline
% {\tt hdecade} & {\tt agegrp} & {\tt gender} & {\tt occupation} & {\tt avg score}  &  \\\hline
%%{\tt decade} & {\tt range} & & &  \\\hline 
%{\bf 1975}    & {\bf 20s}    & {\bf M}    & {\bf  Student}  &   {\bf 4.24 }   & $\blacktriangledown$  \\\hline
%{\bf 1980}    & {\bf $*$}    & {\bf M}    & {\bf $*$}   & {\bf  3.96}  & $\blacktriangledown$   \\\hline
%{\bf 1985}  &  {\bf $*$}  &   {\bf  M }   & {\bf $*$}   & {\bf  3.76}  & $\blacktriangledown$   \\\hline
%{\bf 1995}   & {\bf 30s}  &   {\bf F} &    {\bf Educator} &   {\bf  3.70  }   & $\blacktriangledown$\\\hline
%\end{tabular}
}
\caption{{\small \bf Clusters with average score.}} %, $k = 4, D = 2, L = 8$}}
\vspace{-0.5cm}
\label{fig:eg-intro-clusters}
}
\end{figure}
}

The user may choose to investigate any of these clusters by expanding the cluster on \oursystem\ (clicking  $\blacktriangledown$). If all four clusters are expanded by the user, the second-layer  will reveal all original result tuples they cover, as shown in Figure~\ref{fig:eg-intro-two-layer}. In this particular example, no other tuples outside top 8 have been chosen by our algorithm (which is also the optimal solution), but in general, the selected clusters may contain other tuples (high-valued but not necessarily in top $L$). 

\cut{
\begin{figure}[h]
{\centering
{\scriptsize
\begin{tabular}{|c|c|c|c||c|c|}
\hline
 {\tt hdec} & {\tt agegrp} & {\tt gender} & {\tt occupation} & {\tt avg val} & {\tt rank}  \\\hline
%{\tt decade} & {\tt range} & & & & \\\hline 
\rowcolor{Gray}
{\bf 1975}    & {\bf 20s}    & {\bf M}    & {\bf  Student}  &   {\bf 4.24 } & $\blacktriangledown$  \\\hline
1975 	& 20s 	&  M 	& Student 	& 4.24 & 1 \\\hline
\rowcolor{Gray}
{\bf 1980}    & {\bf $*$}    & {\bf M}    & {\bf $*$}   & {\bf  3.96}  &  $\blacktriangledown$ \\\hline
1980 	& 20s 	& M  	&  Programmer 	& 4.13  & 2\\\hline
1980 	& 10s 	& M 	& Student  	&  3.96 & 3\\\hline
1980 	& 20s		 & M 	& Student  	&  3.91 & 4 \\\hline
1980 	& 20s 	& M & Engineer  	&  3.83 & 6 \\\hline
\rowcolor{Gray}
{\bf 1985}  &  {\bf 20}  &   {\bf  M }   & {\bf Programmer}   & {\bf  3.86} &   $\blacktriangledown$\\\hline
1985	& 20s 	& M 	& Programmer  	&  3.86 & 5 \\\hline
\rowcolor{Gray}
{\bf 1985}   & {\bf $*$}  &   {\bf M} &    {\bf Student} &   {\bf  3.76  }  & $\blacktriangledown$\\\hline
1985 	& 20s 	& M 	& Student 	& 3.76 & 7  \\\hline
1985 	& 10s 	& M 	& Student  	& 3.76 & 8 \\\hline
\end{tabular}
}
\vspace{-2mm}
\caption{{\small \bf Original result tuples (with ranks) in the clusters if all clusters are expanded.}}
\vspace{-6mm}
\label{fig:eg-intro-two-layer}
}
\end{figure}
}

% (the last column shows the rank of the original top $L$ elements):
\end{example}
The above example illustrates several advantages and features of \oursystem\ in providing a meaningful and holistic summary of high-valued aggregate query answers. {\em First,} the original top 8 result tuples are not lost thanks to the second layer, whereas the properties that combine multiple top result tuples are clearly highlighted in the clusters in the first layer. {\em Second,} the chosen clusters are diverse, each contributing some extra novelty to the answer. {\em Third,} the clustering captures the properties of the top result tuples that distinguish them from those with low values. \reva{For instance, the cluster for  {\tt (20s, M)} does not appear in the solution, since this property is prevalent in both high-valued and low-valued tuples as discussed before. Clearly, this could not be achieved by simply clustering top $L$ tuples by $k$ clusters. This is ensured by our objective function that aims to maximize the {\em average} value of the tuples covered by all clusters (instead of maximizing the sum).}  
% since clusters containing low-valued elements are likely to have low average values; this cannot be achieved simply by clustering the top $L$ tuples with $k$ clusters as discussed before (\eg, (20s, M)). 
%Our framework is particularly useful for 
\par
To achieve the solution as described above, we make the following technical contributions in the paper:
\begin{itemize}[leftmargin=*]
\itemsep0em
\item To ensure that the chosen clusters cover answer tuples with high values, we formulate an optimization problem that takes $k, L, D$ as input, and outputs clusters such that the \emph{average value} of the tuples covered by these clusters is maximized.
We study the complexity of the above problem (both decision and optimization versions) and show NP-hardness results (Section~\ref{sec:properties}).
\item We design efficient heuristics \reva{satisfying all constraints} using properties %given by 
of the %natural 
\emph{semi-lattice structure} on the clusters imposed by the attributes (Section~\ref{sec:algorithms}). % values. 
\item We perform extensive experimental evaluation using the MovieLens \cite{movielensdata} and TPC-DS benchmark \cite{nambiar2006making} datasets (Section~\ref{sec:experiments}).
\end{itemize}
%\vspace{0.5ex}
\noindent\textbf{(2) Interactive Clustering and Parameter Selection.~~}
The intended application of our framework is an interactive exploration of query results where the user 
%starts with original top $k$ answers (with $D = L = 0$), and %by noticing many overlapping attributes, 
may keep updating $k, L$, or $D$ to understand the key properties of the high-valued aggregate answers. \ansa{One challenge in this exploration is to select values of $k, L, D$ while ensuring interactive speed, since straightforward implementations of our algorithms would not be fast enough. To support parameter selection, we provide the user with a holistic view of how the objective varies with different choices of parameters. This view helps users identify ``flat regions'' (uninteresting for parameter changes) vs.\ ``knee points'' (possibly interesting for parameter changes) in the parameter space.
One example is shown in Figure~\ref{fig:intro-guid}, where given selected values of $L = 15$ as in  Section~\ref{sec:guid-visualization}, how the average value of the solutions ($y$-axis) varies with different $k$ ($x$-axis) is shown. This figure illustrates that if $k$ is changed from  $k = 11$ to $k = 7$, there will be a drop in the overall value. The user can select different legends for different lines and check the value in detail by hovering over a point. This visualization also helps the user validate the choice of parameters, \eg, if a smaller value of $k$ can give a similar quality result, the user may want to reduce the value of $k$ to have a more compact solution. 
\cut{
Indeed, an alternative approach of automatically selecting ``good'' parameter values would be an interesting direction of future work, whereas this work focuses on helping users choose parameter values since the ``goodness'' of parameters often depends on the tasks and preferences of different users (more in Section~\ref{sec:guid-visualization}).}
}
\reva{This feature not only helps in guiding the user select parameter values\footnote{\ansc{
To further assist in parameter selection, our system also allows visual comparison of two successive solutions showing how the clusters are redistributed. We formulated an optimization problem to enable clean visualization and provided optimal solutions. 
%We also produce a visualization showing the changes in successive solutions, as shown in Figure~\ref{fig:intro-vizcompare}. 
%Due to space constraints, 
The details are in Appendix~\ref{sec:visualization} and in our demonstration paper in SIGMOD 2018~\cite{qagviewdemo2018}.}}, 
it also serves as a \emph{precomputation} step to retrieve the actual solutions 
for different combinations of input parameters $k, L$, and $D$ at an interactive speed.}
%at an interactive speed 
(Section~\ref{sec:guidance}). 
\begin{itemize}[leftmargin=*]
\vspace{-0.5mm}
\itemsep0em
\item We develop techniques for \emph{incremental computation and efficient storage} for solutions for multiple combinations of input parameters %and store them efficiently 
using an \emph{interval tree} data structure.
\item We implement multiple optimization techniques to further speed up computation of these solutions. We evaluate the effect of these optimizations experimentally. \reva{Eventually we achieved 30x-1000x speed up from these optimizations, which helped us achieve our goal of interactive speed.}
\vspace{-0.5mm}
\end{itemize}

\cut{
(b) However, the straightforward implementations of these heuristics did not give interactive speed of our algorithms, which we solved by developing multiple optimizations (precomputation, batch computation, interval tree data structure, incremental updates, using hash values for non-numeric fields etc.) (also see Section~\ref{sec:guidance}). Eventually we achieved 30x-1000x speed up from these optimizations, which helped us achieve our goal of interactive speed.\\
(c) As an additional feature, these optimizations provided the users a visual guide for parameter selection showing a holistic view of how the objective varies with different choices of parameters, and identify ``flat regions'' (uninteresting for parameter changes) vs. ``knee points'' (possibly interesting for parameter changes).\\
}

\begin{figure}[t]
\includegraphics[width=0.9\linewidth]{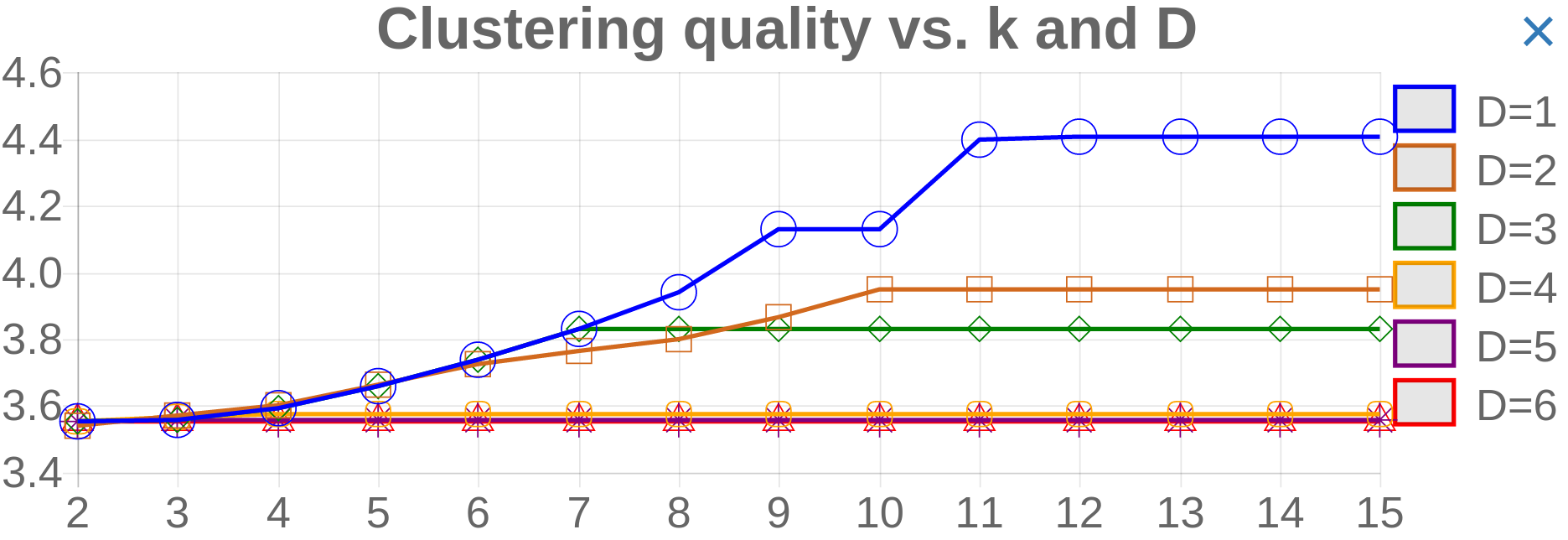}
\vspace*{-2ex}
\caption{\ansc{Visualization for parameter selection: how results vary for different $k$ and $D$ (some lines overlap).}}
\vspace{-2mm}
\label{fig:intro-guid}
\end{figure}

\vspace{0.5ex}
\noindent\textbf{Roadmap.~~}  We discuss the related work in Section~\ref{sec:related} and define some  preliminary concepts in Section~\ref{sec:prelim}. The above sets of results are discussed in Sections~\ref{sec:properties}, \ref{sec:algorithms}and \ref{sec:guidance}.  %Section~ %\ref{sec:algorithms} describes the algorithms and 
The experimental results are presented in Section~\ref{sec:experiments}. %.  We discuss the related work in Section~\ref{sec:related} and 
\ansa{A user study is conducted on Section~\ref{sec:user-study}.} We conclude in Section~\ref{sec:conclusions} with scope of future work. Some details are deferred to Appendix. %the full version\cite{fullversion} due to space constraints.
\section{Related Work}\label{sec:related}
First we discuss three recent papers relevant to our work that consider {\em result diversification} or {\em result summarization}: smart drill-down \cite{JoglekarGP16}, diversified top-$k$ \cite{QinYC12}, and DisC diversity\cite{drosou2012disc}. We explored using or adapting the approaches proposed in these papers for our problem, but since they focus on different problems, as expected, the optimization, objective, and the setting studied in \cite{JoglekarGP16, QinYC12, drosou2012disc} do not suffice to meet the goals in our work; %we discuss these qualitative comparison results in the full version \cite{fullversion}. %Appendix~\ref{sec:expt-compare-others}.  %We give an overview of these papers in this section. 
There are several other related work in the literature %including standard clustering and classification methods  
that we briefly mention %in this section, 
below. Qualitative comparison results and details are discussed in Appendix~\ref{sec:more-related}. %the full version \cite{fullversion}. 
%We compare the results from our work qualitatively with them in Section~\ref{sec:expt-compare-others}. In this section, we briefly describe these three approaches and discuss other related work.\\
%%%
\par
\textbf{Smart drill-down \cite{JoglekarGP16}:} 
In a recent work, Joglekar et al. \cite{JoglekarGP16} proposed the {\em smart drill-down} operator for interactive exploration and summarizing interesting tuples in a given table. The outputs show top-$k$ rules (clusters) with don't-care $*$-values. The goal is to find an ordered set of rules with maximum score, which is given by the sum of product of the  {\em marginal coverage} (elements in a rules that are not covered by the previous rules) and {\em weight}  of the rules (a ``goodness'' factor, \eg, a rule with fewer $*$ is better as it is more specific).
%As we discuss in Section~\ref{sec:expt-compare-others}, 
In Appendix~\ref{sec:expt-compare-others}, we show with examples that this approach is not suitable for summarizing aggregate query answers, since it will prefer common attribute values prevalent in many tuples and may select rules containing both high- and low-valued tuples.
\par
\cut{
In a recent work, Joglekar et al. \cite{JoglekarGP16} proposed the {\em smart drill-down} operator for interactive exploration and summarizing interesting tuples in a given table. The outputs show top-$k$ rules with don't-care $*$-values (clusters in our framework), and either a rule $r \in R$ or an attribute with the $*$-value can be expanded further  (rule and star drill downs). The goal is to find an ordered set of rules with maximum score with respect to one of these drill down operations, where the score is given by the sum (over all rules) of product of the  {\em marginal count or coverage} (how many original elements a rule covers that are not covered by the previous rules in the list) and {\em weight}  of the rules (a ``goodness'' factor independent of the count, \eg, a rule with fewer $*$ is better as it is more specific).
Unlike our work that considers summarizing aggregate answers with {\em values or scores} attached to individual elements, \cite{JoglekarGP16} focuses on summarizing the input table, and  uses another notion of diversity by using {\em marginal counts} instead of distance between clusters. As we discuss in Section~\ref{sec:expt-compare-others}, this approach is not suitable for summarizing aggregate query answers, since it will prefer common attribute values prevalent in many tuples, and in this process may select attribute values that are common to both high-valued and low-valued tuples.\\}
%%%
\textbf{Diversified top-$k$ \cite{QinYC12}:} 
Qin et al. \cite{QinYC12} formulated the top-$k$ result diversification problem: given a relation $S$ where each element has a score and any two elements have a similarity value between them, output at most $k$ elements such that any two selected elements are dissimilar (similarity $>$ a threshold $\tau$), and maximize the sum of the scores of the selected elements. \cite{QinYC12} considers diversification, but it does not consider result summarization using $*$-values (it chooses individual representative elements instead). In addition to lacking high level properties, % with $*$ values, 
this adapted process would possibly lose the holistic picture since some low-valued elements may be assigned to the chosen representatives from the top elements. 
%Details are further discussed in Section~\ref{sec:expt-compare-others}).\\
\cut{
Qin et al. \cite{QinYC12} formulated the top-$k$ result diversification problem as
%  a bi-criteria optimization problem: 
follows: given a relation $S$ where each element has a score and any two elements have a similarity value between them, the goal is to output at most $k$ elements such that any two selected elements are dissimilar (similarity or distance $> $ a threshold $\tau$) and the sum of the scores of the selected elements is maximized. They proposed procedures for diversification that can extend many existing top-$k$ query processing frameworks to handle diversified top-$k$ search.
Unlike \cite{JoglekarGP16}, \cite{QinYC12} considers diversification of aggregate answers with scores like our work. It does not consider result summarization using $*$ attribute values (unlike us or \cite{JoglekarGP16}), but elements with similarity $\leq \tau$ with some of the chosen representative elements can be assigned to those representatives forming a \emph{cluster}. However, in addition to lacking high level properties with $*$ values, this adapted process would pick representatives from the top elements but lose the holistic picture since some of the lower ranked and low-valued elements may be assigned to the chosen representatives from the top elements. Our objective functions of maximizing average value or minimizing redundant elements address this concern and return clusters with overall high scores. (further discussed in   Section~\ref{sec:expt-compare-others}).\\
}
%%%%
\par
\textbf{DisC diversity \cite{drosou2012disc}:} 
Drosou and Pitoura \cite{drosou2012disc} proposed %the notion of 
\emph{DisC} diversity: %(which stands for \emph{dissimilarity} and \emph{coverage}): 
given a set of elements $S$, the goal is to output a subset $S'$ of smallest size such that all the elements in $S$ are \emph{similar to} at least one element in $S'$ (\ie, have distance at most a given threshold $\tau$), whereas no two elements in $S'$ are similar to each other (distance is $> \tau$). Here %summarization and 
diversification can be achieved similar to \cite{QinYC12}. However, it ignores the values or relevance of the elements (unlike us or \cite{QinYC12}), and has no bound on the number of elements returned (unlike us, \cite{QinYC12, JoglekarGP16}). Therefore this approach may not be useful when the user wants to investigate a small set of answers, and it does not provide a summary of common properties of high valued tuples. 

\textbf{Classification and clustering:}
  Classification and clustering have been extensively studied in the literature. Various classification algorithms like Naive Bayes Classifier~\cite{murphy2006naive} and decision trees~\cite{quinlan1986induction} are widely used and are easy to implement. A simpler variation of our problem---separating top-$L$ elements from others---can be cast as a classification problem. However, this formulation would completely ignore values of elements outside top $L$, whereas our problem considers all element values and uses the top-$L$ elements only as a coverage constraint. 
  One could also formulate the problem of clustering the top-$L$ elements and apply the standard $k$-means algorithm~\cite{hartigan1975clustering} and its variants (\eg, \cite{huang1998extensions, wagstaff2001constrained}).  However, such algorithms do not produce clusters with simple and concise descriptions, and their clustering criteria do not consider values of elements outside top $L$.
  Therefore, it is necessary to find a new approach other than traditional clustering and classification. % algorithms.

\textbf {Other work on result diversification, summarization, and exploration:}
Diversification of query results has been extensively studied in the literature for both query answering in databases and other applications \cite{Carbonell+1998, agrawal2009diversifying, GollapudiD2009, Ziegler+WWW2005, Yu+EDBT2009, DBLP:conf/edbt/GkorgkasVDN15, Xin+2006, FanWW13, ZhuGGA07, RaviRT94, Borodin+2012, DBLP:conf/kdd/AbbassiMT13, Tao09, DengFan2014, Vieira+ICDE2011, QinYC12, drosou2012disc, Fraternali+2012, Zheng2012, VeeSSBA08, Zaharioudakis:2000, Sarawagi00,jagadish2004itcompress,jagadish1999semantic}. 
These include the \emph{MMR (Maximal Marginal Relevance)}-based approaches, the \emph{dispersion} problem studied in the algorithm community, 
\emph{diverse skyline}, summarization in \emph{text and social networks}, relational data summarization and OLAP data cube exploration among others. \ansc{The MMR-based and dispersion approaches consider diversification of results, outputting a small, diverse subset of relevant results, but do not summarize all relevant results. Others focus on various application domains and all have problem definitions different from this work.} %We discuss these related work in further detail in the full version \cite{fullversion}.

\section{Preliminaries}\label{sec:prelim}
%\greedyavg, \greedyavgL
%The most general situation
Let $R$ be a relation with attributes $\attr$, which can either be an input table (base relation) % or extensional database symbol (EDB)), 
\revb{or a derived relation possibly coming from a complex sub-query involving multiple tables.}
%, or intensional database symbol (IDB)). 
%The \emph{active domains} of attributes $A_1, \cdots, A_m$ are $D_1, \cdots, D_m$ respectively. 
Let $\attr_{groupby} \subseteq \attr$ %\{A_1, A_2, \cdots, A_m\}
be a set of \emph{grouping attributes} used in the group-by clause where $|\attr| = m$.
Let $aggr$ be any aggregate function allowed by SQL that outputs a real number.
% where $B \notin \attr_{groupby}$\footnote{The aggregate function can take multiple attributes as input as well.}. 
Therefore, we are considering a query $Q$ of the form: %-- $\pmb{(Q)}$
{\small
\newsavebox\sqlone
\begin{lrbox}{\sqlone}\begin{minipage}{\textwidth}
\lstset{language=SQL, basicstyle=\ttfamily, tabsize=2}
\begin{lstlisting}[mathescape]
SELECT $\attr_{groupby}$, aggr as val          
FROM R <base relation or output of a sub-query>
GROUP BY $\attr_{groupby}$
ORDER BY val DESC	
\end{lstlisting}
\end{minipage}\end{lrbox}
\resizebox{0.9\textwidth}{!}{\usebox\sqlone}
}
%\noindent
%The top-k situation
We denote the output of this query as $S$, where each tuple in $S$ is called an \emph{original element}. Here {\tt val} denotes the \emph{score} or \emph{value} of each output tuple in $S$ \reva{signifying relevance or importance of the tuple in response to the input query\footnote{\reva{Our work is also applicable to the settings where the scores of tuples do not come from an SQL query (\eg, are given by a domain expert).}}}. %, without the requirement of diversity and coverage. 
Usually, the query will output $n$ tuples (\ie, $|S| = n$). Even if the number of attributes $m$ in the group-by clause is small, $n$ might be large due to large domains of the participating attributes.  %For extracting the  part of the query result, a user can run  
Therefore, a user frequently runs a top-$L$ query to retrieve the top-$L$ tuples (denoted by $S_L^*$) with highest scores (adding a \texttt{LIMIT L} clause to the above query). 
\cut{
{\small
\newsavebox\sqltwo
\begin{lrbox}{\sqltwo}\begin{minipage}{\textwidth}
\lstset{language=SQL, basicstyle=\ttfamily, tabsize=2}
\begin{lstlisting}[mathescape]
SELECT $\attr_{aggr}$, aggr(B) as val         -- $\pmb{(Q_k)}$
FROM R
GROUP BY $\attr_{aggr}$
ORDER BY val DESC LIMIT $k$
\end{lstlisting}
\end{minipage}\end{lrbox}
\resizebox{0.9\textwidth}{!}{\usebox\sqltwo}
}
}

\cut{
%\noindent
We denote the subset of $S$ containing the top-$L$ tuples as $S_L^*$. 
%The query above returns $S_k^*$, \ie, a relatively smaller answer set. 
Since such top-$k$ tuples may contain 
%high overlaps in attribute values and therefore return
 redundant information, 
 instead of $S_k^*$, our goal is to output an answer set $\soln$ that achieves a good balance of 
 %relevance, diversity, and coverage.
 (i) \emph{relevance}: the answers returned have high scores, (ii) \emph{diversity}: the answers returned are not too similar to each other, and (iii) \emph{coverage}: the answers contain a desired number of top tuples from $S$.

 %\red{can remove last line.}
%\red{add a running example here - a full table sorted in val -- the top values are very similar}
}

%\subsection{Two-Layered Framework}\label{sec:two-layer}
%In Section \ref{sec:introduction}, we provide an intuitive view of the output structure. Here, we definite it more formally. 
%The output of the query $Q1$ is denoted by $S$, which has $m$ original attributes $A_1, \cdots, A_m$ (called standard attributes) having domains $D_1, \cdots, D_m$ respectively, and an aggregate value $val$ (called the aggregate output). 
\textbf{Clusters.~~} To display a solution with relevance, diversity, and coverage, our output is provided in \emph{two layers}: the top layer displays a set of {\em clusters} that hide the values of some attributes by replacing them with \emph{don't-care} ($*$) values, and the second layer contains the original elements covered by them.

\cut{
%\footnote{There is a natural extension of our framework to multiple layers, which we will study in the future.}:
\begin{enumerate} %[leftmargin=*]
\itemsep0em
\item The top layer displays a set of { \bf clusters} (defined below), that hide the values of some attributes by replacing them with \emph{don't-care} ($*$) values. % along with their values; these clusters satisfy additional requirements like diversity and coverage as explained below.
\item The users can expand any cluster and see the {\bf original elements} covered by the cluster along with their value in the relation $S$. 
\end{enumerate}
}

%The reason behind displaying clusters to achieve our goal is that clusters that contain top elements can maintain relevance, diversity, and coverage at the same time.
%Next, we give the formal definitions.
\par 
For every original element $t$ in the output $S$ of $Q$, %(defined on $A_1, \cdots, A_m$), we denote 
let $\val(t)$ denote the value or score of $t$. Other than the value, each $t \in S$ has $m$ attributes $A_1, \cdots, A_m$ with active domains $D_1, \cdots, D_m$ respectively. 
%Instead of $k$ elements in $S$, we will output at most $k$ \emph{clusters} given by a set of attribute-value pairs, where some attributes can be optional denoted by $*$. 
A \emph{cluster} $C$ on $S$ has the form: ${C \in \prod_{i = 1}^{m} D_{i}\cup \{*\}}$.
%a cluster  is of the form:
%A cluster $C(e, v)$ is a set of tuples $\{t_1, t_2, \cdots, t_m\}$. $v$ is the value of $C$.
%\begin{displaymath}
%{C \in \prod_{i = 1}^{m} D_{i}\cup \{*\}}
%\end{displaymath}
 Let $\allclusters$ denote the set of all clusters for relation $S$. % considering the active domain of the attributes $D_1, \cdots, D_m$.
We assume that the $m$ attributes $A_1, \cdots, A_m$ 
have a predefined order, and therefore we omit their names to specify a cluster.
For instance, for $m = 4$ attributes $A_1, A_2, A_3, A_4$, $(a_1, b_1, *, *)$ implies that $(A_1 = a_1) \wedge (A_2 = b_1)$, and the values of $A_3$ and $A_4$ are don't-care ($*$). We denote the value of an attribute $A_i$ of $C$ by $C[A_i]$; where $C[A_i]  \in D_i \cup \{*\}$, $i \in [1, m]$. 
In particular, each element $t$ in $S$ also qualifies as a cluster, which is called a \emph{singleton cluster}. 
%Next we define elements \emph{covered} by a cluster.
%For instance, in our running example, $(a_1, b_1, *, *)$ or $(a2, *, b2, c2)$ are example clusters; the first cluster captures that $(A = a_1) \wedge (B = b_1)$, the values of $C$ and $D$ are don't-care ($*$). \red{complete}. We denote the value of attribute $A_i$ of $C$ by $C[A_i]$; $C[A_i] \in D_i \cup \{*\}$. In particular, each element in $S$ also qualifies as a cluster.
%\par 
%%definition of coverage
%{\bf Coverage} Ne
\cut{
\begin{definition}\label{def:cov}
A cluster $C$ \emph{covers} another cluster $C'$ if 
$\forall i \in \left[ 1,m \right]$, $C\left[ A_{i} \right] = *$ or $C \left[ A_{i} \right] = C'\left[ A_{i} \right]$.
%A cluster covers an tuple $t$, if $t$ is consistent with $e$. i.e.\
\end{definition}
}
\par
A cluster $C$ \emph{covers} another cluster $C'$ if 
$\forall i \in \left[ 1,m \right]$, $C\left[ A_{i} \right] = *$ or $C \left[ A_{i} \right] = C'\left[ A_{i} \right]$.
Since each element $t$ in $S$ is also a cluster, each cluster $C$ covers some elements from $S$. Further, the notion of coverage naturally extends to a subset of clusters $\soln$. For $C \in \allclusters$, $\cov(C) \subseteq S$ denotes the elements covered by $C$, and 
for $\soln \subseteq \allclusters$, $\cov(\soln) \subseteq S$ denotes the elements covered by at least one cluster in $\soln$, \ie, $\cov(\soln) = \cup_{C \in \soln} \cov(C)$. Figure~\ref{fig:clusters} shows two clusters $C_1 = (*,*, c_1, d_1)$, $C_2 = (a_2, b_1, *, d_1)$, and the elements they cover. 
%For any tuple $t = (a, b, c_{1}, d_{1})$ where $a \in D_1, b \in D_2$, $t$ is covered by $C_1$. Also n
Note that two clusters may have overlaps in elements they cover. Here $C_1, C_2$ 
 have overlap on the tuple $(a_{2}, b_{1}, c_{1}, d_{1})$.
%\begin{figure}[ht]
%\begin{center}
%{\small
%\begin{tabular}{|c||c|c|c|c|}
%\hline
%& $A$ & $B$ & $C$ & $D$ \\\hline %& $\val$
%$C_1$ & $*$ & $*$ & $c_{1}$ & $d_{1}$ \\\hline % 47.5 \\\hline
%    & $a_{1}$ & $b_{2}$ & $c_{1}$ & $d_{1}$ \\ %& 49\\
% & $a_{1}$ & $b_{3}$ & $c_{1}$ & $d_{1}$ \\ %& 48\\
% & $a_{1}$ & $b_{4}$ & $c_{1}$ & $d_{1}$ \\ %& 47 \\
% & $a_{2}$ & $b_{1}$ & $c_{1}$ & $d_{1}$ \\\hline %& 46\\\hline
%$C_2$ & $a_{2}$ & $b_{1}$ & $*$ & $d_{1}$ \\\hline % &\\\hline % 38.75 \\\hline
% & $a_{2}$ & $b_{1}$ & $c_{1}$ & $d_{1}$ \\ %& 46\\
%    & $a_{2}$ & $b_{1}$ & $c_{4}$ & $d_{1}$ \\\hline %& 40\\\hline
%% & $a_{2}$ & $b_{1}$ & $c_{2}$ & $d_{1}$ & 35\\
%% & $a_{2}$ & $b_{1}$ & $c_{3}$ & $d_{1}$ & 34 \\\hline
%\end{tabular}
%}
%\end{center}
%\vspace{-0.2cm}
%\caption{Example clusters and the elements they cover.}
%\vspace{-0.2cm}
%\label{fig:clusters}
%\end{figure}

\begin{figure}[t]
\centering
%\begin{minipage}[t]{.4\textwidth}
 %\vspace*{\fill}
\vspace{-2ex}
\subfloat[{\scriptsize }]{
\tiny
\begin{tabular}{|c||c|c|c|c|}
\hline
& $A$ & $B$ & $C$ & $D$ \\\hline %& $\val$
$C_1$ & $*$ & $*$ & $c_{1}$ & $d_{1}$ \\\hline % 47.5 \\\hline
    & $a_{1}$ & $b_{2}$ & $c_{1}$ & $d_{1}$ \\ %& 49\\
 & $a_{1}$ & $b_{3}$ & $c_{1}$ & $d_{1}$ \\ %& 48\\
 & $a_{1}$ & $b_{4}$ & $c_{1}$ & $d_{1}$ \\ %& 47 \\
 & $a_{2}$ & $b_{1}$ & $c_{1}$ & $d_{1}$ \\\hline %& 46\\\hline
$C_2$ & $a_{2}$ & $b_{1}$ & $*$ & $d_{1}$ \\\hline % &\\\hline % 38.75 \\\hline
 & $a_{2}$ & $b_{1}$ & $c_{1}$ & $d_{1}$ \\ %& 46\\
    & $a_{2}$ & $b_{1}$ & $c_{4}$ & $d_{1}$ \\\hline %& 40\\\hline
% & $a_{2}$ & $b_{1}$ & $c_{2}$ & $d_{1}$ & 35\\
% & $a_{2}$ & $b_{1}$ & $c_{3}$ & $d_{1}$ & 34 \\\hline
\end{tabular}
\label{fig:clusters}}
%\hspace{0.03\linewidth}
\subfloat[{\scriptsize}]{\raisebox{-7ex}
{
\includegraphics[width=0.50\linewidth]{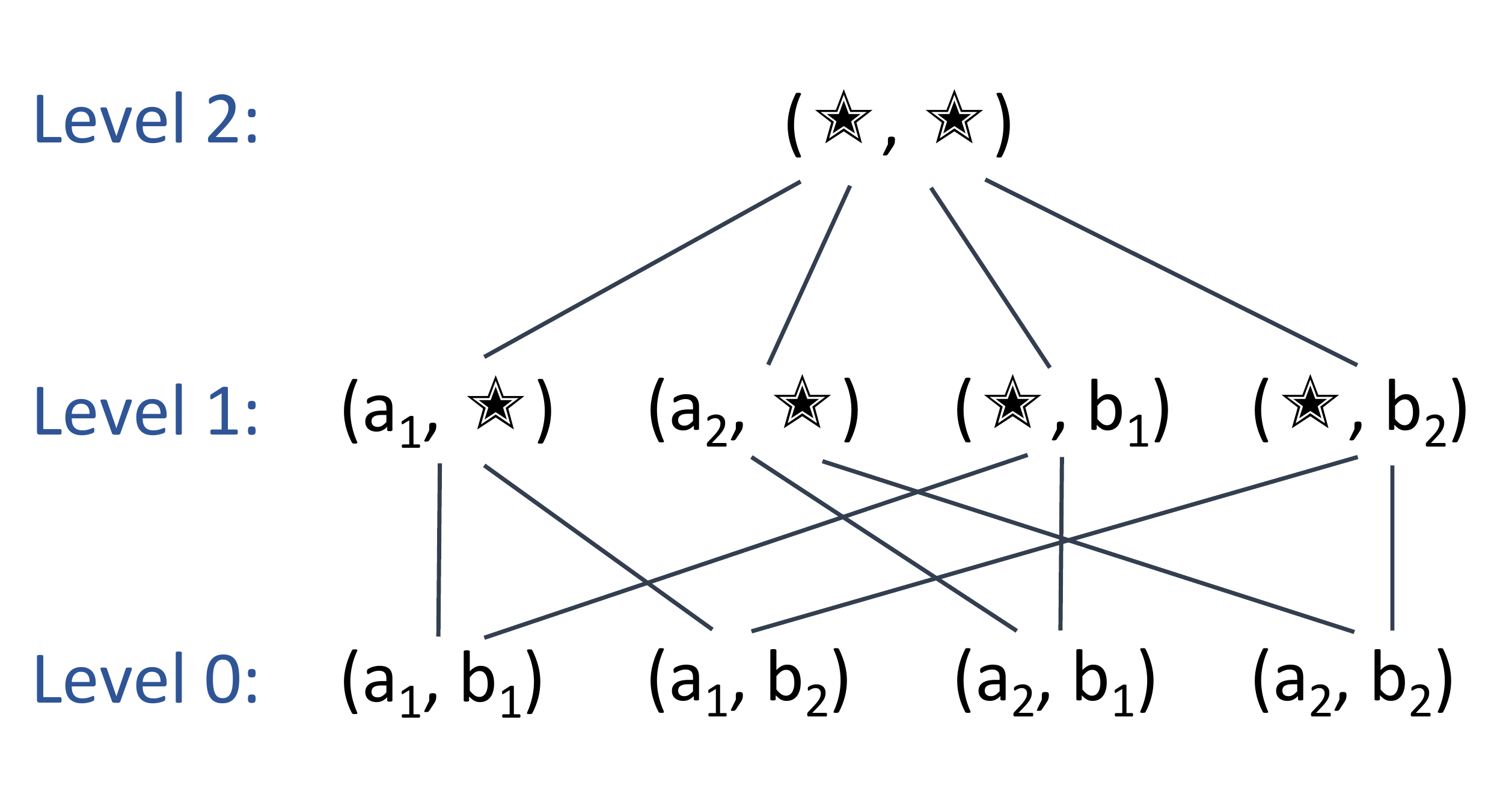}}
\label{fig:semilattice}}
%\end{minipage}\hfill
\vspace{-3mm}
\caption{(a) Example clusters, and (b) semilattice on clusters.}
\vspace{-1mm}
\label{fig:cluster-lattice}
\end{figure}

\textbf{Distance function.~~}
While the distance between two elements is straightforward (the number of attributes where their values differ\footnote{\reva{In this paper we focus on categorical attributes; other distance functions suitable for numeric attributes is a direction for future work (Section~\ref{sec:conclusions}).}}, the distance between two clusters has several alternatives due to the presence of the don't care ($*$) values. We define the distance between two clusters as the number of attributes where they do not have the same value from the domain. The distance function can be shown to be a metric and it exhibits monotonicity property (discussed in Section~\ref{sec:properties}) that we use in our algorithms. 
% two clusters, which has the additional useful property of being monotone as we move along an edge in the cluster semi-lattice. 

\begin{definition}\label{def:distance}
%We define the distance between elements and the distance between clusters as following:
%\begin{itemize}
%\item
The distance $d(t, t')$ between two elements $t, t'$ is the number of attributes where their values differ, i.e., $d(t, t') = |\{i \in [1, m]:  t[A_i] \neq t'[A_i]\}|$.
%\item 
The distance between two clusters $C, C'$ is the 
%maximum possible distance between two elements contained in these two clusters and is measured by the 
number of attributes where either \revb{(i) at least one one of the values is $*$, or (ii) the values are different in $C, C'$}:~
$d(C, C') = |\{i \in [1, m]: C[A_i] = *, {\rm or}, C'[A_i] = *, {\rm or}, C[A_i] \neq C'[A_i]\}|$.
%i.e., 
%$$d(C, C') = \max_{e \in C, e' \in C'} d(e, e')$$ 
%Equivalently, $$d(C, C') = |\{A~:~ C[A] \neq C'[A],\ or\ C[A] = *,\ or\ C[A'] = *\}|$$
%\end{itemize}
\end{definition}
In Figure~\ref{fig:clusters}, the distance between $C_1 = (*, *, c_1, d_1)$ and $C_2 = (a_2, b_1, *, d_1)$ is 3 due to the presence of $*$-s in $A_1, A_2, A_3$. 
Intuitively, the distance between two clusters is the maximum possible distance between any two elements that these two clusters may contain, and therefore is measured by counting the number of attributes where they do not agree on a value from the domain. The distance function can also be explained in terms of similarity measures between two tuples or clusters:
if the distance between two clusters is $\geq D$, then the number of common attribute values between them is $\leq m - D$ where $m$ is the total number of attributes.

\section{Framework}\label{sec:properties}

\cut{
The incomparability condition and the distance function in Definition~\ref{def:problem} %that ensures that none of the clusters displayed exhibit redundant information, 
can be analyzed by viewing the clusters as a natural semilattice. This semilattice structure provides additional insights into the optimization problems in Definition~\ref{def:problem}, and 
%as we will discuss later, helps design exact algorithms in certain special cases as well as 
helps design efficient heuristics in Section~\ref{sec:algorithms}.  
}

In this section, we discuss the technical details for our framework: Section~\ref{sec:optimization} formally defines the optimization problem, Section~\ref{sec:semilattice} discusses the semilattice structure and properties of the clusters, and Section~\ref{sec:complexity} discusses the complexity of the optimization problem.

\subsection{Optimization Problem Definition}\label{sec:optimization}
%\textbf{Optimization problem.~~}
%Given input parameters $k, L, D$, next we define a feasible solution and the optimization problem. 
For a cluster $C$, let $\avg(C)$ denote the average value of all the elements contained in $C$, \ie, $\avg(C) = \frac{\sum_{t \in \cov(C)} \val(t)}{\ansc{|\cov(C)|}}$. Similarly, for a set of clusters $\soln$, $\avg(\soln)$ denotes the average value of the elements covered by $\soln$. 

\begin{definition}\label{def:problem}
\vspace{-1.5mm}
%\textbf{\em Problem Definition:~} 
Given relation $S$ with original tuples and their values, size constraint $k$, coverage constraint $L$, distance constraint $D$, and set $\allclusters$ of possible clusters for $S$, a subset $\soln \subseteq \allclusters$ is called a \emph{feasible solution} if all the following conditions hold:
%\begin{enumerate}
%\itemsep0em
(1) \textbf{(Size $k$)} The number of clusters in $\soln$ is at most $k$, \ie, $|\soln| \leq k$.
(2) \textbf{(Coverage $L$)} $\soln$ covers all top-$L$ elements in $S$, \ie, $S_L^* \subseteq \cov(\soln)$.
(3) \textbf{(Distance $D$)} The distance between any two clusters $C_1, C_2$ in $\soln$ is at least $D$, \ie, $d(C_1, C_2) \geq D$.
(4) \textbf{(Incomparability)} \ansb{No clusters in $\soln$ cover any other cluster in $\soln$} (equivalently, the clusters should form an \emph{antichain} in the semilattice discussed in Section~\ref{sec:semilattice}).
%\end{enumerate}
The objective (called \maxval) is to find a feasible solution $\soln$ with maximum average value $\avg(\soln)$.
% where  the value of a set of clusters $\soln$ is the average value of the elements covered by $\soln$: 
%$\val(\soln) = \frac{\sum_{t \in \cov(\soln)} \val(t)}{|\cov(\soln)|}$.
%\item \textbf{(Non-redundancy)} Any cluster in $\soln$ covers at least one element from the top-$L$ original elements $S_L^*$.
\vspace{-1.5mm}
\end{definition}

The first three conditions in the above definition correspond to the input parameters, whereas the last condition eliminates unnecessary information from the returned solution. All these three parameters, $k, D, $ and $L$, are optional and can have a default value; \eg, the default value of $k$ can be $n$, if there is no constraint on the maximum number of clusters that can be shown. If maintaining diversity in the answer set is not of interest, then $D$ can be set to 0. Similarly, if coverage is not of interest, $L$ can be set to 0 (to display a set of clusters with high overall value), or 1 (to cover the element with the highest value in $S$), or to $k$ (to cover the original top-$k$ elements from $S$). To maintain all the constraints, the chosen clusters may pick up some \emph{redundant elements} $t \notin S_L*$ that do not belong to the top-$L$ elements.
\par
The optimization objective, called \maxval, intuitively highlights the important attribute-value pairs across all tuples with high values in $S$, even if they are outside the top-$L$ elements.\footnote{ We also investigated an alternative objective called \minsize\ that minimizes the number of redundant elements. %and work better if the user intends to inspect fewer elements 
%(details in \cite{fullversion}). 
However it may miss some interesting global properties covering many high-valued elements in $S$, and is less useful for summarization.} % high-valued aggregate answers.}. 
%In this process the output clusters may cover additional redundant elements.
In any solution, the value of each covered element contributes only once to the objective function, hence the selected clusters in $\soln$ do not get any benefit by covering the elements with high value multiple times. In fact, the optimal solution when $D = 0$ and $k \geq L$ %for both objectives 
is obtained by selecting top-$k$ original elements. The optimal solution considers the average value instead of their sum since otherwise, always the \emph{trivial solution} $(*, *, \cdots, *)$ %(discussed below) 
covering all elements %in $S$ 
and satisfying all constraints will be chosen. 

\subsection{Semilattice on Clusters and Properties}\label{sec:semilattice}
A \emph{partially ordered set} (or, \emph{poset}) is a binary relation $\leq$ over a set of elements that is reflexive ($a \leq a$), antisymmetric ($a \leq b, b \leq a \Rightarrow a = b$), and transitive ($a \leq b, b \leq c \Rightarrow a \leq c$). 	A poset is a \emph{semilattice} if it has a \emph{join or least upper bound} for any non-empty finite subset. The coverage of elements described in Section~\ref{sec:prelim} naturally induces a semilattice structure
%\footnote{It is not a lattice since the clusters do not have a \emph{greatest lower bound or meet}.} 
on our clusters $\allclusters$, where for any two clusters $C, C' \in \allclusters$, $C \leq C'$ if and only if $C'$ covers $C$, \ie, $\cov[C] \subseteq \cov[C']$. If $C \leq C'$, then $C'$ is called an \emph{ancestor} of $C$ in the semilattice, and $C$ is a \emph{descendant} of $C'$.
%\ subsection{Semilattice Structure}
%(say that we want to exploit the lattice structure.briefly define the lattice structure and explain Figure 3. 2-4 lines will be fine)
%There is an intuitive semilattice structure based on the partial order of the coverage relationship between clusters. 
Equivalently, if a cluster $C_{up}$ covers another cluster $C_{down}$ by replacing exactly one attribute value of $C_{down}$ by the don't care value ($*$), then we draw an edge between them, and put $C_{up}$ at one level higher than $C_{down}$ in the semilattice (this gives a \emph{transitive reduction} of the poset). % \cite{transitive-reduction}). 
Level $\ell$ of the semilattice  is the set of clusters with exactly $\ell$ $*$ values. 
\cut{
Since there are $m$ attributes in the relation $S$, there are $m+1$ levels in the semilattice, the top-most level has $m$ $*$-s ($m$-th level), \ie, the trivial solution, whereas the bottom-most level has 0 $*$-s (0-th level), \ie, the singleton clusters containing the original elements. }
%We call the clusters at level $\ell$ of the semi-lattice as $\allclusters_\ell$, $\ell \in [0, m]$.
Figure \ref{fig:semilattice} shows the semilattice structure of $\allclusters$ that has two attributes $A_1$ and $A_2$, where the domains are $D_1 = \{a_1, a_2\}$ and $D_2 = \{b_1, b_2\}$.
\cut{
\begin{figure}[t]
\centering
\includegraphics[scale=0.2]{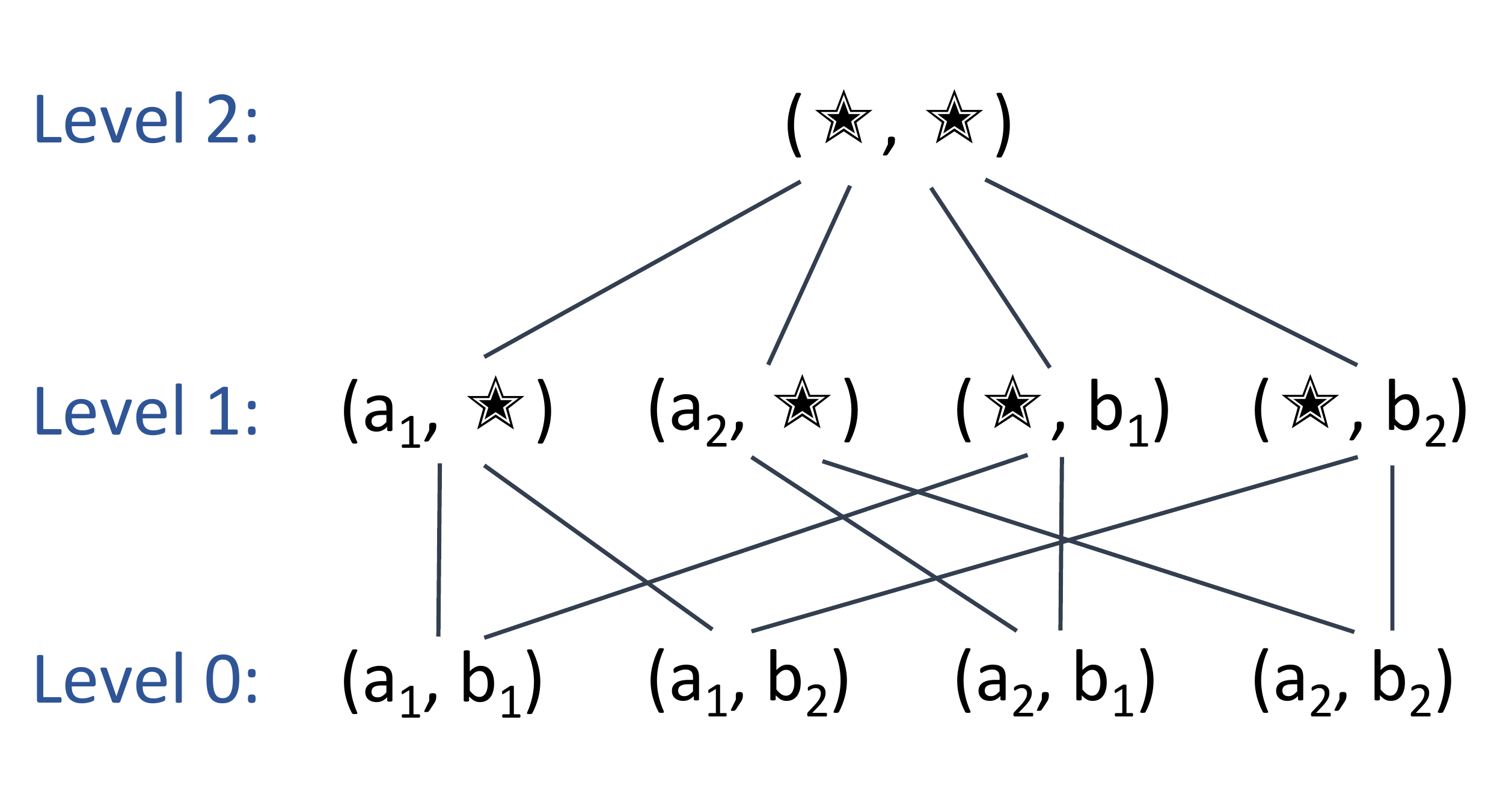}
\vspace{-3ex}
\caption{Semilattice on clusters; $m = 2$.}
\vspace{-0.5cm}
\label{fig:semilattice}
\end{figure}
}
The distance function described in Section~\ref{sec:prelim} has a nice monotonicity property 
%with respect to the partial orders in the lattice 
that we use in devising our algorithms in Section~\ref{sec:algorithms} (proof is in Appendix~\ref{app:prop:distance:mono}):
%Appendix~\ref{app:prop:distance:mono}): 
%In particular, the distance between two clusters $C, C'$ does not decrease if one of them, say $C'$, is replaced by another cluster $C^{''}$ covering $C'$ (although this may violate the incomparability condition) (\red{proof is in \cite{fullversion}})

\begin{proposition}\label{prop: distance:mono}
\vspace{-1.5mm}
{\bf (Monotonicity)~} Let $\soln$ be a set of clusters. % (that may also include elements). 
Let $\lambda$ be the minimum distance between any two clusters in $\soln$ as defined in Definition~\ref{def:distance}, i.e. 
%\begin{displaymath}
$\lambda = \min_{C, C' \in SC} d(C, C')$.
%\end{displaymath}
Let $SC' = (SC \setminus \{C_1\}) \cup \{C_2\}$, where a cluster $C_1$ is replaced by another cluster $C_2$ such that $C_2$ covers $C_1$ (\ie,  $C_2$ is an ancestor of $C_1$ in the semilattice). Let $\lambda^{'}$ be the minimum distance in $SC^{'}$, \ie,
%\begin{displaymath}
$\lambda^{'} = \min_{C, C^{'} \in SC^{'}} d(C, C^{'})$.
%\end{displaymath}
Then
%\begin{displaymath}
$\lambda^{'} \geq \lambda$. %(proof in full version \cite{fullversion}). %Appendix~\ref{app:prop:distance:mono}).
%\end{displaymath}
\vspace{-1.5mm}
\end{proposition} 
\cut{
\begin{proof} \red{move to appendix}
The proof is similar to that of Proposition~\ref{prop:metric}. Consider any cluster $C \in SC$ other than $C_1$. Suppose $d(C_1, C) = \ell$, \ie, $\ell$ of the attributes contribute to the distance function. Fix such an attribute $A$. There are three possibilities. {\bf (1)} $C_1[A] = C[A] = *$. If $C_2[A] = *$, it contributes 1 to $d(C_2, C)$. If $C_2[A] \neq *$, \ie, an attribute value, then also it contributes 1 to $d(C_2, C)$.~~ {\bf (2)} $C_1[A] \neq *$, $C[A] \neq *$, and $C_1[A] \neq C[A]$. If $C_2[A] = *$, it contributes 1 to  $d(C_2, C)$. If $C_2[A] \neq *$, \ie, an attribute value, then it must be the same as $C_1[A]$, and therefore contributes to $d(C_2, C)$.~~ {\bf (3)} $C_1[A] = *$ and $C_1[A] \neq *$. Then $C_2[A] = C_1[A] = *$ and  it contributes 1 to  $d(C_2, C)$. ~~ {\bf (3)} $C_1[A] \neq *$ and $C[A] = *$. Then either $C_2[A] = C_1[A]$ or, $C_2[A] = *$. In either case, it contributes 1 to $d(C_2, C)$. Summing over all $A$ and considering all $C$, $\lambda' \geq \lambda$.
\end{proof}
}

Assuming the semilattice structure in Figure \ref{fig:semilattice}, note that $\{(a_1, b_2), (*,b_1)\}$ satisfies the distance constraint for $D = 2$. If we replace $(a_1, b_1)$ by one of its ancestors $(a_1, *)$, the new two clusters $\{(a_1, *), (*,b_1)\}$ also satisfies the constraint for $D = 2$.

\subsection{Complexity Analysis}\label{sec:complexity}
%\red{need to revise this section later - now focus on the heuristic and experiments}
%In this section, we discuss the complexity of the optimization problem. 
The optimization problem can be solved in polynomial time in \emph{data complexity} \cite{Vardi1982} if the size limit $k$ is a constant. This is because we can iterate over all possible subsets of the clusters of size at most $k$, check if they form a feasible solution, and then return the one with the maximum average value. %(for \maxval) or with the minimum number of redundant elements (for \minsize). 
However, this does not give us an efficient algorithm to meet our goal of interactive performance. For example, if $k=10$, the domain size of each attribute is 9, and the number of attributes is 4, the number of clusters (say $N$) can be $10^{4}$, and the number of subsets will be of the order of $N^k = 10^{40}$. 
%Therefore, our goal is to find efficient algorithms that run in polynomial time in $k$ as well. 
%In this section, we treat $k$ as an input parameter. In our experiments, we will compare our algorithms with the brute force method. %% (given in Algorithm \ref{instance of problem}). 
\par
When $k$ is variable, the complexity of the problem may arise due to any of the four factors in Definition~\ref{def:problem}: the size constraint $k$, the coverage parameter $L$, the distance parameter $D$, and the incomparability requirement that the output clusters should form an antichain. Due to multiple constraints, it is not surprising that in general, even checking if there is a non-trivial feasible solution is NP-hard. 
In particular, when $k \leq L$, simply the  requirement of covering $L$ original elements by $k$ clusters in a feasible solution lead to NP-hardness without any other constraints. However, in the case when $k \geq L$ (the user is willing to see $L$ clusters),
% in the solution), 
the decision and optimization problems become relatively easier. 

In Appendix~\ref{sec:app-nphard} we show the following: % as summarized below:
\cut{
In this section, we show the following for both \maxval\ and \minsize\ objectives. (note that $D = 0$ implies that there is no diversity constraint; the maximum value of $D$ can be $m-1$ according to Definition~\ref{def:distance}); all proofs appear in the full version \cite{fullversion}.}
%%
%\begin{itemize} %[leftmargin=*]
%\itemsep0em
%\item 
(1) $k \geq L$, $D = 0$:~ Top-$k$ elements give the optimal solution, %(Proposition~\ref{prop:top-k-no-D} and 
%(proof in full version \cite{fullversion}), %Appendix~\ref{app:prop:top-k-no-D}), 
since adding any redundant element worsens the \maxval\ objective. %and the \minsize\ objectives.
 %(Case I, Section~\ref{sec:k-geq-L-D-0}).
(2) $k \geq L$, arbitrary $D$:~ A non-trivial feasible solution always exists,
%(Proposition~\ref{prop:feasible-exists} and proof in full version \cite{fullversion}): %Appendix~\ref{app:prop:feasible-exists}):
since we can pick arbitrary ancestors of each top-$L$ element from level $D-1$ satisfying all the constraints.  
However, the optimization problems are NP-hard. %(Theorem~\ref{thm:opt-NP-hard}). %(Case II, Section~\ref{sec:k-geq-L-D-arbit}).
(3) $k < L$, $D = 0$:~ Even checking  whether a non-trivial feasible solution exists is NP-hard.
% (Theorem~\ref{thm:decision-nphard}). %(Case III, Section~\ref{sec:k-lt-L-D-0}).
(4) $k < L$, arbitrary $D$: The same hardness as above holds.
%\end{itemize}
%%

\cut{Indeed, the optimization problems, even the decision problem of finding a non-trivial solution for $k < L$, might appear at least as hard as the \emph{set cover} problem (given a set of elements and subsets, find the minimum number of subsets that can cover all elements), which is known to be NP-hard, and therefore NP-hardness results are as expected. However, an arbitrary instance of the set cover problem may not give us a valid instance of the lattice structure required in or problem.}

Although the optimization problem shows similarity with set cover,  for a formal reduction, we need to construct an instance of our problem by creating a set of tuples and ensure that the `sets' in this reduction conform to a semi-lattice structure. To achieve this, we give reductions from the \textbf{tripartite  vertex cover} problem that is known to be NP-hard \cite{LlewellynTT93}, and construct instances $S$ with only $m  = 3$ attributes. The NP-hardness proof the optimization problem for $k \geq L$ is more involved than the NP-hardness proof for the decision problem for $k < L$, since in the former case the coverage constraint with $k \geq L$ does not lead to the hardness. 
\renewcommand{\algorithmicrequire}{\textbf{Input:}} 
\renewcommand{\algorithmicensure}{\textbf{Output:}}
\section{Algorithms}\label{sec:algorithms}
Given that the optimization problem for the case $k \geq L$, and even the decision problem for the case $k < L$, are NP-hard, we design efficient heuristics that are implemented in our prototype and are evaluated by experiments later. Not only finding provably optimal solutions for our objectives is computationally hard, but designing efficient heuristics for these optimization problems is also non-trivial. The optimization problem in Definition~\ref{def:problem}
%as well as the decision problem of checking whether a feasible solution exists (Definition~\ref{def:problem}) 
has four orthogonal objectives for feasibility: incomparability, size constraint $k$, distance constraint $D$, coverage constraint $L$. In addition, the chosen clusters should have high quality in terms of their overall average value.
% (\maxval) or the number of redundant elements covered (\minsize). 
In Section~\ref{sec:bottom-up}, we discuss the \bottomup\ algorithm that starts with $L$ singleton clusters satisfying the coverage constraint, and merges clusters greedily when they violate the distance, incomparability, or the size constraints. 
%We also discuss two variations of this algorithm. 
Then in Section~\ref{sec:fixedorder}, we discuss an alternative to \bottomup\ that we call the \fixedorder\ algorithm that builds a feasible solution incrementally considering each of the top-$L$ elements one by one. In general, \bottomup\ gives better quality solution \revc{and as discussed in Section~\ref{sec:guidance}, is amenable to processing of multiple parameter settings as precomputation}, whereas \fixedorder\ is more efficient, hence in Section~\ref{sec:hybrid} we describe a \hybrid\ algorithm combining these two.

% in Section~\ref{sec:level-algo}, and then discuss several optimizations in Section~\ref{sec:optimizations}, 

\subsection{The Bottom-Up Greedy Algorithm}\label{sec:bottom-up}
%
%Given L top groups to cover with k clusters that are pairwise at least d apart and no one cluster covers another.  Minimize V(C).
Here we start with $L$ singleton clusters with the top-$L$ elements as our current solution $\soln$, which satisfies the coverage and incomparability constraints, but may violate size and distance constraints. Then we iteratively merge clusters in two phases: the \emph{first phase} ensures that no two clusters in  $\soln$ are within distance $D$ of each other, the \emph{second phase} ensures that the number of clusters is $k$ or less. The following invariants are %consistently
 maintained by the algorithm at all time steps: (1) {\em (Coverage)} Clusters in $\soln$ cover the top-$L$ answers. (2) {\em (Incomparability)} No cluster in $\soln$ covers another. (3) {\em (Distance)} The minimum distance among the pairs of clusters in $\soln$ never decreases. 
During the execution of the algorithm, the only operation is \emph{merging of clusters}, therefore, the coverage invariant above is always maintained. Further, the \merge\ procedure described below maintains the incomparability invariant.
\par
\textbf{The $\merge(\soln, C_1, C_2)$ procedure.~~} Given two clusters $C_1, C_2 \in \soln$, the $\merge(\soln, C_1, C_2)$ procedure replaces $C_1, C_2$ by a new cluster $C_{new} = \mathrm{LCA}(C_1, C_2)$, their {\em least common ancestor}, and also removes any other cluster in $\soln$ that is also covered by $C_{new}$.  $\mathrm{LCA}(C_1, C_2)$ is computed simply by replacing by $*$ any attribute whose values in $C_1, C_2$ differ.   For instance, the LCA of $(a_1, *, c_1, *)$ and $(a_1, b_2, c_2, *)$ is $(a_1, *, *, *)$. Further, if another cluster $(a_1, b_3, *, *)$ belongs to $\soln$, $\merge$ would also remove this cluster, since it is covered by $(a_1, *, *, *)$. 
\par
In addition to maintaining the coverage condition, the merging process does not add any new violations to the distance condition in $\soln$. This follows from the monotonicity of the distance condition given in Proposition~\ref{prop: distance:mono}. However, due to the merging process, the value of the solution may decrease, % that in the bottom-up algorithm, 
%the solution $\soln$ is monotone, \ie, 
since $\mathrm{LCA}(C_1, C_2)$ covers all the elements covered by $C_1, C_2$ and all other clusters that are removed from $\soln$, and can potentially cover some more.

\par
The bottom-up algorithm is given in Algorithm~\ref{algo:bottomup}. The $\updatesolution(\soln, P)$ procedure used in this algorithm  takes the current solution $\soln$ and a set of pairs of clusters $P$ to be considered for merging, and greedily merges a pair.
% as shown in Algorithm~\ref{algo:mergeprocedure}. 
The first and second phases of Algorithm~\ref{algo:bottomup} are very similar, the only difference being the pairs of clusters $P$ they consider for merging. In the first phase, only the pairs with distance $< D$ are considered, whereas in the second phase, all pairs of clusters in $\soln$ are considered for merging.

 \begin{algorithm}[t] %[!ht]
 {\small
\begin{algorithmic}[1]
\Require Size, coverage, and distance constraints $k, L, D$
%\Ensure Result set $\soln \subseteq \allclusters$, $\left|\soln\right|\leq k$
\State {$\soln$ = set of $L$ singleton clusters with the top-$L$ elements.}
\State {\emph{/* First phase to enforce distance */}}
\While{$\soln$ has two clusters with distance $< D$}
\State{Let $P_D$ be the pairs of clusters in $\soln$ \emph{at distance $< D$}.}
\State{Perform $\updatesolution(\soln, P_D)$.}
\EndWhile
\State {\emph{/* Second phase to enforce size limit $k$, almost the same as above except all pairs of clusters are considered. */}}
\While{$|soln| > k$}
\State{Let $P_{all}$ be the \emph{all} pairs of clusters in $\soln$.}
\State{Perform $\updatesolution(\soln, P_{all})$.}
\EndWhile
\State\Return $\soln$\\
\vspace{2mm}
\textbf{Procedure $\updatesolution(\soln, P)$}
\State{\textbf{Input}: current solution $\soln$, a set of pairs $P$ of clusters}
\State{$(C_1, C_2) = {\tt argmax}_{(C_1, C_2) \in P} \avg(\soln \cup \mathrm{LCA}(C_1, C_2))$}
\State{Perform $\merge(\soln, C_1, C_2)$.}
\caption{The \bottomup\ algorithm}
\label{algo:bottomup}
\end{algorithmic}
}
\end{algorithm}

 \cut{
\begin{algorithm}[t] %[!ht]
{\small
\begin{algorithmic}[1]
\Require Current set of clusters $\soln$, and a set $P$ of pairs of clusters to be considered for merging.
%\State{\emph{/* Below we state two variants of the greedy step, only one of them will be executed depending on the given algorithm. */}}
%\State{\emph{/*$\mathrm{LCA}(C_1, C_2)$ denote the smallest cluster containing both.*/}}
%\emph{/* Choose the pair that gives highest value of the new solution after merging */}
\State{$(C_1, C_2) = {\tt argmax}_{(C_1, C_2) \in P} \avg(\soln \cup \mathrm{LCA}(C_1, C_2))$}
%\If{the algorithm is \mergeindiv}
%\State{$(C_1, C_2) = {\tt argmax}_{(C_1, C_2) \in P} \avg(\mathrm{LCA}(C_1, C_2))$ \emph{/* the new cluster has the highest average value */}}
%\Else\emph{/* the algorithm is \mergeall */}
%\State{$(C_1, C_2) = {\tt argmax}_{(C_1, C_2) \in P} \avg(\soln \cup \mathrm{LCA}(C_1, C_2))$ \emph{/* the new solution after including $\mathrm{LCA}(C_1, C_2)$ has the highest average value */}}
%\EndIf
\State{Perform $\merge(\soln, C_1, C_2)$.}
\caption{$\updatesolution(\soln, P)$  for \maxval.}
\label{algo:mergeprocedure}
\end{algorithmic}
}
 \end{algorithm}
 }

\par
We also implemented and evaluated other variants of bottom-up algorithms: (i) when we start at the clusters at level $D-1$ (instead of individual top-$L$ tuples that satisfy the distance constraint), and (ii) when we greedily merge pairs $C_1, C_2$ with maximum value of $\avg(\mathrm{LCA}(C_1C_2))$ (instead of maximum average value of the overall solution after merging). Both these variants had efficiency and quality comparable or worse than the basic \bottomup\ algorithm as observed in our experiments.

\subsection{The Fixed-Order Greedy Algorithm}\label{sec:fixedorder}

\cut{The \fixedorder\ algorithm (Algorithm~\ref{algo:fixedorder}) maintains a set of clusters $\soln$, considers top-$L$ elements in order, }
\red{The \fixedorder\ algorithm maintains a set of clusters $\soln$, and considers top-$L$ elements in descending order by value.}
It decides whether the next element is already covered by an existing cluster in $\soln$ or can be added as is (satisfying $D$ and $k$ constraints); otherwise \red{\fixedorder\ merges the element} with one of the existing clusters in greedy fashion. All the constraints ($k, D,$ and incomparability of clusters) are maintained after each of the top-$L$ is processed, so at the end the coverage on top-$L$ is satisfied too. \fixedorder\ considers a smaller solution space than \bottomup, since it processes each top-$L$ element in an online fashion, and therefore may return a solution with worse value. However, instead of all pairs of initial clusters (quadratic in number of clusters) it considers each cluster only once (linear), so has better running time than \bottomup. %Pseudocode for \fixedorder\ is shown as Algorithm~\ref{algo:fixedorder} in Appendix (Section~\ref{sec:app-fixedorder}).
Details and pseudocode for \fixedorder\ 
%(and experimental comparison with a variant called \randomorder\ where top $L$ elements are considered in a random order) 
are shown in Appendix~\ref{sec:app-fixedorder}.

\reva{We also consider two variants of \fixedorder\ and evaluate them
  later in experiments: i)~\emph{k-means}-\fixedorder, where we first
  run the $k$-means clustering algorithm~\cite{hartigan1975clustering}
  (with random seeding) on the top $L$ elements, find the minimum
  pattern covering all elements in each of the resulting clusters, and
  make \fixedorder\ process these $k$ patterns first before moving on
  to the top $L$ elements (in descending-value order);
  ii)~\emph{random}-\fixedorder, where we first pick $k$ element at
  random from the top $L$ elements to process first, before moving on
  to the remaining top $L$ elements (still in descending-value order).
  Both variants introduce some randomness in the results, and
  \emph{k-means}-\fixedorder\ has considerable higher initial
  processing overhead.  However, as we shall see in
  Section~\ref{sec:experiments}, they do not produce higher-quality
  results.}

\subsection{The Hybrid Greedy Algorithm}\label{sec:hybrid}
%\bottomup\ uses a relatively cautious strategy to decrease the candidate clusters which tends to have a better approximation ratio. In contrast, \fixedorder\ algorithm keeps a maximum number of candidate clusters at all time which will limit the increase of calculation time when the dataset scales up. 

%In order to combine the advantages of both \bottomup\ (better quality) and \fixedorder\ (better efficiency),
%Both \bottomup\ and \fixedorder\ have deficiencies. 
\bottomup\ tends to produce results with higher quality than \fixedorder, and can process multiple $k, D$ values at the same time as discussed in Section~\ref{sec:guidance}, but usually requires more iterations than \fixedorder. In order to get a good trade-off between these factors,
we introduce the hybrid algorithm. It has two  phases - the \fixedorder\ phase and \bottomup\ phase. For a given $k, L,$ and $D$, the first phase for \hybrid\ is the same as \fixedorder, but with a larger number of \ansc{$c \times k$}, $c > 1$ is a constant, initial singleton clusters. After covering all top-$L$ elements in \ansc{$c \times k$} clusters, \hybrid\ goes into the \bottomup\ phase to reduce the number of clusters from \ansc{$c \times k$} to $k$ using the $\merge$ procedure that can collect redundant elements. Like \bottomup, \hybrid\ also helps in incremental computation for different choices of parameters as discussed in the next section.

\section{Interactive Parameter Selection}\label{sec:guidance}

%\subsection{Motivation}
One of the main challenges in a system with multiple input parameters is choosing the input parameter combination carefully to help the user explore new interesting scenarios in the answer space. 
%Our framework takes three parameters: size limit on the clusters $k$, distance lower bound $D$ among solutions, and the coverage parameter $L$ so that the clusters cover at least the top-$L$ original answer tuples. 
To help the user choose interesting values of $k, L, D$, we provide an overall view of the values of the solutions (average value of all element covered by the clusters chosen by our algorithm) that at the same time precomputes the results for certain parameter combinations and helps in interactive exploration. In Section~\ref{sec:guid-visualization} we describe the visualization facilitating parameter selection, in Section~\ref{sec:precomputation} we discuss how the precomputation is achieved to plot these graphs,
% and how does it help in efficiently retrieving the answers. I
and in Section~\ref{sec:guid-opt} we discuss a number of optimizations for interactive performance of our approach.
% computation of the answers for multiple combinations of input parameters efficient. 
    
\cut{
Sometimes users may not be familiar with the database they are trying to explore. In this scenario, it is easy for them to get lost. A clean figure containing some precomputed data points can be greatly helpful and are capable of guiding users to their destinations. Figure~\ref{fig:guid-vizdetail} gives an example of our implementation to be introduced in this section. At first only the trends (lines) are shown . If users would like to explore more details, they can hover over the visualization and details will be shown as in Figure~\ref{fig:guid-vizdetail}. As a result, users have a clear view what is the trend varying different parameters, and avoid regions that they are not interested in. Besides, with affordable extent of precomputation, if users get interested in parameter combinations that have been precomputed, the system can instantly respond to the request rather than run from scratch. This section will introduce all aspects about the guidance visualization, including the visualization, implementation and some optimizations that have been applied.
}

\subsection{Visual Guide for Parameter Selection}\label{sec:guid-visualization}
\ansc{Figure~\ref{fig:intro-guid} gives an example of visualization showing the overview of the solutions that is generated for each chosen value of $L$, and illustrates the values of the solutions for a range of choices on $D$ and $k$.  
%Once the user issues a certain input SQL query and coverage parameter $L$, the system will immediately pass the query and the specified parameter to the backend and return the line plot. 
The $y$-axis shows the average value of the tuples covered by the chosen clusters by our algorithm (Definition~\ref{def:problem}), the value of $k$ (in a chosen range) varies along the $x$-axis, and different lines correspond to different values of $D$ (also in a chosen range).
%$$For a fixed $L$, the $x$-axis is $k$. $D$ varies among different lines. 
%The right end of each line represents the position where the distance enforcement of the \bottomup\ phase finishes. Applying any $k$ larger than that position for the given $D$ will share the same result as the right end. Changing $L$ value will result in a different plot. COMMENTED:NO LONGER TRUE!
\par
%This plot provides the user with a glimpse of the big picture of the result clusters in terms of the average value for different combinations of  combinations. 
With the help of this visualization, the user can avoid selecting certain \emph{uninteresting} or \emph{redundant} parameter combinations.
% where result clusters have most attributes as stars or $k$ is too loose (large) to make a difference in output. 
\cut{
For example, the bottom-left region in Figure~\ref{fig:guid-vizdetail} can be considered as uninteresting with a low average value (only the trivial solution with all $"*"$ belongs to the solution. Also the user can avoid changing $k$ or $D$ that lead to the same value of the solution (\eg, for $D = 2$, the range of $k = [20, 30]$ gives almost the same values), which can make the exploration easy and more efficient for the user. }
For example, with the visualization in Figure~\ref{fig:intro-guid}, a user can quickly see that the bottom-left region (where $k=2,3$) is uninteresting, with low average values. The user also sees that certain ranges of parameter settings are not worth exploring as they do not affect the solution quality or in very predictable ways: \eg, for $D = 1$, the range of $k > 12$ yields almost the same solution quality, while for $k \in [2,9]$, the quality changes predictably with $k$.  On the other hand, the ``knee points'' (\eg, $k=9,11$ for $D=1$) suggest good choices of parameters.  The visualization also reveals the trade-off between different choices of $D$; \eg, at $k=9$, the user can decide between a solution set with a higher value ($D=1$) or more diversity ($D=2$).  Note that in Figure~\ref{fig:intro-guid}, curves for different $D$ values may overlap, which suggests ranges of $D$ values with little impact on solution quality, allowing the user to work on ``bundles'' of $D$ values instead of individually.  (If the user cares about curves for individual $D$ values, the user can click on the legend on the right to hide particular curves to reveal others that overlap.)
}

%The generated line plot is also helpful when the user targets on some patterns upon changing attributes. It saves efforts for the user from testing different attribute combinations excessively.

%\subsection{Implementation}

\subsection{Incremental Computation and Storage}\label{sec:precomputation}
To be able to generate plots in Figure~\ref{fig:intro-guid}, one obvious approach is running an algorithm from Section~\ref{sec:algorithms} for all combinations of $k$ and $D$ given an $L$ value. However, for interactive exploration, this approach is sub-optimal. The \hybrid\ algorithm (and \bottomup) exhibits \emph{two levels of} incremental properties that help in computing the solutions for a range of $k, D$ values in a batch. %without redoing the computations from the beginning.
%In order to generate the guidance visualization, precomputation of the object value for different combinations of $k, L$ and $D$ is essential. Given that expensive precomputation may affect interactivity, an effective way to precompute necessary data points is in need. 
%In \bottomup\ algorithm and the \bottomup\ phase of the \hybrid\ algorithm, efficient precomputation can be achieved by two level of incremental precomputation:
 
%when $L$ is fixed, $k, D$ can be incrementally computed; when $L,D$ are fixed, $k$ can be incrementally computed. 
 In \hybrid, for a given value of $L$, the \fixedorder\ phase outputs a set of initial clusters that can be used for all combinations of $k, D$, and therefore, this step can run only once. %All parameters and lists are recorded when the \fixedorder\ phase finishes. 
Remembering this intermediate solution,  
 the \bottomup\ phase can run for all $D$ values from the stored status. For each $D$, it computes results for all $k$ values (ranging from the maximum to the minimum value) since in every round of iteration, two clusters are merged to reduce the number of clusters by one. The procedure for this incremental computation is shown in Figure~\ref{fig:guid-incremental}. In the following, we discuss how we materialize and index solutions for efficient retrieval.
 
 %Results for each $D$ are saved in the data structure that will be introduced in the following subsection. 

%\subsubsection{Retrieval Data Structure}

\textbf{Retrieval Data Structure.~~} 
%It is challenging to find a suitable data structure for effective retrieval. An intuitive way to store the precomputed data is that 
The computed solutions for different $k, D$ values serve as pre-computed solutions when the user wants to inspect the solution in detail for a certain choice of $k, L, D$. The obvious solution for storage is to record the set of output clusters for every choice of $(k,D)$.
%, record every cluster that appears in the solution cluster list. This way is fast for retrieval, but takes more storage room than necessary. 
However, we implemented a combined retrieval data structure for storage that is both space and time efficient based on the following observation in the execution of \hybrid\ (and \bottomup) algorithm: % (the proof is intuitive and is omitted for space constraints):

\begin{proposition}\label{prop:guid:int-consec}
{\bf (Continuity)~} Given solution cluster lists $\soln_1,\soln_2,...,\soln_r$ where $r$ rounds are executed, for any cluster $c \in \soln_a$ where $1\leq a<i$, once $c$ is removed from $\soln_i$ at the end of round $i$ (because of merging), for all $j > i$, $c \notin \soln_j$.
% will not happen for all Round $j$ where $j>i$.
\end{proposition}
In other words, once a cluster is merged and therefore vanishes from the set of clusters in the solution, it never comes back. Hence, if $\soln_{L, D, k}$ denotes the solution for a given combination of $L, D, k$,  the set of values of $k$ for which a cluster $c \in \soln_{L, D, k}$
forms a continuous interval. 
%Since the existence of a cluster in the solution $\soln$ is continuous as stated in Proposition~\ref{prop:guid:int-consec}, 
Therefore, instead of storing the set of clusters for all values of $D, k$ given an $L$ value (where the solutions may have substantial overlap),  we use an \emph{interval tree}\cite{cormen2009introduction} 
$S_D$ for each value of $D$ to store the range of $k$ for which a cluster appears in $\soln_{L, D, k}$ storing only the {\tt maximum} (or starting) and {\tt minimum} (or ending) $k$ value for this cluster (see Figure~\ref{fig:guid-interval}). It reduces the number of solutions (sets of clusters) to be stored from $O(N_k \times N_D)$ (where $N_k$ and $N_D$ denote the total number of $k$ and $D$ values under consideration respectively)  to $O(N_d)$. Further, the interval tree data structure supports efficient retrieval in time $O(\log N_k)$\cite{cormen2009introduction}. .
%the interval structure is chosen for efficiently storing solutions for different $k$ and $D$ values. Figure~\ref{fig:guid-interval} illustrates the interval structure. An interval contains the starting $k$ (maximum possible $k$) and the ending $k$ (minimum possible $k$) that one cluster resides in the answer cluster list under a certain $D$. Intervals for all clusters that have once been in the solution $\soln$ for $D$ are stored in once interval set (denoted by $S_D$). Each $S_D$ can be organized as an interval tree. When a specific $(k,D)$ (denoted by $(k',D')$) is issued, $S_{D'}$ will be chosen for retrieval and all intervals inside $S_{D'}$ that intersect with $k'$ will be retrieved for the solution. This retrieval structure saves efforts from redundant registration of identical clusters for different $k$.

%This data structure starts working at the \bottomup\ phase of \hybrid\ algorithm. That is because in the \fixedorder\ phase, the distance constraint $D$ is not ensured and the realtime $k$ value fluctuates whenever a new atomic cluster comes. In the \fixedorder\ phase, either $k$ or $D$ value cannot be separated for incremental computation.

\begin{figure}[t]
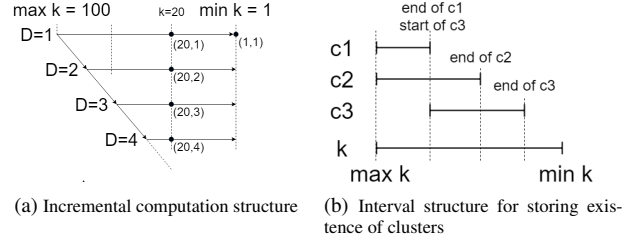

\centering
\begin{minipage}[t]{.47\textwidth}
  \centering
 %\vspace*{\fill}
\subfloat[{\scriptsize Incremental computation structure}]{
\includegraphics[width=0.42\linewidth]{figures/guidIncremental.png}~~~~
\label{fig:guid-incremental}}
\hspace{0.02\linewidth}
\subfloat[{\scriptsize Interval structure for storing existence of clusters}]{
\includegraphics[width=0.42\linewidth]{figures/guid_intervalStructure.png}~~~~
\label{fig:guid-interval}}
\vspace{-0.3cm}
\end{minipage}\hfill
\caption{Incremental computation and interval structure.}
\label{fig:guid-intervalOverall}
\vspace{-0.2cm}
\end{figure}

\subsection{Optimizations} \label{sec:guid-opt}

A number of additional optimizations are implemented to make the system efficient and interactive as described below. 
% several optimizations are designed and implemented. \emph{Delta-Judgment} is a technique to reduce the calculation time for the \maxval\ value. Section~\ref{guid:opt-delta} introduces \emph{Delta-Judgment} in detail. In terms of initialization, optimizations to be discussed in Section~\ref{guid:opt-init} have been applied in cluster generation and cluster-tuple mapping.  An optimization for storing attributes is introduced in Section~\ref{guid:opt-hash}.

% \subsubsection{Delta Judgment} \label{guid:opt-delta}
\textbf{Delta judgment.~~}
\cut{Several variables and concepts are listed here in order to help reminding and understanding following statements and proofs.
\par
(1) \textbf{Round}: Every iteration in the \bottomup\ phase aiming to greedily merge is counted as one round. A round counter is kept globally.
\par
(2) \textbf{Solution}: denoted by $\soln$, same as discussed in Section~\ref{sec:bottom-up}. $\soln_i$ denotes the solution cluster list at the end of Round $i$.
\par
(3) \textbf{Global covered tuple list}: The covered tuple list contains all tuples that are covered by clusters in the solution. It is denoted by $T$ and $T_i$ denotes the global covered tuple list at the end of Round $i$.
\par
(4) \textbf{Cluster tuple list}: denoted by $T_c$, it contains tuples that are covered by a certain cluster.
\par
}
In every iteration (called \emph{round}) of greedy cluster merging in the \hybrid\ (and \bottomup) algorithm, %and the \bottomup\ phase in \hybrid\ always try to merge 
clusters are merged such that the average value  of the clusters in the resulting solution is maximized using the \updatesolution\ function (Algorithm~\ref{algo:bottomup}).%\(ref{algo:mergeprocedure}). 
Let $\soln_{i}$ be the set of clusters at the end of a round $i$, $T_i = \cov(\soln_i)$ be the tuples covered by $\soln_i$, $v_i = \avg(\soln_i)$ be the average value of $\soln_i$, and $T_c = \cov(c)$ be the tuples covered by a given cluster $c$.
The naive way of executing \updatesolution\ in round $i+1$ involves comparing the tuple list $T_c$ of a given cluster $c$ ($=  LCA(C_1,C_2)$ as mentioned in Algorithm~\ref{algo:bottomup}) and the current set of covered tuples $T_i$, %the global covered tuple list, 
finding out new tuples in $T_c \setminus T_i$ to obtain $T_i \cup T_c$ as potential $T_{i+1}$, and recalculating the objective $\avg(T_i \cup T_c)$ based on the new tuples. However, it takes a huge amount of time doing all the tuple-wise comparison for all possible clusters that are eligible to be merged in this round. Instead, we incrementally keep track of the marginal benefit (as sum and count to compute the average) that a cluster $c$ brings to the new solution $\soln_{i+1}$ compared to $\soln_i$ as follows (pseudocode in Algorithm~\ref{algo:deltajudgment}).
\par
The basic idea is that the improvement in the total average value that a cluster $c$ brings to solution $\soln_i$ is due to the tuples in $T_c \setminus \soln_i$, and that it brings to $\soln_{i-1}$ is due to  the tuples in $T_c \setminus \soln_{i-1}$. The difference can be computed by keeping track of the new tuples  that appear in $T_{i} \setminus T_{i-1}$, and comparing them with the tuples in $T_c$. In addition, we incrementally store $\Delta_{i, c, sum}$ and $\Delta_{i, c, count}$ (the sum of values and the count of tuples in $T_c \setminus T_i$, incrementally computed from $\Delta_{i-1, c, sum}, \Delta_{i-1, c, count}$). Hence the tentative new average value of the solution $\soln_{i+1}$ if we add $c$ to $\soln_{i}$ can be computed as $v_{i+1} = \frac{v_i \times |T_i|  + \Delta_{i, c, sum}}{|T_i| +\Delta_{i, c, cnt}|}$.
This optimization evaluates the \updatesolution\ procedure efficiently since the above computations need comparisons between (i) the list containing $T_{i} \setminus T_{i-1}$ and (ii) $T_c$, and  $T_{i} \setminus T_{i-1}$ is likely to be much smaller than $T_{i}$. This gives $~30x$ speedup in our experiments. 
%A pseudocode is given in  Algorithm~\ref{algo:deltajudgment} and t
%\red{The effect of this optimization is evaluated in our experiments, and the pseudocode is given in the  full version \cite{fullversion}.}

%\vspace{-1ex}

\begin{algorithm}[ht]
{\small
\begin{algorithmic}[1]
\Require Marginal score benefit $\Delta_{sum}$, marginal amount benefit $\Delta_{cnt}$, round indicator $i$ indicating when were $\Delta$ values last updated, current round $j+1$, difference list $T_{(j,j-1)} = T_j \setminus T_{j-1}$\
\State{/* $\Delta_{i-1,c,sum} = \Delta_{sum}$ and same for $\Delta_{i-1,c,cnt}$ in this case*/}
%\State{\red{THIS IS NOT STATING ANYTHING EXTRA -- SHOULD REMOVE}}
\State{$v_{j+1}$ = new score to be calculated.}
\State{$v_j =$ current score.}
\State {\emph{/* Marginal benefits are far outdated, -1 is default value; $i<=j-1$ means $\Delta_{i-1,c,sum}$ and $\Delta_{i-1,c,cnt}$ cannot be updated directly using $T_{(j,j-1)}$ */}}
\If{$i = -1$ \textbf{or} $(j-i\geq1)$}
\State{$\Delta_{cnt} = \sum(val(T_j \setminus T_c))$}
\State{$\Delta_{sum} = |T_j \setminus T_c|$}
\State {\emph{/* Marginal benefits were updated last round (round $j$) and can use $T_{(j,j-1)}$ for comparison */}}
\ElsIf{$i=j$}
\State{$\Delta_{cnt} = \sum(val(T_{(j,j-1)} \setminus T_c))$}
\State{$\Delta_{sum} = |T_{(j,j-1)} \setminus T_c|$}
\State {\emph{/* Marginal benefits were updated %before 
in the same round */}}
\ElsIf{$i=j+1$}
\State{\emph{/* Do nothing. It is up-to-date. */}}
\EndIf
\State{$v_{j+1} = (v_j\times|\soln_j| + \Delta_{sum})/(|\soln_j|+\Delta_{cnt}))$}
\State{\emph{/* Update the round indicator */}}
\State{$i=j+1$}
\State\Return $v_{j+1}$
\caption{The \deltajudgment\ Procedure}
\label{algo:deltajudgment}
\end{algorithmic}
}
\vspace{-1mm}
 \end{algorithm}
\textbf{Cluster generation and mapping to tuples.~~}
%Naive implementations could take tens of hundreds of seconds to finish the initialization. 
The semilattice structure on the clusters given an $L$ value is required to run our algorithms  that may contain a number of clusters in a naive implementation. 
To reduce this space to contain only the relevant clusters, clusters are first generated by each tuple in top-$L$, which ensures that each generated cluster is a possible cluster covering at least one tuple in top-$L$. Besides, we need to maintain mappings between clusters and the tuples they contain, for which tuples generate \emph{matching expressions} for their target clusters and search through the cluster list (instead of starting with a cluster and searching for matching tuples). Experiments in Section~\ref{sec:exp-opt} shows the benefit - $100x-1000x$ speedup in running time.

% \subsubsection{Hash Values for Fields}\label{guid:opt-hash}
\textbf{Hash values for fields.}
The value of an attribute is often found to be text (or other non-numeric value).
While storing information on the clusters, we maintain hashmaps for each field between actual values and integer hash values, and store the hash values inside each cluster (mapped back to the original values in the output). %So that the previously used JsonNode structure for value saving can be now changed into integer array. This helps greatly when there are text fields in expresstions of clusters which is the common case. Once the retrieval of clusters is in need, hash values can be efficiently mapped  to text values by the hashtable built before. 
This optimization reduces the running time of the order of $50x$. % compared with the naive implementation.

\section{Experiments}\label{sec:experiments}
We develop an end-to-end prototype with a graphical user interface (GUI) to help users interact with the solutions returned by our two-layered framework. The prototype is built using Java, Scala, and HTML/CSS/JavaScript as a web application based on Play Framework 2.4, and it uses PostgreSQL at the backend %as the relational DBMS
\revb{%The architecture of the system and snapshots 
(screenshots of the graphical user interface %shown in Figure~\ref{fig:architecture}. 
%described in the full version\cite{fullversion} %Appendix~\ref{sec:architecture}, 
%and a snapshot of the GUI is shown in Figure~\ref{fig:newgui} %in Section~\ref{sec:app-gui}. 
can be found in the demonstration paper for our system %in SIGMOD 2018 
\cite{qagviewdemo2018}).} %the full version \cite{fullversion}
In this section we experimentally evaluate our algorithms using our prototype by varying different parameters (Section~\ref{sec:perf-params}), and then test the precomputation and guidance performance (Section~\ref{sec:exp-precomputation}). The effects of optimizations  are given in Section~\ref{sec:exp-opt}, scalability of our algorithms for a larger dataset is discussed in Section~\ref{sec:exp-tpcds}.
%, and visualizations are evaluated in Section~\ref{sec:exp-viz}.
%Comparison with other related approaches are illustrated in Appendix (Section~\ref{sec:expt-compare-others}). 
\cut{Additional discussion about the element selection order in \fixedorder\ is provided in Section~\ref{sec:app-exprandomorder}.}
%that let users interactively query database and get cluster results. 
% and %The media used in the communication between 
%the front-end and the back-end of the system communicate with JSON. 
%The system uses PostgreSQL as the relational DBMS storing the datasets and running the aggregate queries. The user explores the results of an aggregate query using a single-page interface on a web browser.
\par
\textbf{Datasets.~~} In most of the experiments, we use the MovieLens 100K dataset \cite{movielens, movielensdata, movielenspaper}. We join all the tables in the database (for movie-ratings, users, their occupation, etc) and materialize the universal table as \emph{RatingTable}. %for all the ratings. 
Each tuple in this rating table has 33 attributes %\red{correct?}. 
of three types: (a) \emph{binary} (\eg, whether or not the movie is a comedy or action movie), (b) \emph{numeric} (\eg, age of the user), and %The second type of attributes contains discrete numerical values. The last type is 
(c) \emph{categorical} (\eg, occupation of the user). We join the tables as a precomputation step to avoid any interference while measuring the running time of our algorithms. 

\cut{
The aggregate queries used in this section are of the following form.
% (unless mentioned otherwise, $m = 8$ grouping attributes are used):
\newsavebox\sqlexp
\begin{lrbox}{\sqlexp}\begin{minipage}{\textwidth}
\lstset{language=SQL, basicstyle=\ttfamily, deletekeywords={year,month,action},tabsize=2}
\begin{lstlisting}[mathescape]
SELECT $\langle$ grouping attributes$\rangle$, avg(rating) as $\val$ 
FROM RatingTable 
GROUP BY $\langle$ grouping attributes$\rangle$
HAVING count(*) > 50 
ORDER BY $\val$ DESC
\end{lstlisting}
\end{minipage}\end{lrbox}
\resizebox{0.85\textwidth}{!}{\usebox\sqlexp}
}

The other dataset we use is TPC-DS benchmark \cite{nambiar2006making} primarily for evaluating scalability of our algorithms. The table we materialized via generator is \emph{Store\_Sales}, which contains 23 attributes and 2,880,404 tuples in total. The aggregate queries used for these two datasets (average rating for MovieLens and average net profit for TPC-DS) can be found in Appendix~\ref{app:exp-queries}.
\cut{
The aggregate queries we use for TPC-DS related experiments %are shown as follows.
share the same form with the example given above.

\begin{lrbox}{\sqlexp}\begin{minipage}{\textwidth}
\lstset{language=SQL, basicstyle=\ttfamily, deletekeywords={year,month,action},tabsize=2}
\begin{lstlisting}[mathescape]
SELECT $\langle$ grouping attributes$\rangle$, cast(avg(net_profit) 
        as int) as $\val$ 
FROM store_sales
GROUP BY $\langle$ grouping attributes$\rangle$
HAVING count(*) > 10 
ORDER BY $\val$ DESC
\end{lstlisting}
\end{minipage}\end{lrbox}
\resizebox{0.85\textwidth}{!}{\usebox\sqlexp}
}
All experiments were run on a %virtual machine %locally on a \red{64-bit OSX 10.11.4 with Intel(R) Core(R) i7 (8 GB RAM, 2.6 GHz)}. 
64-bit Ubuntu 14.04.4 LTS machine, with Intel Core i7-2600 CPU (4096 MB RAM, 8-core, 3.40GHz).

\subsection{Varying Parameters}\label{sec:perf-params}
%First, we compare our algorithms with brute-force, and then we compare our algorithms varying one of the input parameters in terms of the value of the solution and running time. 
%We present running time and quality of solution for the objective \maxval; the \minsize\ objective yields similar conclusions and is omitted due to space constraints.
%\par
Unless mentioned otherwise, the three algorithms from Section~\ref{sec:algorithms} are compared in this section: (i) {\bf \bottomup}, (ii) {\bf \fixedorder}, (iii) {\bf \hybrid}.
% two level-based algorithms (\greedyavg, \greedyavg) and two bottom-up algorithms (\mergeindiv, \mergeall). 
In the plots showing the values, we also include (iv) {\bf \tt Lower Bound}: value of the trivial (feasible) solution (a single cluster with don't-care $*$ values for all attributes) as a baseline. % which is always feasible.
\cut{
 and (v) {\bf \tt Upper Bound}: this is the average value of top-$L$ elements, which gives an upper bound on the value obtained by any solution, but \emph{it may be infeasible to achieve this value due to constraints on $k$ and $D$}. This is illustrated in Figure~\ref{fig:brute-force-val} by comparing the upper bound with the optimal (feasible) value returned by the brute-force algorithm which is close to our results, but the upper bound is much higher. The upper bound is only shown to display the range of the values in consideration since brute-force is practically infeasible for larger values of the parameters.
 }
%The {\tt baseline} in all the graphs showing the value exhibits the value of the trivial solution, which is always a feasible solution for any value of $k, L, D$; we also show {\tt upper bound (infeasible)} as an upper bound of the value function, which is the average value of the top-$L$ elements, but it is possible that no feasible solution can attain that upper bound due to constraints on $k$ and $D$.
%\subsubsection{Comparison with brute-force algorithm} \label{sec:brute-force}

\textbf{Comparison with baselines.~}
We compare our algorithms with two baselines: the brute-force algorithm considers all possible cluster combinations, and outputs the global optimal; the lower-bound algorithm simply returns the trivial answer containing one single cluster with all attributes as ``*''s, which is always feasible for any value of $k, L, D$. \ansa{We also consider two variants of \fixedorder: random and $k$-means, discussed in Section~\ref{sec:fixedorder}.}  Figure~\ref{fig:brute-force-time} shows the running time for $L=5, D=3$ and $k=2,3,4$ (lower-bound is omitted because it returns trivial answers).
%During this experiment, for three parameters $k$, $L$, and $D$, we changed one of them and kept the other two fixed. Figure \ref{fig:brute-force}(a) shows the result of changing $k$ from 2 to 4 with $L=5$ and $D=3$. It illustrates that the running time of the brute force algorithm boosts rapidly when $k$ increases. 
Even with such small parameter values, the brute-force algorithm is not practical: \eg, at $k = 4$, it takes more than 2.5 hours. Figure~\ref{fig:brute-force-val} compares the average values produced by different algorithms. \ansa{Since the random and $k$-means variants of \fixedorder\ are randomized, we report their average values over 100 runs each.}  From Figure~\ref{fig:brute-force-val}, we see that the results of \fixedorder\ and its variants are comparable with brute-force's, and are much better than the trivial solution. \ansa{Another observation is that neither random nor $k$-means variant improves the quality of plain \fixedorder. Further, they introduce more variance in the result quality (0.033 for random and 0.045 for $k$-means in terms of combined standard deviation), and slightly increase the running time. Therefore, in the rest of the section, we focus on the plain \fixedorder\ algorithm.}

\begin{figure}[t]
\vspace{-3ex}
\centering
%\begin{minipage}[t]{.47\textwidth}
%  \centering
 %\vspace*{\fill}
\subfloat[{\scriptsize Running time: $L{=}5,D{=}3$}]{
\includegraphics[width=0.48\linewidth]{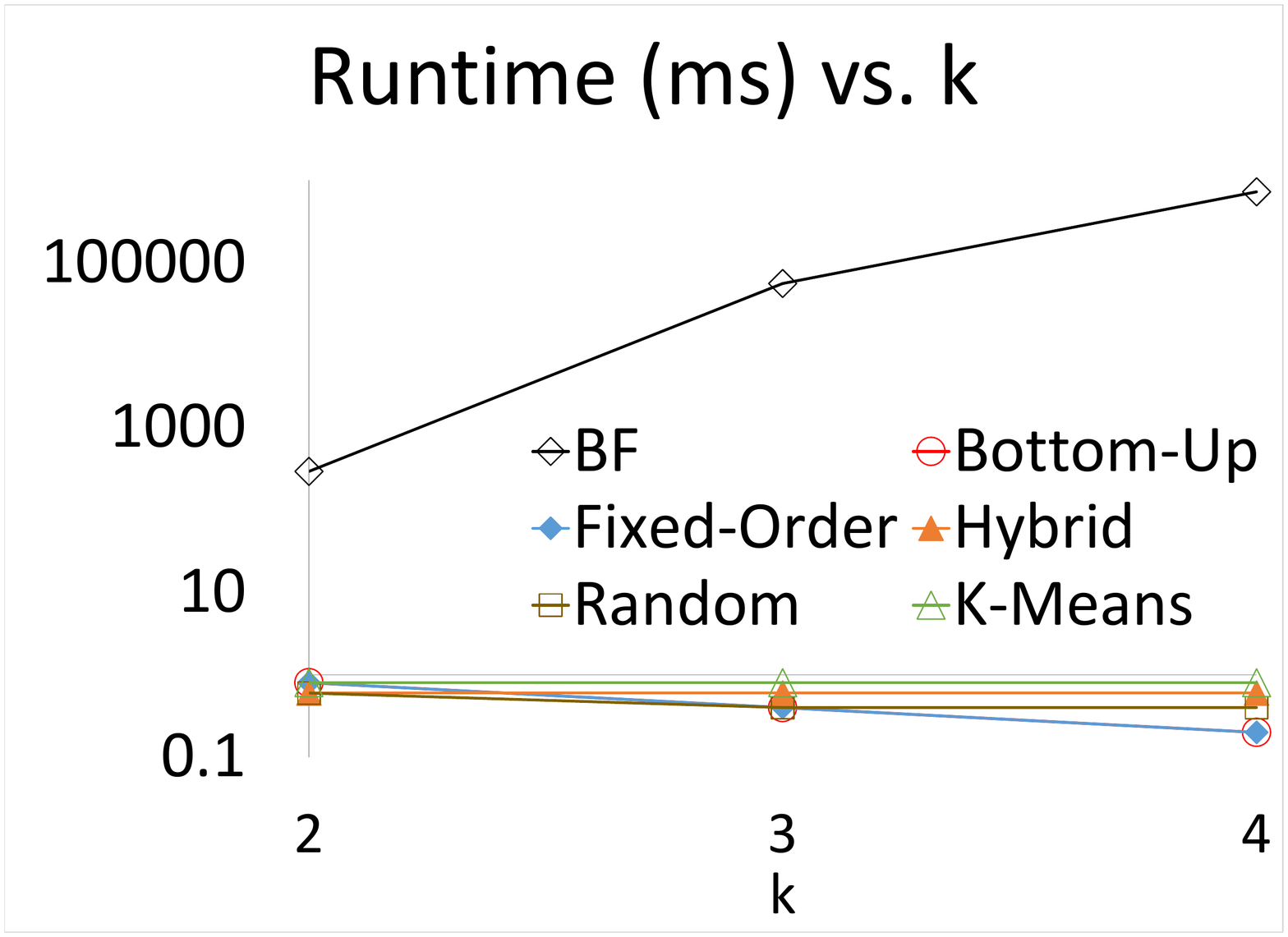}
\label{fig:brute-force-time}}
\hfill
%\hspace{0.02\linewidth}
\subfloat[{\scriptsize Value: $L{=}5,D{=}3$}]{
\includegraphics[width=0.48\linewidth]{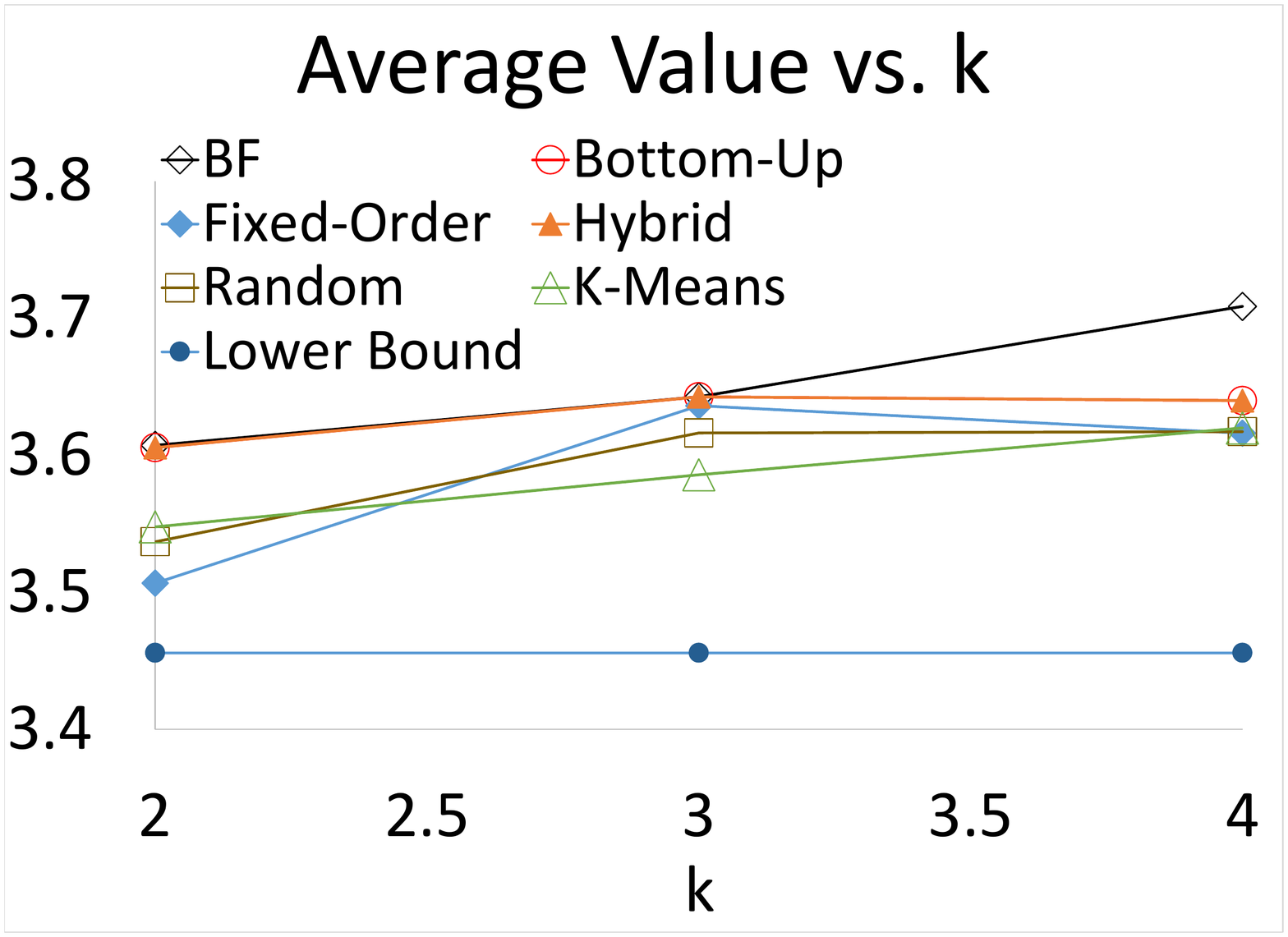}
\label{fig:brute-force-val}}
\vspace{-3ex}
\revb{\caption{Comparison with brute-force.}} %parameters, algorithms, and optimization criteria. The default values of parameters are $m = 8$, $k = 3$, $L = 40$, $D = 3$.}
\label{fig:allexpts-brute}
\end{figure}

\begin{figure*}[t]
\centering

%\begin{minipage}[t]{.47\textwidth}
%   \centering
%\vspace*{\fill}
\subfloat[{\scriptsize Running time vs.\ $k$}]{
\includegraphics[width=0.24\linewidth]{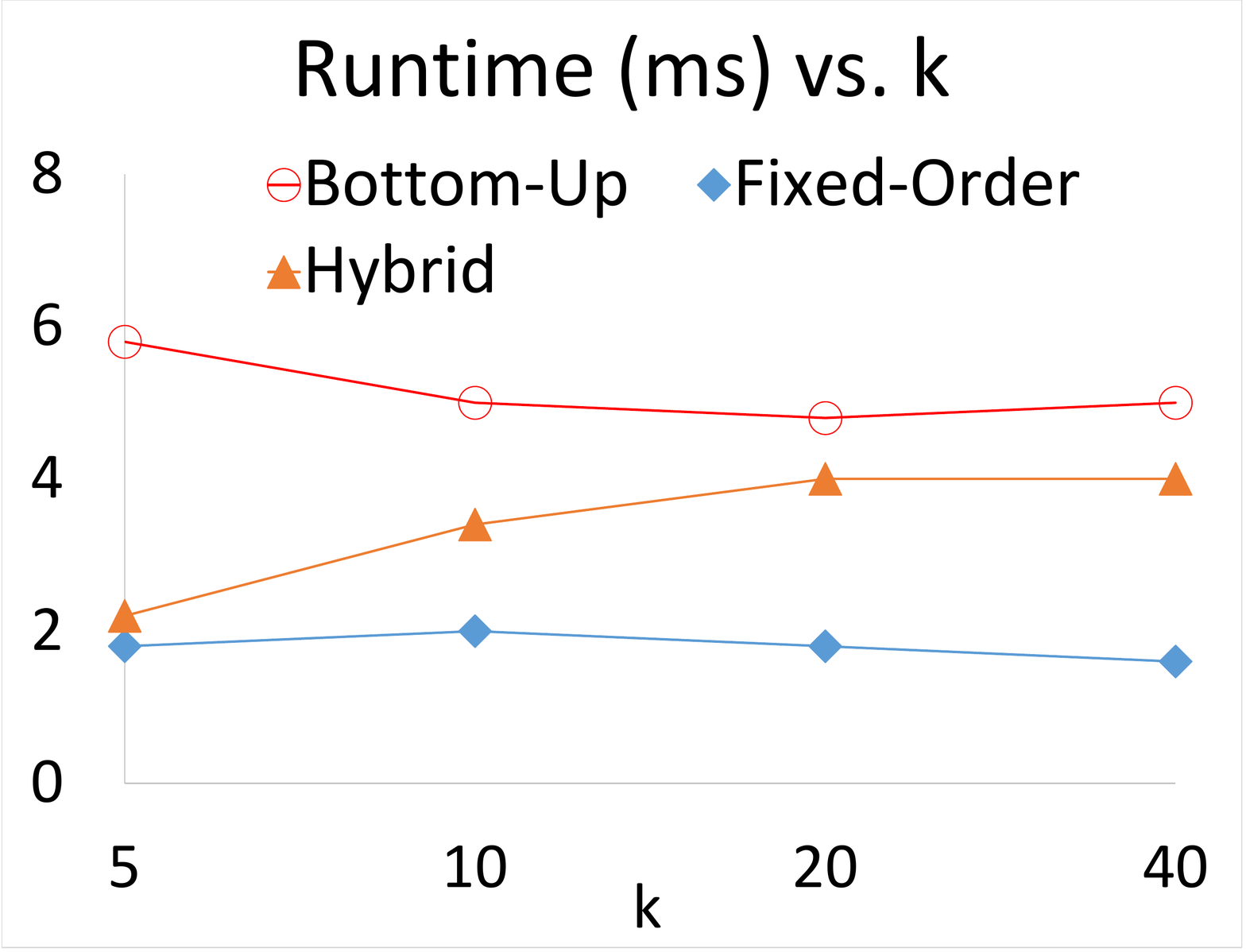}
\label{fig:VaryK-time}}
\hfill
\subfloat[{\scriptsize Value vs.\ $k$}]{
\includegraphics[width=0.24\linewidth]{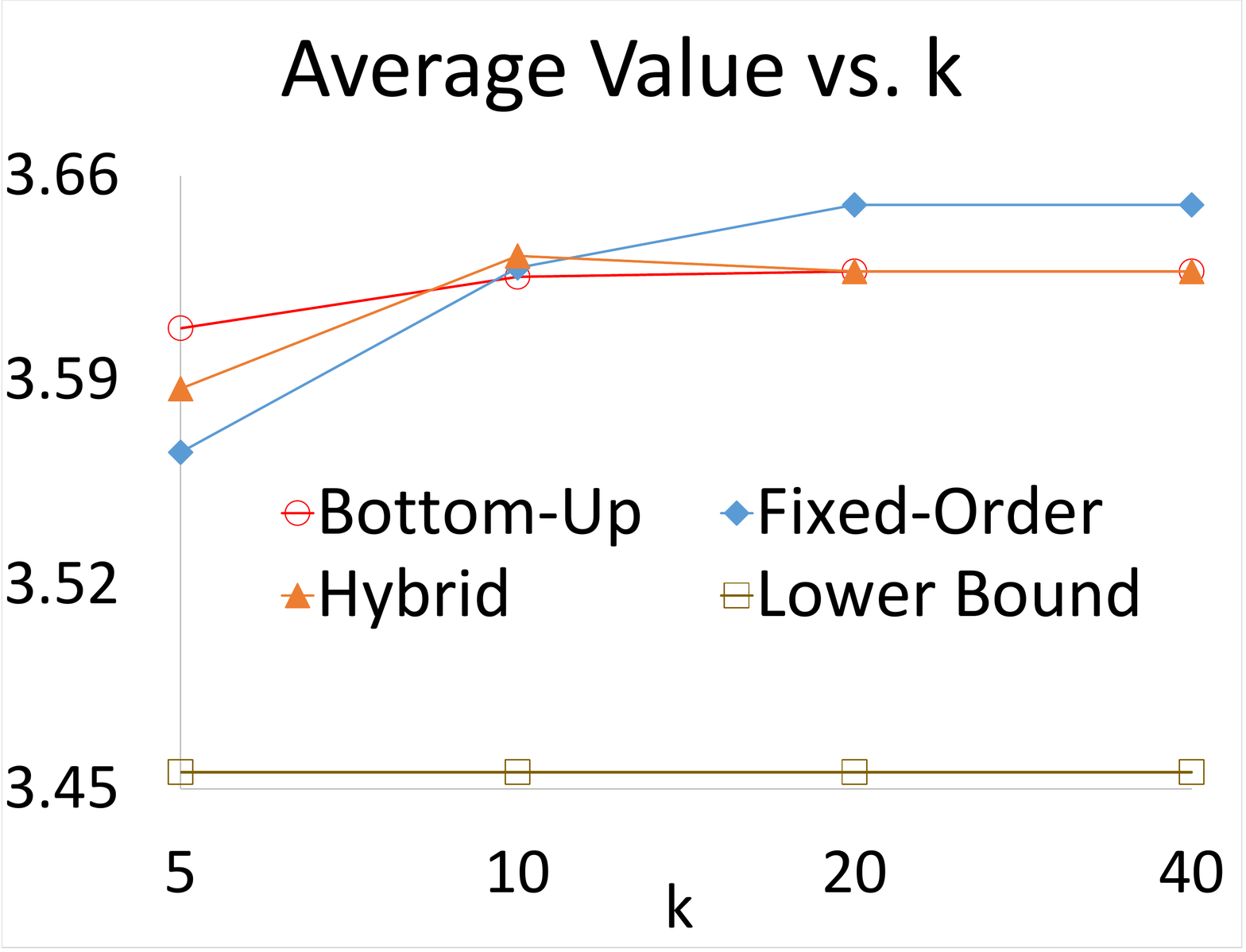}
\label{fig:VaryK-val}}
%\vspace{-0.3cm}
%\end{minipage} \hfill
%\begin{minipage}[t]{.47\textwidth}
%  \centering
 %\vspace*{\fill}
\subfloat[{\scriptsize Running time vs.\ $L$}]{
\includegraphics[width=0.24\linewidth]{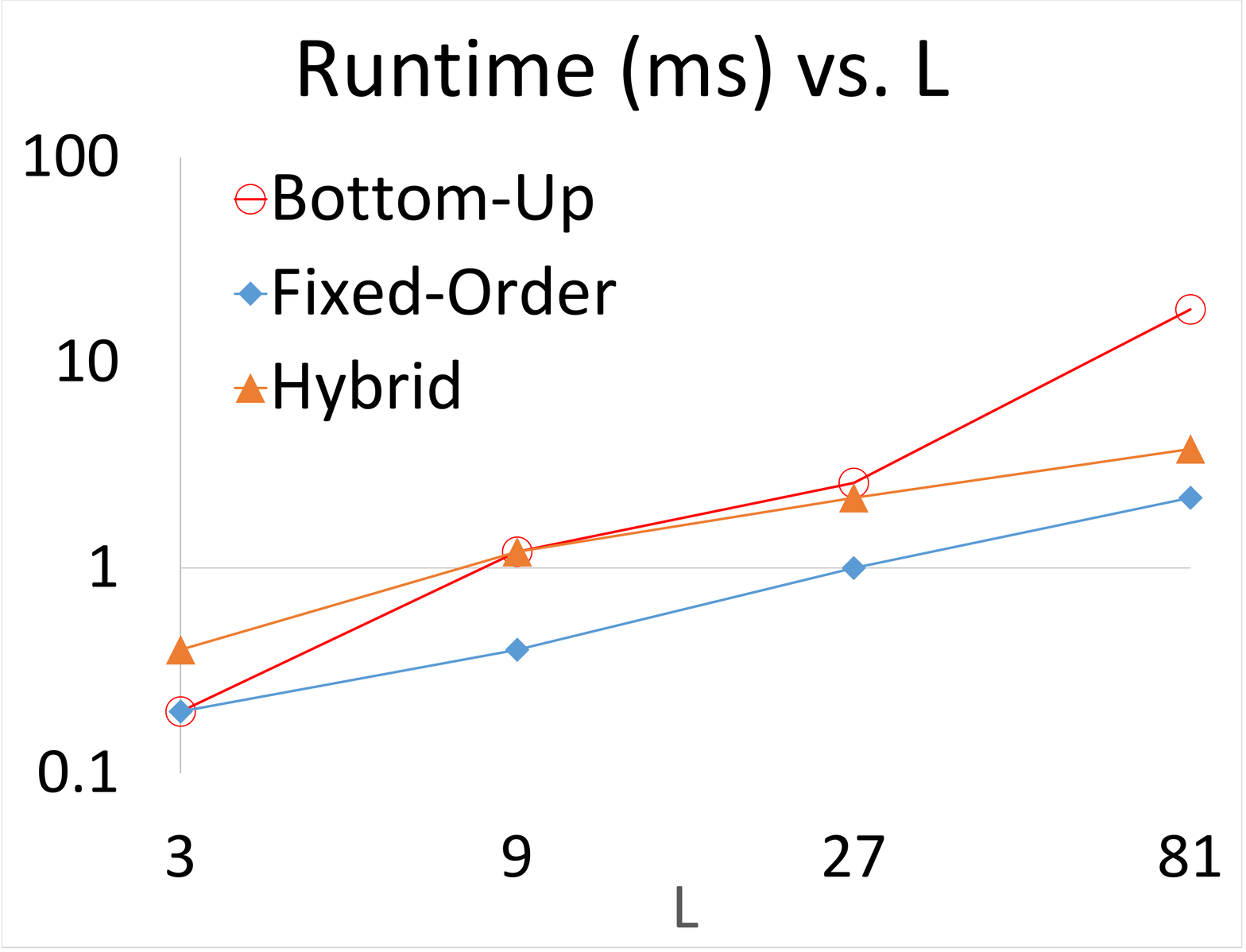}
\label{fig:VaryL-time}}
\hfill
\subfloat[{\scriptsize Value vs.\ $L$}]{
\includegraphics[width=0.24\linewidth]{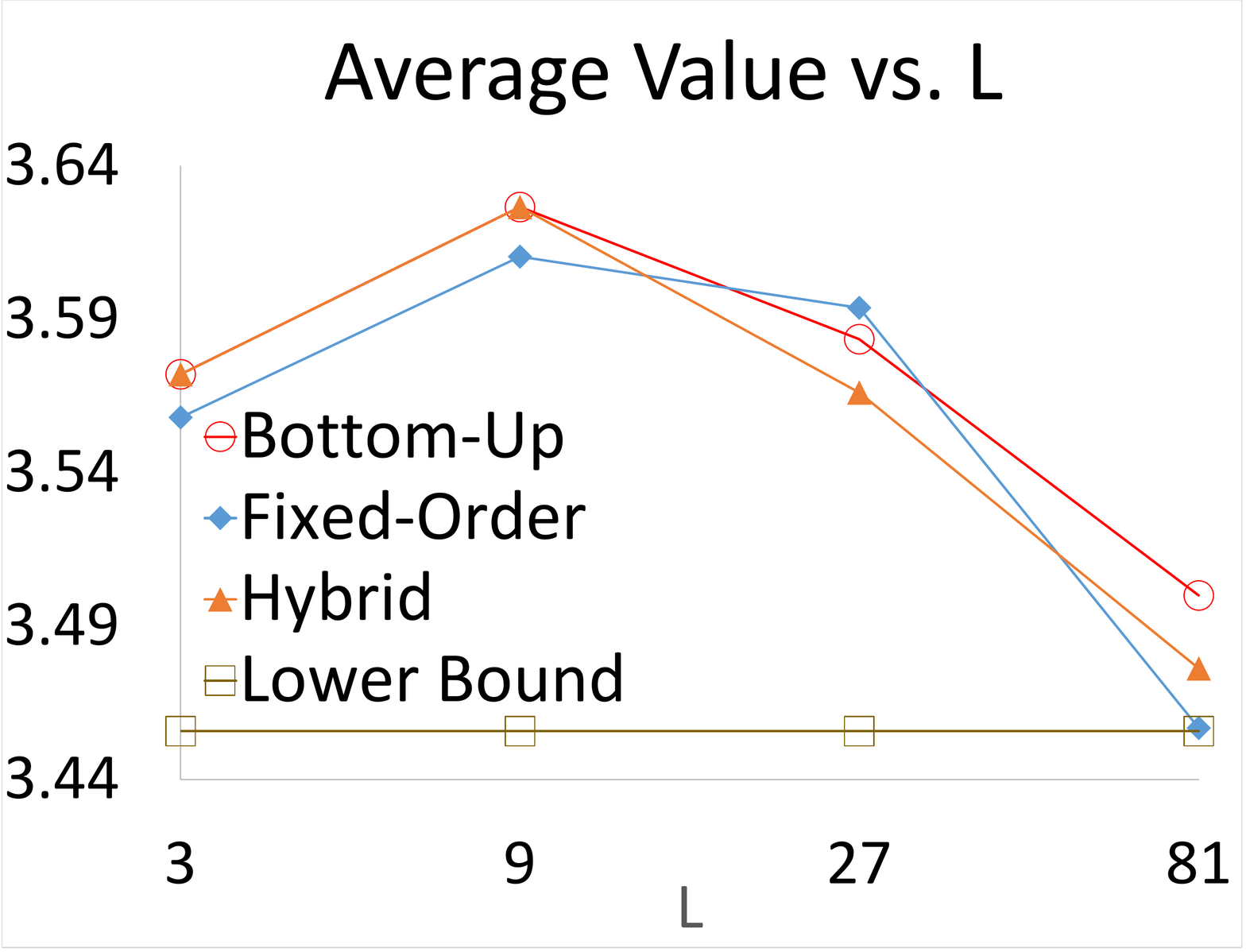}
\label{fig:VaryL-val}}
%\vspace{-0.3cm}
%\end{minipage}\\
\vspace{-0.3cm}
%\begin{minipage}[t]{.47\textwidth}
%   \centering
 %\vspace*{\fill}
\subfloat[{\scriptsize Running time vs.\ $D$ ($k{=}10, L{=}40$)}]{
\includegraphics[width=0.24\linewidth]{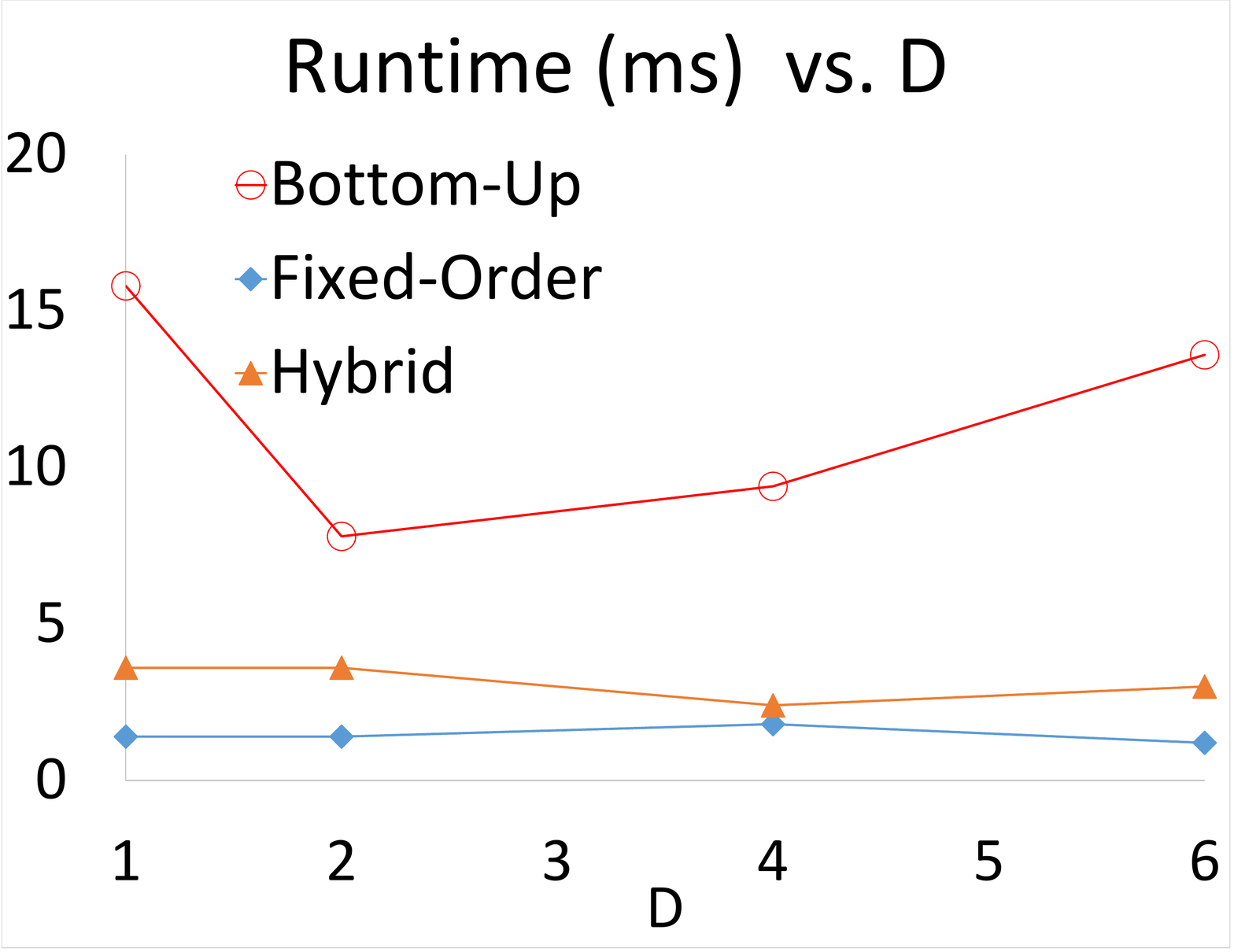}
\label{fig:VaryD-time}}
\hfill
\subfloat[{\scriptsize Value vs.\ $D$ ($k{=}10, L{=}40$)}]{
\includegraphics[width=0.24\linewidth]{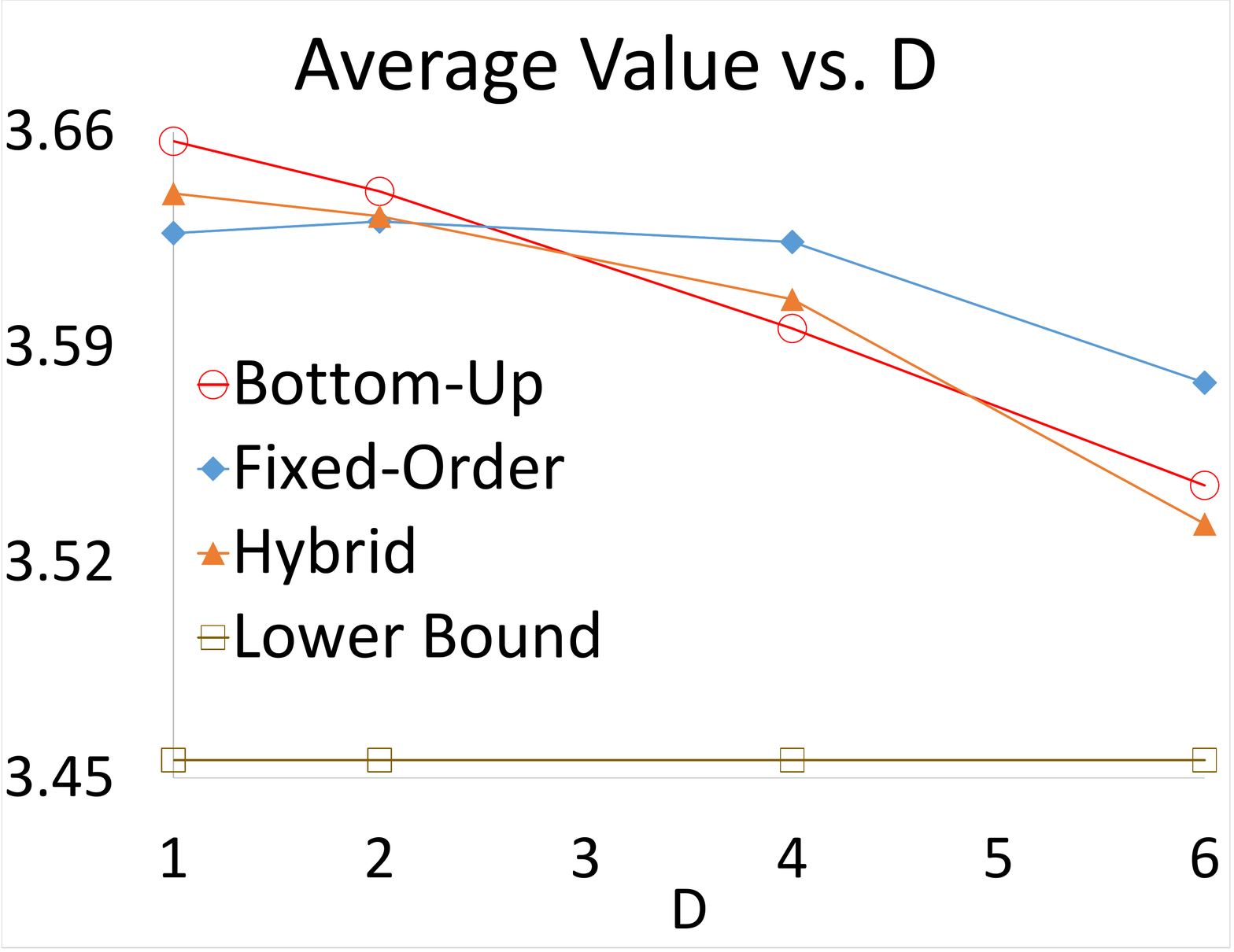}
\label{fig:VaryD-val}}
%\vspace{-0.3cm}
%\end{minipage} \hfill
%\begin{minipage}[t]{.47\textwidth}
%  \centering
 %\vspace*{\fill}
\subfloat[{\scriptsize Initialization time vs.\ $m$  ($k{=}L{=}20$)}]{
\includegraphics[width=0.24\linewidth]{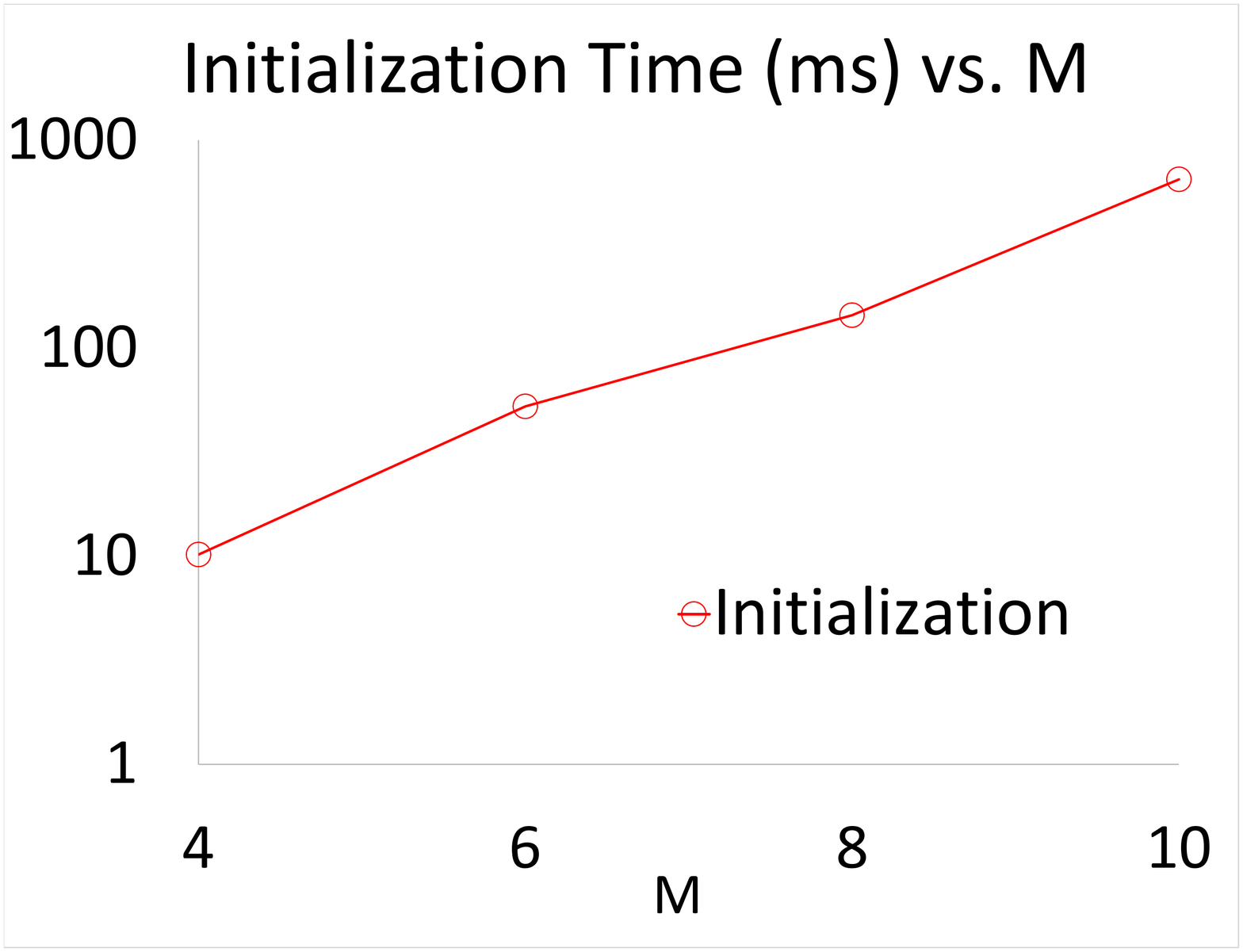}
\label{fig:VaryM-precomp}}
\hfill
\subfloat[{\scriptsize Running time vs.\ $m$  ($k{=}L{=}20$)}]{
\includegraphics[width=0.24\linewidth]{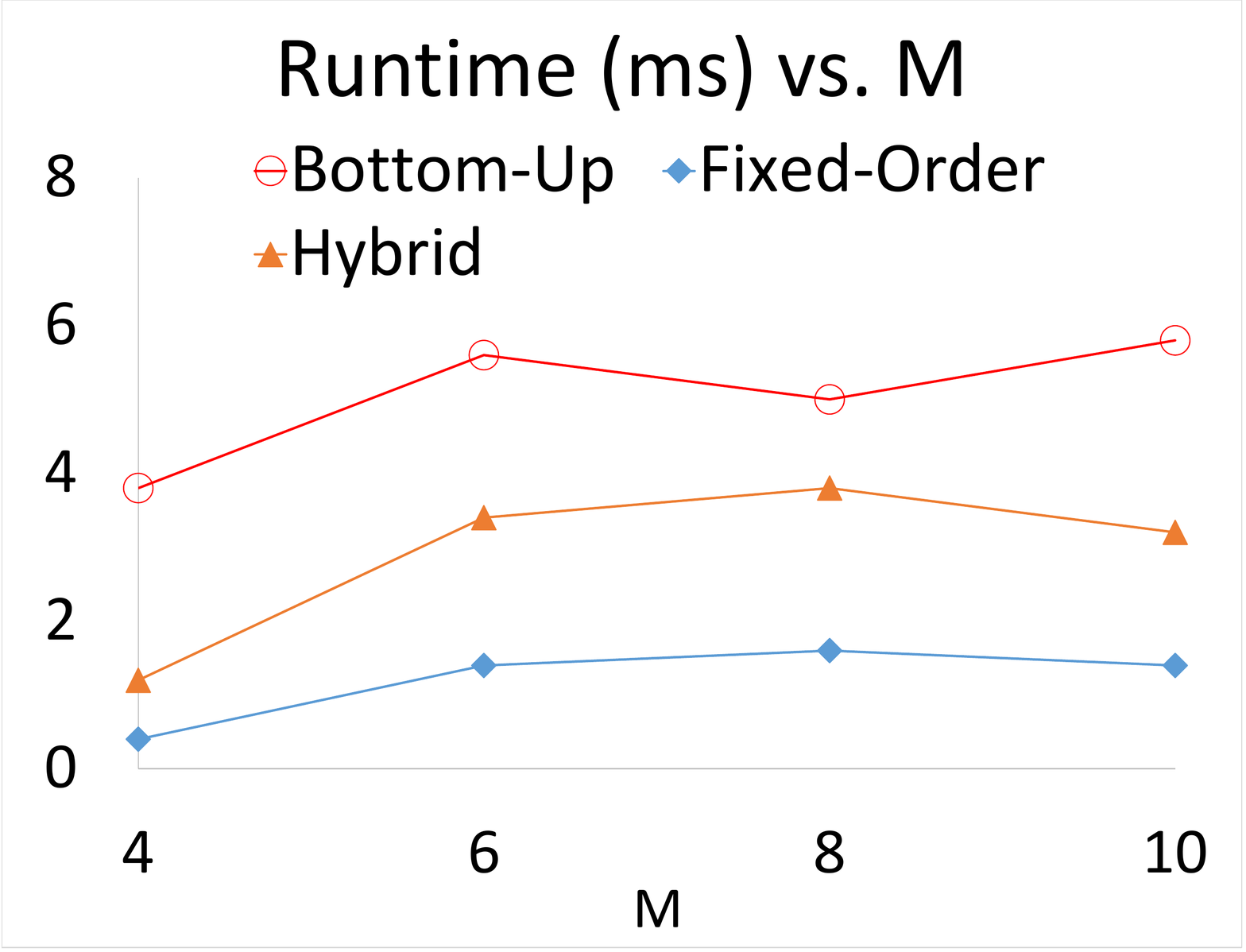}
\label{fig:VaryM-time}}
%\vspace{-0.3cm}
%\end{minipage}

\cut{
\begin{minipage}[t]{.45\textwidth}
\centering
 %\vspace*{\fill}
\subfloat[{\scriptsize Average value: \maxval\ vs. \minsize}]{
\includegraphics[width=0.45\linewidth]{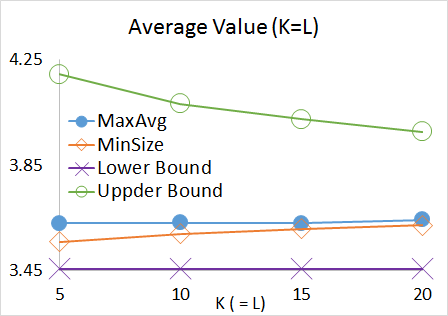}~~~~
\label{fig:maxmin-val}}
\hspace{0.05\linewidth}
\subfloat[{\scriptsize No. of elements covered:  \maxval\ vs. \minsize}]{
\includegraphics[width=0.45\linewidth]{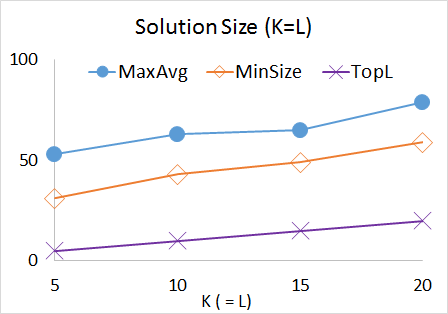}~~~~
\label{fig:maxmin-size}}
\vspace{-0.3cm}
\end{minipage}\\
}
\vspace{-2ex}
\revb{\caption{Experimental results varying parameters. %, algorithms, and optimization criteria. 
The default values of parameters are $m = 8$, $k = 3$, $L = 40$, $D = 3$.}}
\label{fig:allexpts}
\vspace{-0.5cm}
\end{figure*}

%\subsubsection{Effect of size parameter $k$}\label{sec:vary-k}
%\vspace{5mm}
\textbf{Effect of size parameter $k$.~}
Figure~\ref{fig:VaryK-time} shows the running time varying $k$. The running time of \fixedorder\ is the best as it never considers more than $k$ candidate merges per step; in contrast, \bottomup\ may consider a quadratic number of candidate merges per step and it is slower than \fixedorder\ as a consequence. \hybrid\ is in the middle for runtime as expected. Furthermore, $D=3$ helps bound the size of $\allclusters_\ell$ and hence the cost of computing the set cover.  The running time tends to decrease with bigger $k$ for both \fixedorder\ and \bottomup; the reason is that fewer merges are needed to reach the desired $k$. However, for \hybrid, since larger $k$ makes the candidate pool larger and might brings in more calculation in the second phase (\bottomup\ phase), the run time for \hybrid\ tends to get closer to \bottomup.

%In Figure~\ref{fig:VaryK-val}, the fastest algorithm,
The average value of \fixedorder\ is lower than the value of \bottomup\ or \hybrid\ as explained in Section~\ref{sec:algorithms}, although gets better with larger $k$ in Figure~\ref{fig:VaryK-val}.
%, is surprisingly competitive in quality as well: it is comparable to those of \bottomup\ and \hybrid. Its heuristic ordering of processing top-$L$ result tuples seems to work well in practice in this case.

%\subsubsection{Effect of coverage parameter $L$.~}\label{sec:vary-L}
\textbf{Effect of coverage parameter $L$~}
%In this part of experiment we set L as our only variable. For other constraints, we have $k = 3$, $D = 3$ and $M = 8$. We took four points, $L = 3,9,27,81$ in order to explore the trend for both runtime and average output value for greedy-k-average(denoted by \emph{MaxAvg}) and greedy-k-average-L-weighted(denoted by \emph{MaxAvgWeighted}). 
%The results varying $L$ are shown in Figures~\ref{fig:VaryL} (a) and (b) ($k = 3$, $D = 3$). 
Figure~\ref{fig:VaryL-time} shows that running time of all algorithms increase as the number of elements to be covered $L$ increases. Since \fixedorder\ depends linearly on $L$, it is less affected by $L$, whereas \bottomup\ treats individual elements as clusters and may incur quadratic time w.r.t. $L$. For \hybrid, with the restriction of the size of the candidate pool determined by $k$, the run time increase is slower than \bottomup\ and is comparable with \fixedorder. Note that in Figure~\ref{fig:VaryL-val}, the upper bound decreases since with $L$ increasing, the average value of the top-$L$ elements decreases. All three algorithms seem to be close in terms of average values, but \bottomup\ has the highest value most of the times and \hybrid\ usually gets results close or equal to \bottomup.

%\subsubsection{Effect of distance parameter $D$}\label{sec:vary-D}
\textbf{Effect of distance parameter $D$.~}
In Figure~\ref{fig:VaryD-time}, \fixedorder\ is mostly unaffected by $D$ since the distance value is checked only once when an element is considered. \hybrid\ is relatively constant as well given that when the distance check starts, the number of unchecked tuples is limited by the candidate pool. For \bottomup\,  as $D$ increases, the run time drops first and then climbs. It may be caused by the existence of a balance point on number of calculations between distance insurance (phase 1 in \bottomup) and greedy merge (phase 2 in \bottomup).
\par
The average value of the output (the value of objective function) is highest when $D = 1$ (since singleton clusters are collected for $L = k = 20$), then drops with $D$ going up as shown in Figure~\ref{fig:VaryD-val}.
% which also leads to the jump from $D = 1$ to $D = 2$.
%The former phenomenon happens because $D = 0$ is actually the same case with $D = 1$ since two different clusters' distance is at least 1. 
%For larger values of $D$, greater distance leads to fewer available clusters since we discard all lower levels.
% it's reasonable to have the objective function value decreasing. The sudden fall between $D=1$ and $D=2$ is cased by definition of the objective function. In our case, an element is counted only once even if it is contained by multiple selected clusters. Consequently for $D=1$ all clusters are top-L elements theirselves (which is the highest value in theory), while for $D>=2$ the algorithms must find clusters to cover those elements, which usually bring in lots of redundant elements (elements not contained in top-L). That's the reason for the sudden drop of the value of the objective function.

%\subsubsection{Effect of number of attributes $m$}\label{sec:vary-m}
\textbf{Effect of number of attributes $m$.~}
%All discussion till here are based on a fixed query shown in {\color{red}{SECTION ???}} 
%Here we vary the total number of group-by attributes $m$ in the input query. 
Varying the number of grouping attributes $m$ also illustrates the effect of varying input data size. Since our algorithms run on the output of an aggregate query, as $m$ increases, our input data size $|S| = n$ is likely to increase (for the $m$ values in Figure~\ref{fig:VaryM-precomp} and \ref{fig:VaryM-time}, the size of the input ranges from 140 to 280). 
When a new query comes, the system performs an initialization step of constructing clusters and the semi-lattice structure. This initialization time is shown in Figure~\ref{fig:VaryM-precomp}. This step is performed only once per query, varying $k$ and $D$ does not need another initialization. 
Our implementation takes from $10ms$ when $m = 4$ to about $1s$ when $m = 10$.
% (the time in the figure is in milliseconds (ms)). 
Note that this is the number of group-by attributes in the top-$k$ aggregate query, not in the original dataset. So it is likely to have a small value $\leq10$. Figure~\ref{fig:VaryM-time} has the running time of the algorithms for $k=L=20, D=3$ and shows that all the algorithms return results in real time (in a few ms) after the initialization step.

\begin{figure}[t]
\centering
%\begin{minipage}[t]{.95\textwidth}
%  \centering
 %\vspace*{\fill}
\subfloat[{\scriptsize Runtime vs.\ $k$}]{
\includegraphics[width=0.46\linewidth]{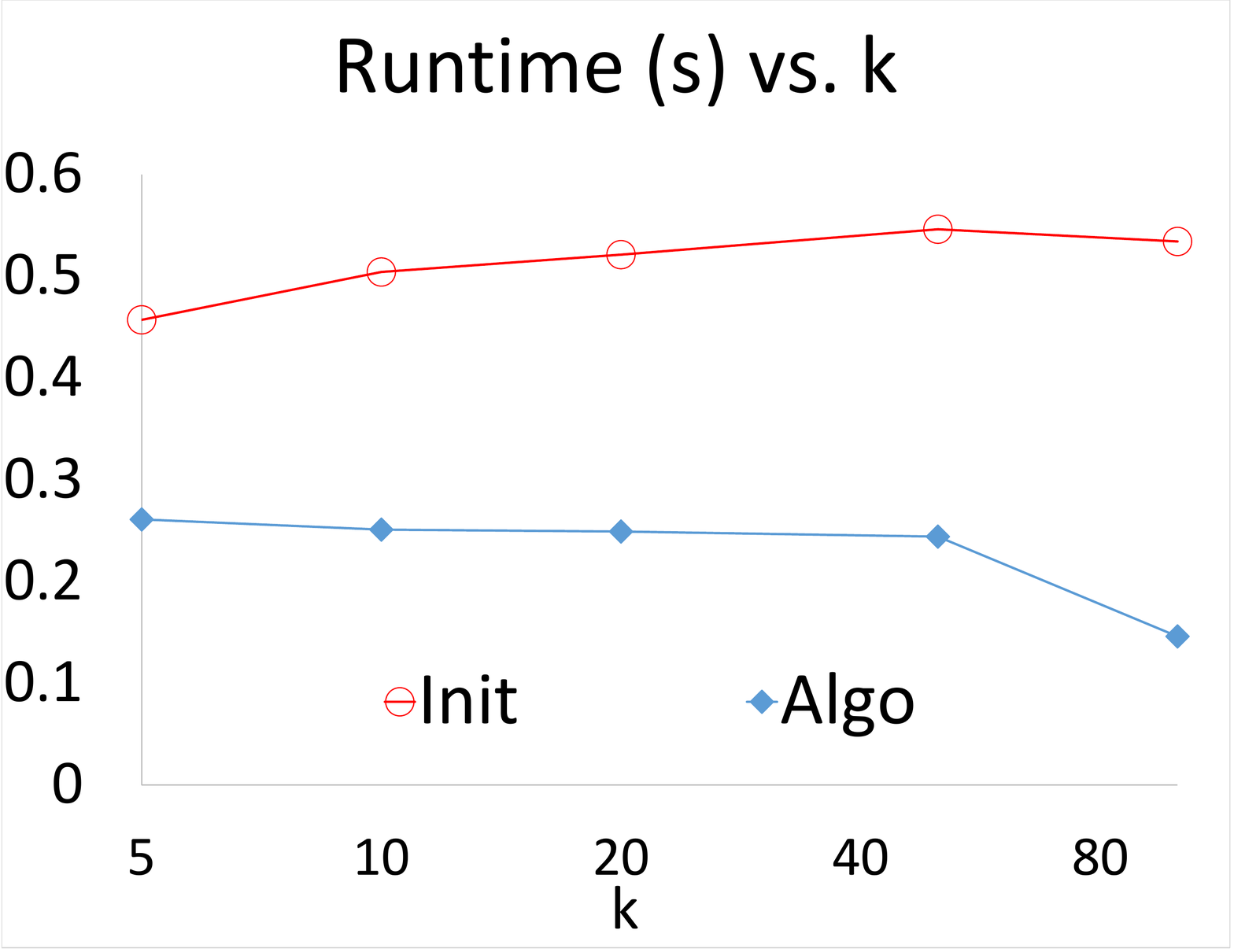}
\label{fig:pre-varyK}}
%\hspace{0.05\linewidth}
\hfill
\subfloat[{\scriptsize Runtime comparison, $N=7000$}]{
\includegraphics[width=0.46\linewidth]{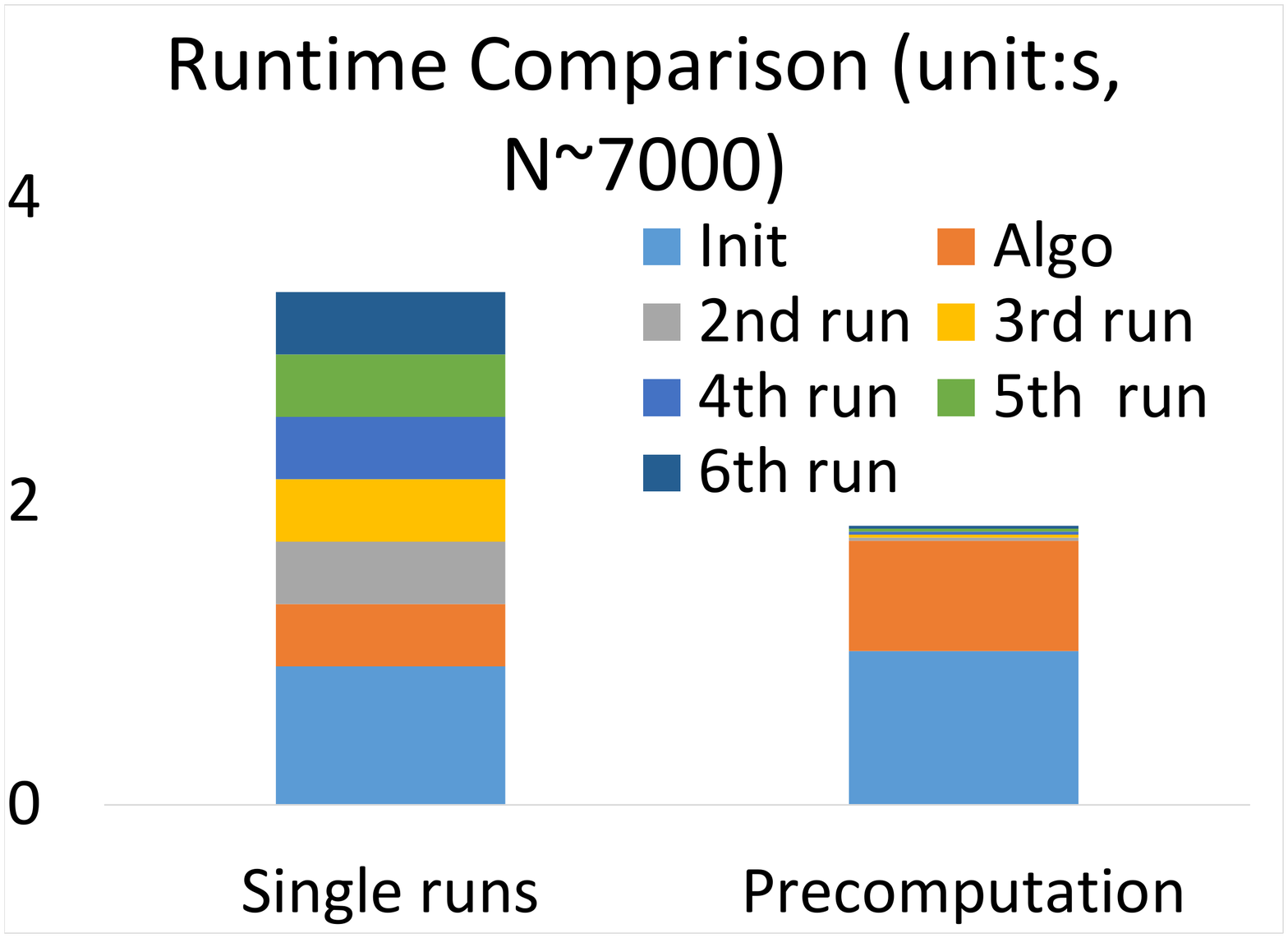}
\label{fig:pre-comparison}}
%\hspace{0.02\linewidth}
\vspace{-0.3cm}
\subfloat[{\scriptsize Single runtime vs.\ $L$}]{
\includegraphics[width=0.46\linewidth]{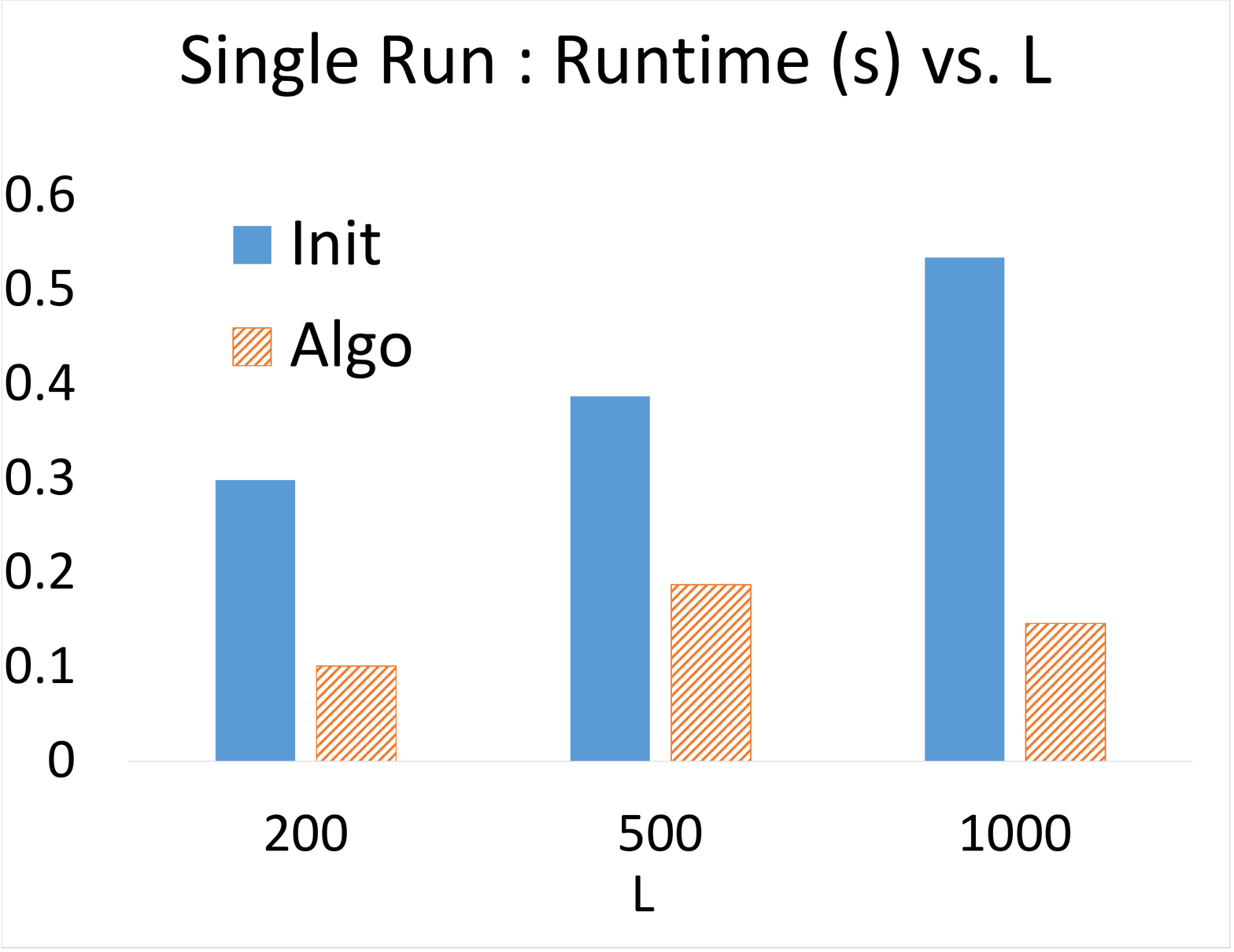}
\label{fig:pre-singleVaryL}}
%\vspace{-0.3cm}
%\end{minipage}\hfill
%\begin{minipage}[t]{.95\textwidth}
%   \centering
 %\vspace*{\fill}
\hfill
\subfloat[{\scriptsize Runtime vs.\ $L$; with precomputation}]{
\includegraphics[width=0.46\linewidth]{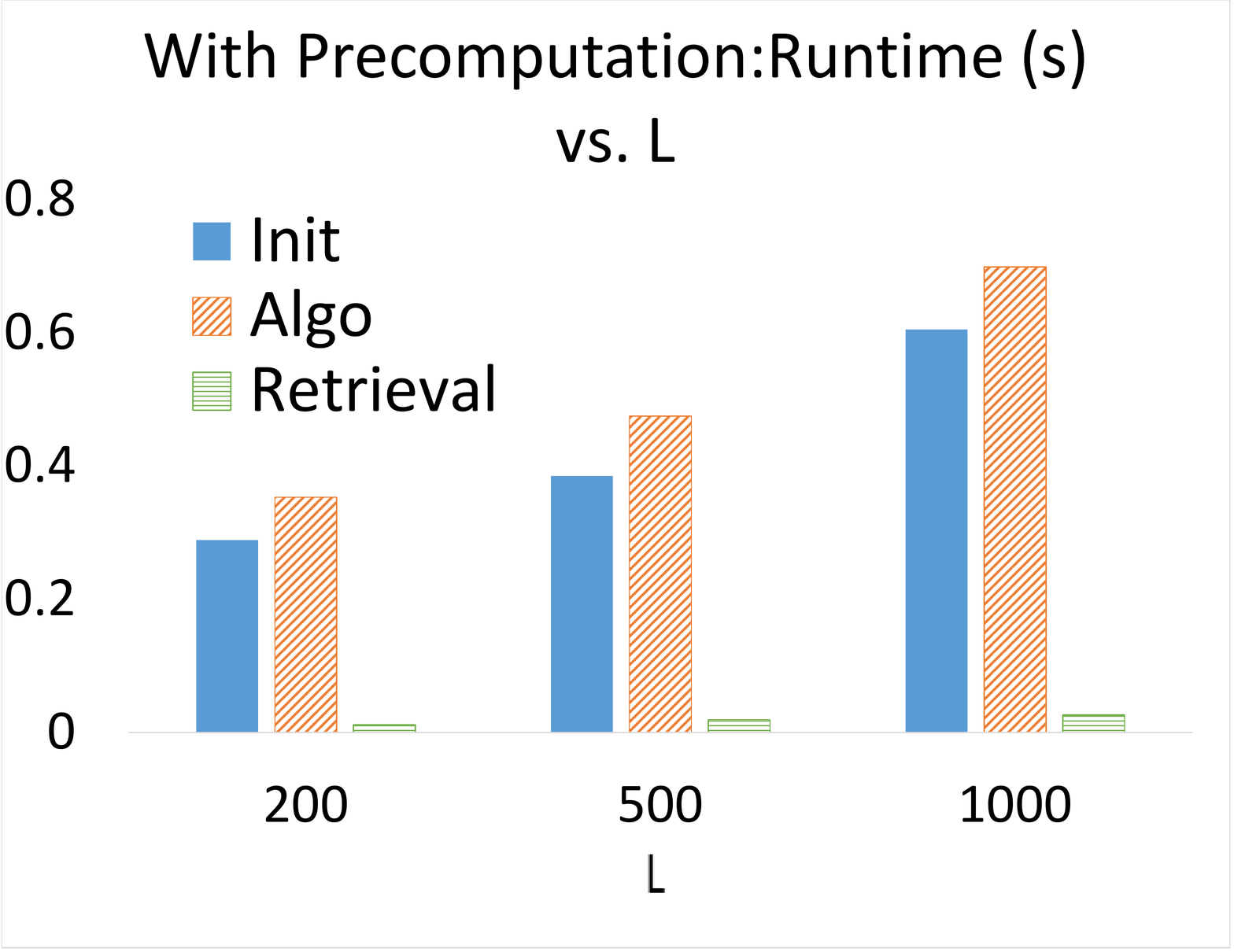}
\label{fig:pre-multiVaryL}}
\vspace{-0.3cm}
%\hspace{0.025\linewidth}
\subfloat[{\scriptsize Single runtime vs.\ $N$}]{
\includegraphics[width=0.46\linewidth]{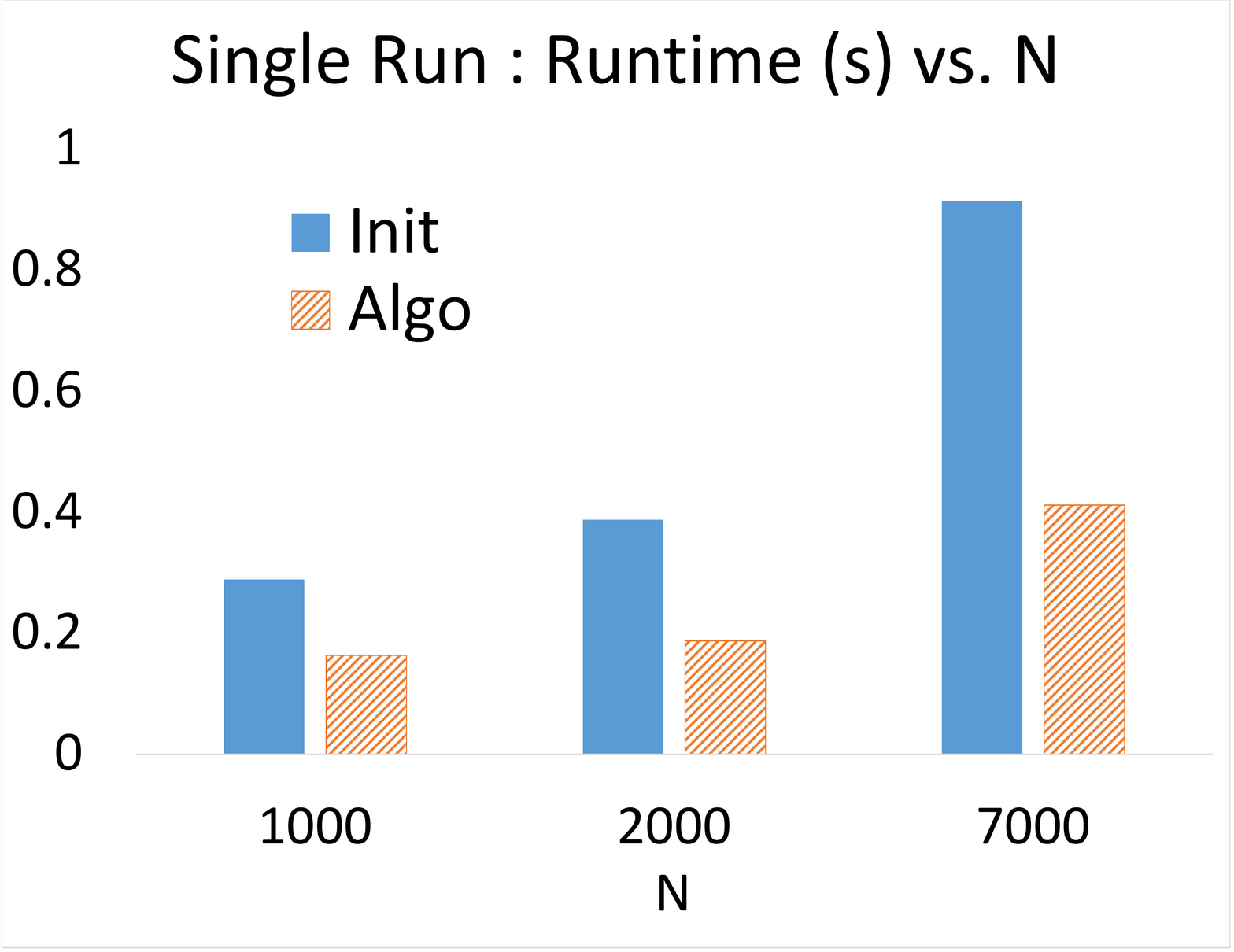}
\label{fig:pre-singleVaryN}}
%\hspace{0.01\linewidth}
\hfill
\subfloat[{\scriptsize Runtime vs.\ $N$; with precomputation}]{
\includegraphics[width=0.46\linewidth]{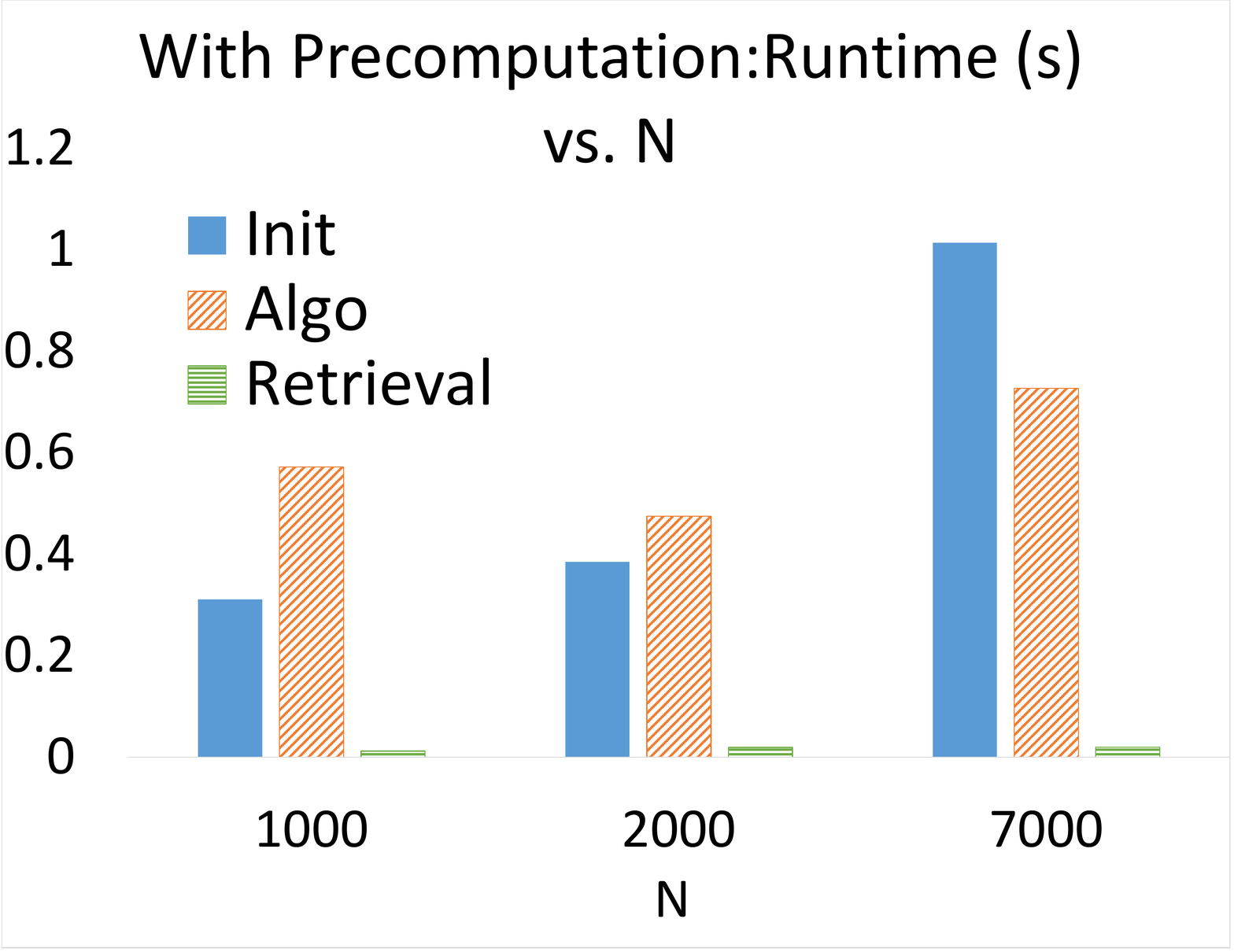}
\label{fig:pre-multiVaryN}}
%\vspace{-0.3cm}
% \end{minipage}\hfill
\vspace*{-2ex}
\revb{\caption{Experimental results varying parameters, and with or without
  precomputation}}
\label{fig:preexpts}
%\vspace{-4ex}
\end{figure}

\vspace{-2mm}
\subsection{Cost and Benefit of Precomputation}\label{sec:exp-precomputation}

The performance evaluation for precomputation is shown by varying $k$, $L$ and $D$ separately, and comparing the running time of \hybrid\ between precomputation implementation and non-precomputation (single) implementation. %From Section~\ref{exp-pre-diffk} to Section~\ref{exp-pre-diffn}, the discussion will be focused on the effect brought by varying parameters. Discussion on single run vs. multiple runs will be presented in Section~\ref{exp-pre-sm}.

%\subsubsection{Effect of size parameter $k$}\label{exp-pre-diffk}
\textbf{Effect of size parameter $k$.~}
In this experiment, $L=1000$, $D=2$ and $N=2087$ are fixed. Five $k$ values are chosen: $5,10,20,50,100$. The running time result is shown in Figure~\ref{fig:pre-varyK}: the initialization time hardly changes with $k$ growing since $k$ does not affect the initialization process. Given that a larger $k$ requires less operations in \bottomup\ phase to reach the target $k$, the running time for the algorithm (\hybrid) has a descending trend.

%\subsubsection{Effect of coverage parameter $L$}\label{exp-pre-diffl}
\textbf{Effect of coverage parameter $L$.~}
The fixed parameters are $k=20$, $D=2$, and $N=2087$. Three $L$ values are selected for the experiment:$L=200,500,1000$. The running time results for single version and precomputation version are presented in Figures~\ref{fig:pre-singleVaryL} and \ref{fig:pre-multiVaryL}. Both implementations have rising trend with respect to $L$ and share similar initialization times as expected. Although under the same parameter combinations, algorithm runtime for single implementation is much lower than precomputation time in the other implementation (about 1/3 to 1/4), but the retrieval time for precomputation implementation is extremely short (tens of milliseconds), which can make up for the time in multiple runs.

%\subsubsection{Effect of total elements N}\label{exp-pre-diffn}
\textbf{Effect of total elements N.~}
Here three parameters $k, L, D$ are fixed as $k=20$, $L=500$ and $D=2$. We varied total input elements to test the system's performance with relatively higher capacity: $N=927,2087$ and $6955$. The running time result is shown in Figures~\ref{fig:pre-singleVaryN} and \ref{fig:pre-multiVaryN}. The changing trends are similar with those in %Section~\ref{exp-pre-diffl}
Figures~\ref{fig:pre-singleVaryL} and \ref{fig:pre-multiVaryL}, but a significant increase for the initialization time can be observed with $N$ growing. This is caused by materializing more possible clusters brought by variety of tuples.

%\subsubsection{Single run vs.\ multiple runs}\label{exp-pre-sm}
\textbf{Single run vs. multiple runs.~}
From Figure~\ref{fig:pre-singleVaryL}, \ref{fig:pre-multiVaryL}, \ref{fig:pre-singleVaryN} and \ref{fig:pre-multiVaryN}, the information is enough for comparing precomputation and non-precomputation versions on both single run and multiple runs scenario - For a single run, precomputation process is unused, making precomputation version the slower and more expensive choice; For multiple runs with similar setup, the precomputation version has increasingly more benefits brought by the rapid retrieval process taking tens of milliseconds. In order to provide a quantitative comparison, we provide Figure~\ref{fig:pre-comparison} with $N=6955$: if only a single run is required, the single version of \hybrid\ is clearly faster and cheaper than the precomputation version. However, when the third run finishes, the precomputation version is already faster; When all six runs finish, the single version takes about two times in terms of running time compared with the precomputation version. %The results are expected and perfect fit for our design.

\begin{figure}[h]
\vspace{-2ex}
\centering
%\begin{minipage}[t]{.47\textwidth}
%  \centering
 %\vspace*{\fill}
\subfloat[{\scriptsize Initialization running time with and without optimization, coverage $L$ varies}]{
\includegraphics[width=0.46\linewidth]{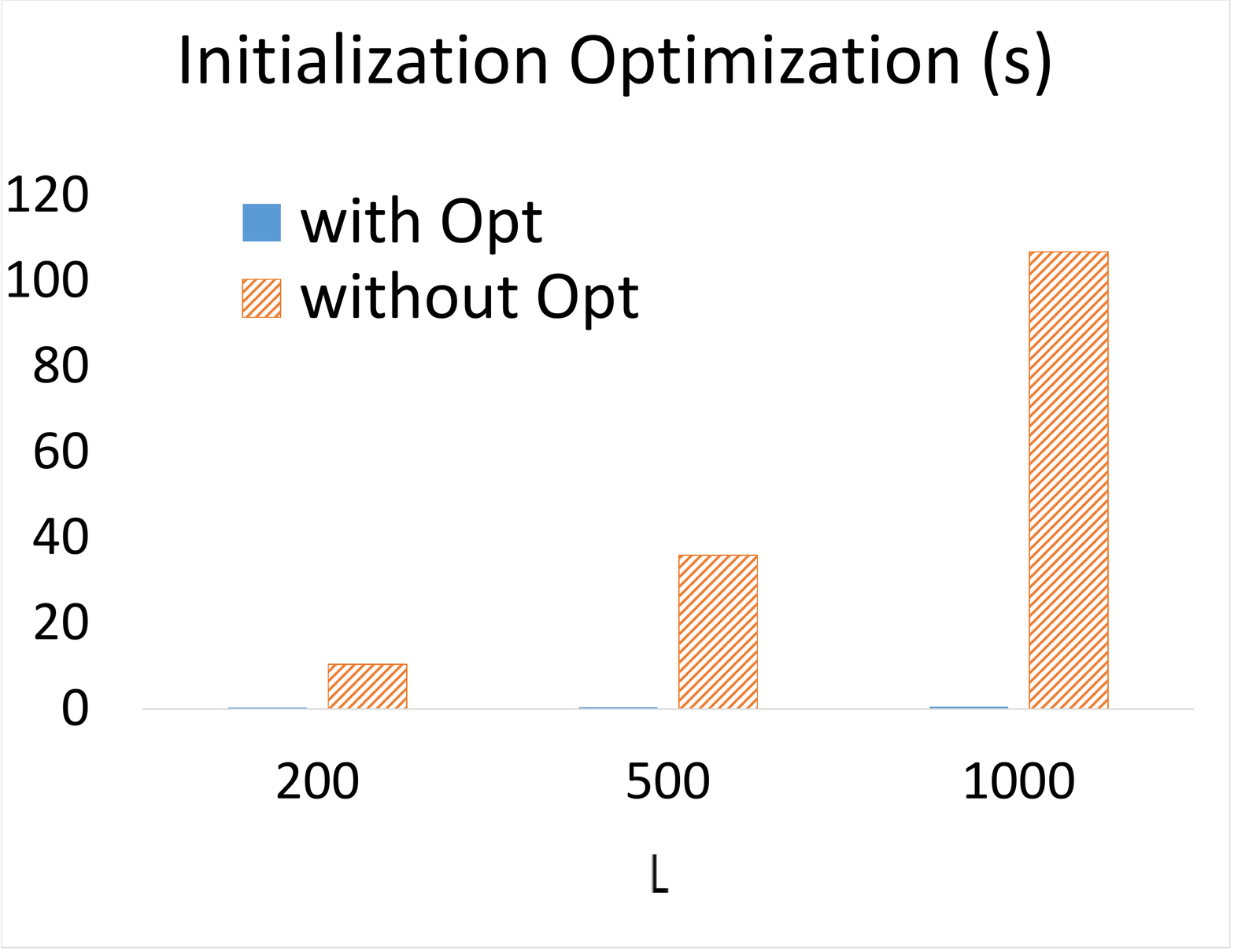}
\label{fig:opt-init}}~~
%\hspace{0.02\linewidth}
\hfill
\subfloat[{\scriptsize Algorithm running time with and without optimization,  coverage $L$ varies}]{
\includegraphics[width=0.46\linewidth]{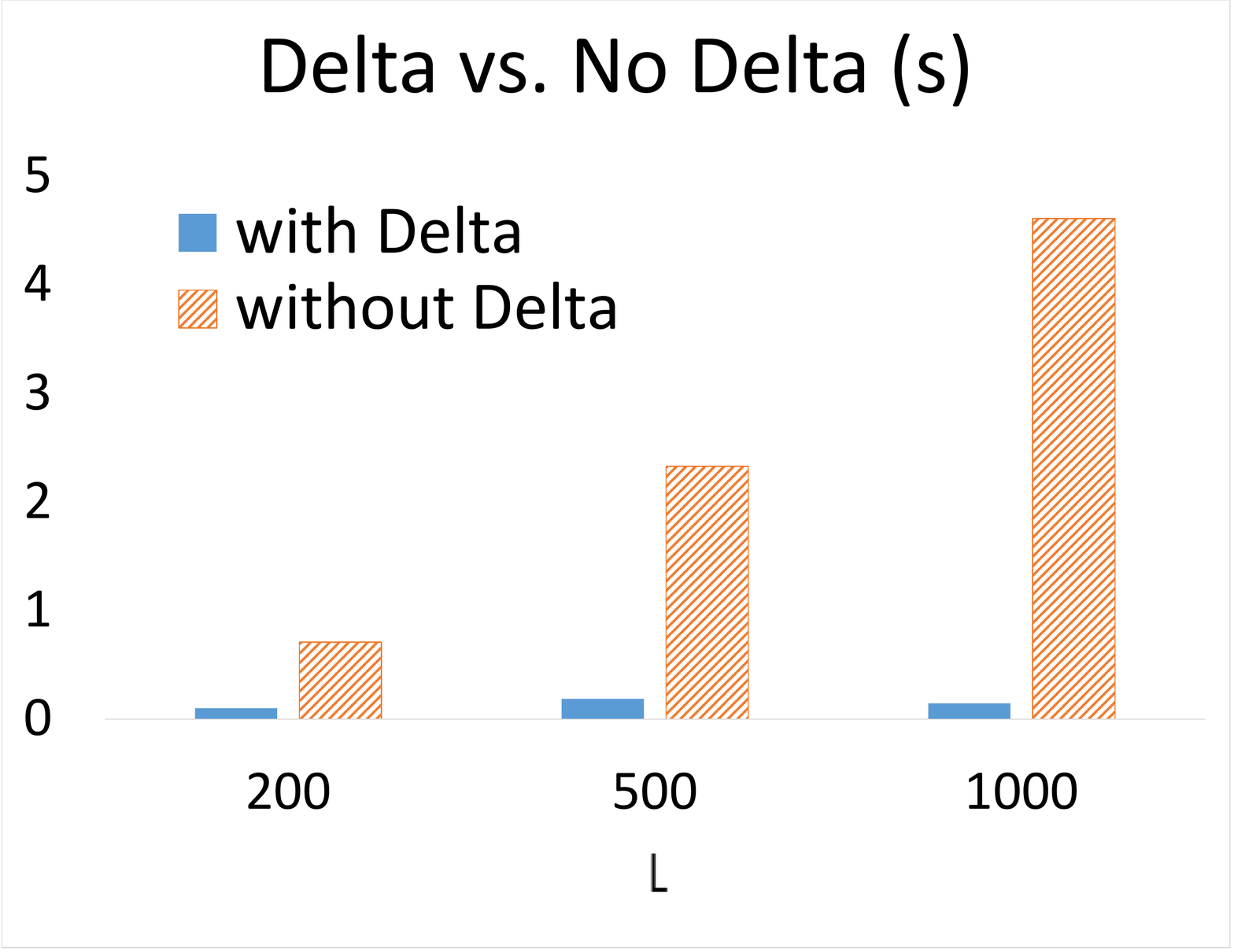}
\label{fig:opt-delta}}
\vspace{0cm}
%\end{minipage}\hfill
\vspace{-1.5ex}
\revb{\caption{Experiment on effects of optimizations}}
\label{fig:optexpts}
\end{figure}

\textbf{Timing for Guidance Visualization.~}
We evaluated the running time for the generation of guidance visualization under different queries. The generation times are similar among different number of attributes - 20-40 milliseconds when the number of attributes is from 4 to 10 with $N=2087$ in MovieLens dataset, meeting the requirement for interactive performance.
% This shows that our implementation meets the interactivity requirement for the guidance visualization. 

\subsection{Benefit of Optimizations}\label{sec:exp-opt}

%In order to show the advantage brought by optimization discussed in Section~\ref{sec:guid-opt}, we performed two set of experiments as discussed below.

%The following experiments show benefits of optimizations discussed in Section~\ref{sec:guid-opt}.

%Section~\ref{exp:opt-init} shows the comparison result for initialization related optimization, and \emph{Delta Judgment} related optimization are examined in Section~\ref{exp:opt-delta}.

%\subsubsection{Optimization in Initialization} \label{exp:opt-init}
\textbf{Cluster generation and mapping to tuples.~}
%The experiment in this section is designed for Section~\ref{guid:opt-init}. 
Since $L$ is the only factor that affects the initialization time when the input size $N$ is fixed, in this experiment, $L$ varies among $200,500$ and $1000$ while others are fixed: $k=20, D=2, N=2087$. The result is presented in Figure~\ref{fig:opt-init}. Only the running time of initialization is drawn because the optimizations in this section only affect the initialization time. The optimizations - cluster generation and cluster-tuple mapping - provide significant performance improvement by cutting down the running time from $>100$s for $L=1000$ to $0.5$s.

%\subsubsection{Delta Judgment}\label{exp:opt-delta}
\textbf{Delta Judgment.~}
The effect by introducing \emph{Delta Judgment} is shown in Figure~\ref{fig:opt-delta}. Given that $L$ is also the most effective variable to affect the running time, the experimental settings are the same as the experiment for Figure~\ref{fig:opt-init}.
%in Section~\ref{exp:opt-init}. 
However, only the running time of the algorithm is plotted since \emph{Delta Judgment} has no effect on the running time of initialization. The result in Figure~\ref{fig:opt-delta} shows that the \emph{Delta Judgment} successfully improves the algorithm's efficiency from $4.6$s to $0.15$s when $L=1000$, which is the slowest case in the experiment in this section.

\begin{figure}[t]
\vspace{-2ex}
\centering
%\begin{minipage}[t]{.47\textwidth}
%  \centering
% %\vspace*{\fill}
\subfloat[{\scriptsize TPC-DS: Single running time vs. coverage $L$}]{
\includegraphics[width=0.48\linewidth]{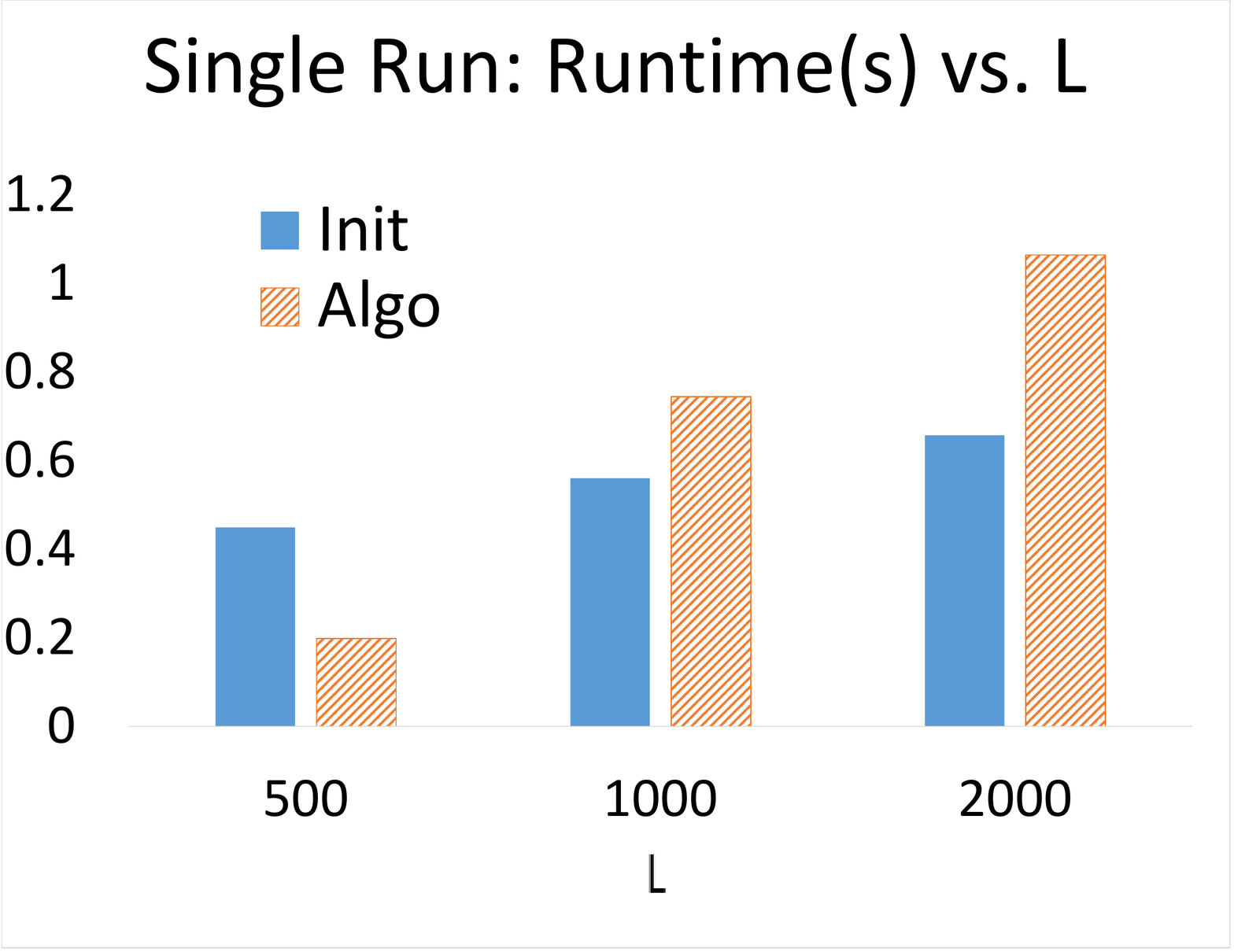}
\label{fig:tpcd-singleVaryL}}~~
%\hspace{0.02\linewidth}
\hfill
\subfloat[{\scriptsize TPC-DS: With precomputation running time vs. coverage $L$}]{
\includegraphics[width=0.48\linewidth]{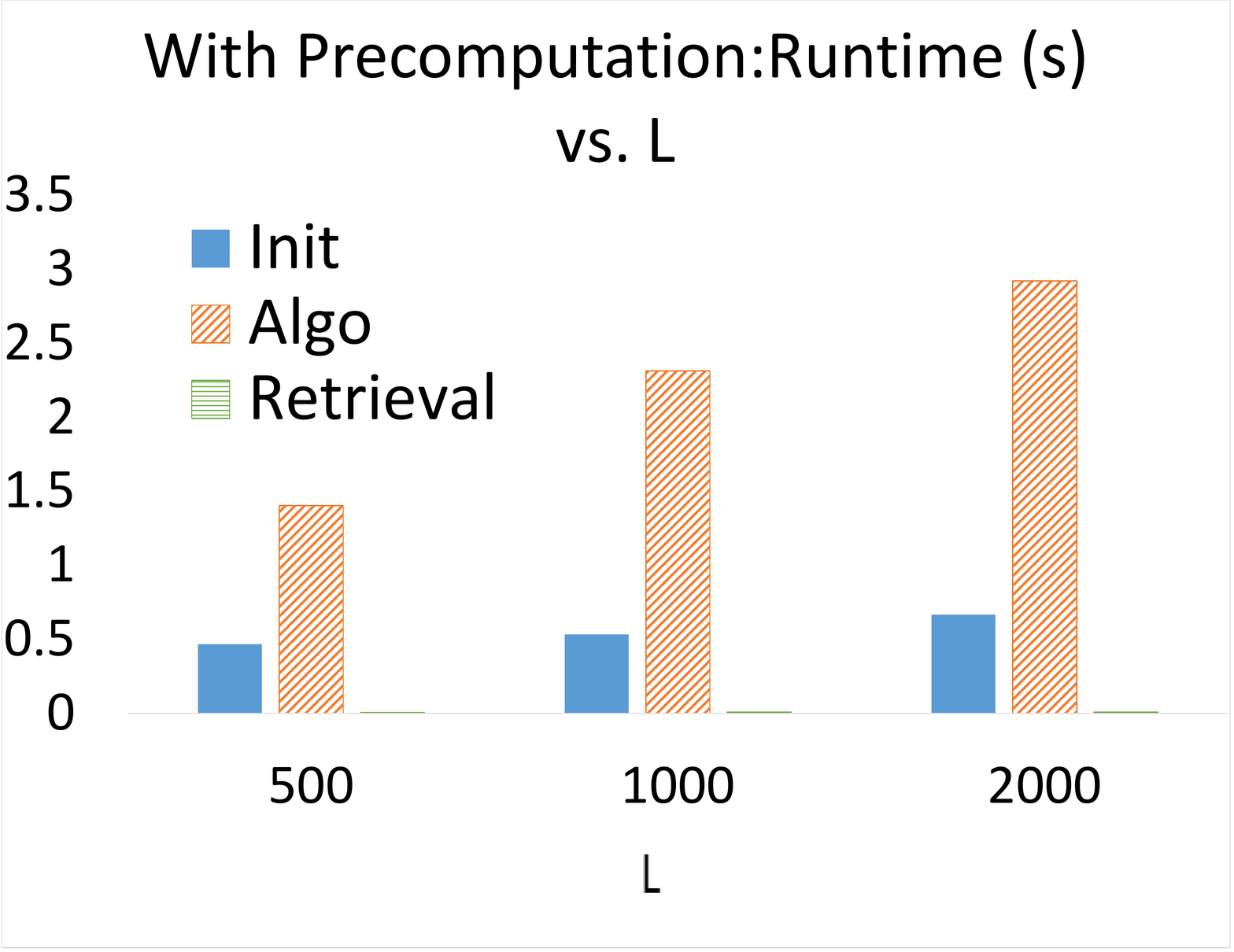}
\label{fig:tpcd-multiVaryL}}
\vspace{-1ex}
%\end{minipage}\hfill
\revb{\caption{TPC-DS experimental results varying parameters and with/without precomputations.}}
\label{fig:tpcdexpts}
%\vspace{-1.5ex}
\end{figure}

%\vspace{-2mm}
\subsection{Scalability with a Larger Dataset}\label{sec:exp-tpcds}

In order to evaluate the scalability of our algorithms %our result 
%on another dataset and test our algorithm with a larger $N$, in this subsection, we carry out 
we perform an experiment with TPC-DS dataset on \emph{Store\_Sales} table. The parameters are set to $k=20, D=2$ and $N=47361$. Coverage parameter $L$ varies among $500,1000$ and $2000$. Both single and precomputation version are evaluated using this set of parameters. From the results shown in Figure~\ref{fig:tpcd-singleVaryL} and Figure~\ref{fig:tpcd-multiVaryL},  the initialization time is interactive - about $1$s for the the largest parameters: $L=2000$ and $N=47361$. However, even for the single version, the running time of the algorithm increases to more than 1s compared with $~200$ms from results in Figures~\ref{fig:pre-singleVaryN} and \ref{fig:pre-multiVaryN},
%Section~\ref{exp-pre-diffn} 
and for the precomputation version it increases to $\sim2.5$s. Although the running time increases, the total running time ($\sim3.5$s) for precomputation is still interactive. Note that the size of the answers ($N$) output by a query is likely to be much smaller than the size of the dataset, even for a big dataset.

\reva{\section{User Study and Survey}\label{sec:user-study}

We conducted a user study with the following high-level goals: (1)~to
compare our approach with an alternative that adapts decision trees~\cite{quinlan1986induction} and (2)~to
evaluate the utility of user-specified parameters in our approach.
Specifically, we want to know: (1)~whether our new problem formulation
provides any advantage over adapting existing methods to the same
usage scenarios; and (2)~whether allowing user-specified parameters in
our problem formulation is warranted in order to capture the range of
different usage scenarios and/or user preferences.  In addition, we
informally solicited feedback during the demonstration of our system
at \emph{SIGMOD} 2018~\cite{qagviewdemo2018}  to assess the effectiveness of our
interactive features in Sections~\ref{sec:guidance}.
%and~\ref{sec:visualization}.

\subsection{User Study Setup}

%DETAILED USER STUDY SETUP HAS BEEN COPIED TO APPENDIX. A SIMPLIFIED VERSION IS PLACED HERE RIGHT NOW. FOR FULL VERSION, DETAILS IN APPENDIX WILL BE HELPFUL.
\textbf{Dataset and queries.~~}
All data are drawn from the \emph{MovieLens} \emph{RatingTable} as
described in Section~\ref{sec:experiments}.  Queries are based on the
same aggregate query template introduced as in Example~\ref{eg:intro}, with an additional
\texttt{WHERE} condition and variations in query constants and
group-by attributes across user tasks. %see~\cite{fullversion} for details.

\textbf{Adapted decision tree.~~}
As discussed in Section~\ref{sec:related}, no existing method suits
our problem setting.  After exploring various possibilities, we
decided to adapt the method of decision
trees~\cite{quinlan1986induction} as it offers the closest match with
our application scenarios.  The structure of a decision tree naturally
induces summaries of top-$L$ tuples in the form of predicates, which
are easier for users to interpret than other classifiers.  It is also
discriminative, as opposed to simply running clustering algorithms
over the top-$L$ tuples while ignoring low-value tuples.\cut{  Finally, it
is possible to control the complexity of the tree.}
We use the standard implementation provided by Python's
\texttt{scikit-learn} package~\cite{pedregosa2011scikit}; we tune the height
of the decision tree such that the number of ``positive''
leaf nodes (wherein top-$L$ tuples are the majority) as close as
possible to, but no greater than, $k$.
Note that the cluster patterns under this approach can be more complex
than ours, as they may involve non-equality comparisons and negations.
This additional complexity increases discrimination, but
makes the patterns more difficult to interpret and
internalize---a hypothesis we shall test with our study.

\textbf{Tasks.~~}
Each study subject is asked to carry out three groups of tasks
(\emph{task groups}): (i) \emph{varying-method}, (ii)
\emph{varying-$k$}, and (iii) \emph{varying-$D$}.  The first group is
designed to compare our approach and decision trees.  The last two are
designed to evaluate the utility of making parameters $k$ and $D$ in
our approach specifiable by users.\footnote{\reva{We do not evaluate
    the utility of making $L$ user-specifiable, as it should be
    evident that what ``top'' tuples mean depends on the
    situation---e.g., a small $L$ means interest in characterizing
    really high-valued tuples, while a larger $L$ means interest in
    tuples whose values are ``good enough.''}}  To account for the
possible learning effect, we sequence the task groups differently
among study subjects---half go through the sequence varying-(method,
$k$, $D$), while the remaining go through varying- ($k$, $D$, method).

Before each task group, we familiarize the subject with the
aggregate query result as well as the tasks; Then, we give the subject a series of questions, organized
into three sections in order.  Each question asks the subject to
classify a given tuple, whose value is hidden, into one of three
categories: ``top'' (value among the top $L$ of all tuples), ``high'' (value above or equal to the average, but outside the top $L$), and ``low'' (value below average).  The three sections are based on the same
``working set'' of clusters, but differ in the information the
subject can access:
\begin{itemize}[leftmargin=*]%
\setlength{\topsep}{0pt}%
\setlength{\partopsep}{0pt}%
\setlength{\itemsep}{0pt}%
\setlength{\parskip}{0pt}%
\setlength{\parsep}{0pt}%
\item \emph{Patterns-only}, 6 questions: The subject can see the clusters and their associated patterns, but not the membership within clusters or the table of all
  query result tuples.  This section is designed to test how well the
  cluster patterns help users understand the data. 
\item \emph{Memory-only}, 6 questions: The subject cannot access any information; all questions must be answered from memory.  This section is designed to test the
  extent to which users can internalize the insights learned from the
  cluster patterns for later use.  We ensure that these six tuples are
  distinct from those chosen before.
\item \emph{Patterns+members}, 8 questions: The subject can see the
  clusters patterns as well as the covered result tuples.  This section is designed to test how our full-fledged
  cluster UI can help user explore data.  The 8 tuples are chosen and
  reordered randomly from the 12 tuples used in the previous two
  sections.
\end{itemize}

After these three sections,  we
present two sets of clusters: one is the working set, the other is obtained under a
different setting (but for the same aggregate query and $L$) for
comparison.  We then ask the subject to choose a preferred set for the tasks just performed. For a \emph{varying-method} task group, the cluster to compare
is produced by decision trees, under the same $k$ setting ($D$ does not apply to
decision trees); For a \emph{varying-$k$} task group, the cluster to compare is
produced by our approach under another $k$, while
other parameters remain the same; For a \emph{varying-$D$} task group, the cluster to compare is
produced by our approach under another $D$.

\textbf{Participants and assignment of tasks.~~}
There are 16 participants - 14 of them are graduate students
at Duke University (12 in computer science and 2 others), while the
remaining 2 are Duke undergraduates. They all have some prior experience working with tabular data and are capable of handling tasks in our user study.

Recall that each of the three task groups compares two sets of
clusters. There are $2^3=8$
possible assignments in total.  We
assign two subjects to each of these $8$ possibilities, each goes through one of the two task group sequences. Finally, we ensure that tuples in our questions are equally distributed among all subjects.

\textbf{Metrics.~~}
We record the time for each subject to complete each of the three
sections in each of the three task groups.  We evaluate the accuracy
of answers using the standard accuracy measure of
$\smash{\frac{TP+TN}{TP+FP+FN+TN}}$ based on confusion
matrices~\cite{fawcett2006introduction}, and we define two variants:
\emph{T-accuracy} focuses on discerning the top tuples from the rest,
where ``positive'' means being in top $L$; \emph{TH-accuracy} focuses
on discerning the top and high tuples from the low ones, where
``positive'' means being in either top or high category.

\subsection{User Study Results}

Table~\ref{tbl:user-study} summarizes both the quantitative results
(subjects' performance in terms of time and accuracy for classifying
tuples into categories) and qualitative results (subjects' preferences
between the clustering outputs compared) of our user study.

\begin{table*}[t]
  \centering\small\setlength\fboxsep{1pt}
  \begin{tabular}{|cr|cc|cc|cc|}
    \hline
    \multirow{2}{*}{} & \multirow{2}{*}{Task group}
        & \multicolumn{2}{c|}{Varying-method} & \multicolumn{2}{c|}{Varying-$k$} & \multicolumn{2}{c|}{Varying-$D$}\\
        & & Decision tree & Our method & $k=5$ & $k=10$ & $D=1$ & $D=3$\\
    \hline
    \multirow{3}{*}{Patterns-only}
        & Time/question & $25.7\pm6.6$ & \fbox{$23.5\pm6.5$} & \fbox{$19.6\pm6.6$} & $22.5\pm5.9$ & $13.5\pm3.1$ & \fbox{$9.6\pm2.9$}\\
        & T-accuracy & $0.792\pm0.121$ &\fbox{$0.854\pm0.132$} & $0.771\pm0.139$ & \fbox{$0.833\pm0.088$} & \fbox{$0.771\pm0.087$} & $0.750\pm0.108$\\
        & TH-accuracy   & $0.646\pm0.075$ & \fbox{$0.854\pm0.075$} & $0.625\pm0.121$ & \fbox{$0.708\pm0.121$} & $0.771\pm0.087$ & \fbox{$0.813\pm0.098$}\\
    \hline
    \multirow{3}{*}{Memory-only}
        & Time/question & $9.6\pm3.8$ & \fbox{$8.3\pm2.9$} & $11.1\pm4.3$ & \fbox{$9.8\pm4.2$} & $8.4\pm2.1$ & \fbox{$6.6\pm3.4$} \\
        & T-accuracy   & $0.625\pm0.150$ & \fbox{$0.792\pm0.121$} & \fbox{$0.771\pm0.139$}
        & $0.667\pm0.125$ & $0.646\pm0.116$ & \fbox{$0.667\pm0.108$}\\
        & TH-accuracy   & $0.625\pm0.174$ & \fbox{$0.792\pm0.121$} & \fbox{$0.667\pm0.108$} & $0.625\pm0.121$ & $0.771\pm0.139$ & \fbox{$0.875\pm0.083$}\\
    \hline
    \multirow{3}{*}{Patterns+members}
        & Time/question & $22.9\pm4.2$ & $23.7\pm2.4$ & \fbox{$21.0\pm6.2$} & $22.5\pm4.2$ & \fbox{$13.5\pm1.9$} & $14.3\pm3.3$\\
        & T-accuracy   & $0.922\pm0.165$ & \fbox{$0.953\pm0.061$} & $0.938\pm0.088$ & $0.953\pm0.061$ & $0.906\pm0.104$ & \fbox{$0.953\pm0.061$}\\
        & TH-accuracy   & $0.750\pm0.088$ & \fbox{$0.844\pm0.104$} & $0.938\pm0.088$ & \fbox{$0.97\pm0.054$} & $0.844\pm0.054$ & \fbox{$0.922\pm0.087$}\\
    \hline
    \multicolumn{2}{|r|}{Overall preferred} & $12.5\%$ & \fbox{$87.5\%$} & $43.8\%$ & \fbox{$56.2\%$} & $37.5\%$& \fbox{$62.5\%$}\\
    \hline
  \end{tabular}
 \vspace{-2.5ex}
  \caption{\label{tbl:user-study}\reva{Summary of results from the user
    study.  Times are in seconds, and accuaries are between $0$ and
    $1$; we report average and standard deviation over all subjects.
    Better performances (shorter times and higher accuracies) and
    stronger preferences are highlighted with box enclosures, unless
    the advantage is too small.}}

\vspace{-4ex}
\end{table*}

%User study table moved in appendix. Should move back in full version.

\textbf{Varying-method task group.~~}
For this task group, we set $L=50,k=10,D=1$ for our approach, and
$L=50,k=10$ for the method based on decision trees.  For this
scenario, tree depth of $7$ gives exactly $10$ positive leaf nodes.

First, note that among the three sections, memory-only is the fastest,
patterns-only is considerably slower, and patterns+members is the
slowest.  This observation holds both for our approach and for
decision trees (as well as under each setting of other task groups).
This universal trend can be intuitively explained by the fact that
users tend to spend more time on a question if more information is
presented to them.

As for accuracy, %there is also a universal trend: 
patterns+members has the highest accuracy, and patterns-only is usually no worse than
memory-only.  This trend also makes intuitive sense as users are
generally able to achieve higher accuracy if aided with more
information.  Across settings, patterns+members is always nearly
perfect, as expected.

Comparing our approach and decision trees in terms of time spent by
study subjects, our approach is consistently faster over the three
sections.  The biggest advantage is seen in the patterns-only section,
suggesting that our patterns are much easier to apply. The advantage is less pronounced in the
other two sections.  For patterns+members, a
possible explanation is that users spend bulk of the time examining
detailed memberships.  For memory-only, our conjecture is that
decision tree patterns are so difficult to recall that our
subjects realized quickly that spending more time did not help.

In terms of accuracy, our approach is better than decision trees for
the patterns-only and memory-only sections (recall that
patterns+members is always nearly perfect across settings).  It is
understandable for decision trees to have lower TH-accuracy, because
they are trained to separate only the top tuples from the rest, while
our approach considers the values of all tuples covered by the
patterns.  On the other hand, while the T-accuracy for decision trees
is good for patterns-only, it drops significantly for memory-only,
because decision tree patterns are difficult for users to memorize.  In
comparison, the accuracy of our approach degrades very little from
patterns-only to memory-only, which is evidence that users can
internalize insights from our simple patterns very well.

Finally, when asked which method they prefer, the overwhelming
majority of the subjects (14/16) chose our approach over decision trees.  The
key reason cited was the simplicity of our patterns.

% Relative results are shown from Fig. \ref{fig:us-comtime} to
% Fig.\ref{fig:us-com2acc}.  Though finding exact tuple takes similar
% time between the two algorithms, our hybird algorithms outperforms
% decision tree in terms of time consumption in all three task
% sections. When it comes to accuracy, users get much better TH-accuracy
% and slightly better T-accuracy in task section 1 compared with
% decision tree. The difference may be caused by decision tree's design
% that seperate top-valued tuples with other tuples, which leads to
% relatively worse capture of high-valued tuples.

\textbf{Varying-$k$ task group.~~}
In this task group, we fix $L=30$ and $D=1$, and compare $k = 5$ vs.\
$k = 10$.  Note that with the bigger $k$, we expect to have more
clusters with more specific patterns, leading to higher discrimination
but more complex summaries.

From Table~\ref{tbl:user-study}, we see that the bigger $k$ leads to
more time spent as long as patterns are accessible to the subjects,
i.e., for patterns-only and patterns+members.  However, for
memory-only, the bigger $k$ actually results in less time spent; one
conjecture is that complex summaries are so difficult to recall from
memory that some subjects simply stopped trying and resorted to
guessing.  This observation is consistent with the low accuracies seen
under the bigger $k$ for memory-only, further discussed below.

In terms of accuracy, favor turns from the smaller $k$ to the bigger
$k$ for patterns-only and patterns+members, pointing to a clear
trade-off between time and accuracy.  On the other hand, for
memory-only, the trend is reversed: accuracies under the bigger $k$
drop dramatically and become lower than under the smaller $k$, because
the subjects had trouble recalling the summaries from their memory.
In comparison, under the smaller $k$, accuracies for memory-only are
at least as good as those for patterns-only.

Finally, when asked whether they prefer the smaller or bigger $k$, a
slight majority of the subjects prefer the bigger, but still a
significant fraction ($7/16$) prefer the smaller.
There is no clear winner here, unlike the case for the varying-method
task group.

\textbf{Varying-$D$ tasks.~~}
We fix $L=10$ and $k=7$, and compare $D=1$ vs.\ $D=3$.  $D=1$
represents a looser constraint, and in this case leads to detailed summaries and higher discriminative power; the trade-off,
of course, is that patterns appear less diverse.

As we can see from Table~\ref{tbl:user-study}, the bigger $D$ leads to
faster answer speed and higher accuracy in most cases, with just two
exceptions: the smaller $D$ is more accurate in terms of T-accuracy
for patterns-only, and it is faster for patterns+members.  Both can be
explained by the fact that, here some clusters produced by
the bigger $D$ happen to have more general patterns and cover more
tuples.  Without access to cluster membership, T-accuracy would suffer
because these clusters may cover some high-valued (but necessarily
top-valued) tuples.  With access to cluster membership, T-accuracy
would not be a problem, but more tuples take longer to examine.

Although the performance results appear to favor the bigger $D$
(looser constraint), preferences are divided.  A
majority of the subjects do prefer the bigger $D$, but still a sizable
number of them ($6/16$) prefer the smaller $D$, which
produces more diverse patterns.

\textbf{Learning effect.~~}
We assess the possible learning effect by comparing the quantitative
result within one experimental sequence (varying-method first, then
varying-$k$ and varying-$D$), and compare with the results in
Table~\ref{tbl:user-study}.  The differences are minor, and the
relative ordering of approaches by performance largely stays the same,
so the conclusions drawn above still stand.  %Because of spaceconstraints, 
Details are shown in Appendix~\ref{sec:app-userstudy}.
%our full paper~\cite{fullversion}.

%COMMENT:EDITED BY YUHAO, BUT STILL NEED ITERATIONS!!
\subsection{Informal User Survey Results}\label{sec:us-sigmod}

To measure the effectiveness of the interactive feature described in
Sections~\ref{sec:guidance}, we asked attendees who visited our demo
booth at \emph{SIGMOD} 2018 to fill out an informal survey.  We
received $18$ responses, and the results are summarized below:

\begin{center}\small
\vspace{-1ex}
  \begin{tabular}{r|c|c|c|c}
    Did you find the
    & Yes, very & Yes & Not that & Not\\
    visualizations helpful?
    & much & & much & at all\\\hline
    For parameter selection & $4$ & $13$ & $1$ & $0$\\
  \end{tabular}
\vspace{-1ex}
\end{center}

The vast majority of the responses are positive.  Some constructive
criticisms were offered too.  One pointed out that the visualization
for guiding interactive parameter selection still required extensive
explanation before users can understand and benefit from it.  Another
pointed out that instead of showing all choices of $k$ and $D$ in this
visualization, it might be possible to use the data behind this
visualization to narrow down the choices further.

\subsection{Summary and Discussion}

The high-level findings are: (1)~our approach is more suitable to the
designed tasks than the decision trees, thanks to the simplicity of
our patterns by design; (2)~while more specific and detailed clusters
can offer better accuracy, this advantage dissipates when users no
longer see the cluster patterns directly, because they are much less
memorable; (3)~parameters $k$ and $D$ affect the complexity of our
clustering results and present various trade-offs (e.g., accuracy vs.\
efficiency), so users have different preferences.

It is also worth noting that while we did not explicitly compare with
the approach of simply showing the top $L$ tuples with no
summarization at all, which can be seen as an extreme case
where $k=L$ and $D=1$.  Hence, the general observation we made when
comparing parameter settings applies here too: showing the top $L$
tuples alone would provide the most detailed information, but that
would be very difficult to use and memorize.

}
% ****************** ONGOING WORK * CONCLUSION **********************************
%\begin{figure}[hb]
%\centering
%\includegraphics[width=0.47\linewidth]{figures/RuntimeVaryD}~~~~
%\includegraphics[width=0.47\linewidth]{figures/ValueVaryD}
%\caption{Effect of distance $D$: (a) Running time, (b) Average value.} %: %(a) data size vs. time, (b) the number of attributes vs. time}
%\vspace{-0.5cm}
%\label{fig:VaryD}
%\end{figure}
%\begin{figure}[h]
%\centering
%\includegraphics[width=0.47\linewidth]{figures/ValueMaxMin}~~~~
%\includegraphics[width=0.47\linewidth]{figures/SizeMaxMin}
%\caption{Comparison of \maxval\ vs. \minsize\ (a) Average value, (b) No. of all elements covered by the clusters.} %: %(a) data size vs. time, (b) the number of attributes vs. time}
%% D = 3
%\vspace{-0.5cm}
%\label{fig:maxmin}
%\end{figure}

\vspace{-0.8mm}

\section{Conclusions}\label{sec:conclusions}
In this paper, we presented a framework for summarization and exploration of high-valued aggregate query answers %. The framework outputs a set of clusters summarizing the attribute values of high-valued answer tuples,  and maintains desired 
maintaining properties like diversity, small size, and coverage of original top elements. 
%The framework facilitates interactive exploration by helping the user choose the input parameters easily, inspect clusters and the tuples they contain, and visualize the effect of changing parameters by comparing two consecutive solutions. 
We studied optimization problems for these tasks, % together with their complexity, explored properties of the solutions, 
developed efficient algorithms and optimizations, evaluated the approach using real and benchmark datasets, and showed that our implementation is capable of interactive exploration in real time. \ansa{We also conducted a user study showing that our constraints are useful and our results are preferred. } % in real application scenario.}
%We plan to refine our algorithms and implementations for a stable version that works consistently on a server for public use. 
\reva{There are several directions for future work. While we mainly focused on categorical attributes, for numeric attributes one can consider other distance functions (\eg, $L_p$ norms) and description of clusters (\eg, ranges). Due to the presence of multiple competing constraints and complex objective function, obtaining formal approximation guarantees will be a challenging theoretical problem.} \revb{One can also consider objective functions other than average and see how our algorithms (that can be adapted to other objectives) perform.}  Other sophisticated visualizations to better help the users can also be explored. % and study the relevant optimization problems further. 

\end{sloppypar}

%\end{document}  % This is where a 'short' article might terminate

% ensure same length columns on last page (might need two sub-sequent latex runs)
%\balance


\begin{sloppypar}

\bibliographystyle{abbrv}
\bibliography{ranking}  % vldb_sample.bib is the name of the Bibliography in this case

%APPENDIX is optional.
% ****************** APPENDIX **************************************
% Example of an appendix; typically would start on a new page
pagebreak

\appendix
\section{Appendix}\label{sec:appendix}

\subsection{Proof of Proposition~\ref{prop: distance:mono}}\label{app:prop:distance:mono}
\begin{proof}[Proof of Proposition~\ref{prop: distance:mono}]
%The proof is similar to that of Proposition~\ref{prop:metric}. 
Consider any cluster $C \in SC$ other than $C_1$. Suppose $d(C_1, C) = \ell$, \ie, $\ell$ of the attributes contribute to the distance function. Fix such an attribute $A$. There are three possibilities. {\bf (1)} $C_1[A] = C[A] = *$. If $C_2[A] = *$, it contributes 1 to $d(C_2, C)$. If $C_2[A] \neq *$, \ie, an attribute value, then also it contributes 1 to $d(C_2, C)$.~~ {\bf (2)} $C_1[A] \neq *$, $C[A] \neq *$, and $C_1[A] \neq C[A]$. If $C_2[A] = *$, it contributes 1 to  $d(C_2, C)$. If $C_2[A] \neq *$, \ie, an attribute value, then it must be the same as $C_1[A]$, and therefore contributes to $d(C_2, C)$.~~ {\bf (3)} $C_1[A] = *$ and $C_1[A] \neq *$. Then $C_2[A] = C_1[A] = *$ and  it contributes 1 to  $d(C_2, C)$. ~~ {\bf (3)} $C_1[A] \neq *$ and $C[A] = *$. Then either $C_2[A] = C_1[A]$ or, $C_2[A] = *$. In either case, it contributes 1 to $d(C_2, C)$. Summing over all $A$ and considering all $C$, $\lambda' \geq \lambda$.
\end{proof}

\subsection{NP-hardness Proofs}\label{sec:app-nphard}
%Proof of Theorem~\ref{thm:opt-NP-hard}}\label{app:thm:opt-NP-hard}
\begin{theorem}\label{thm:opt-NP-hard}
The optimization problem for the objective of \maxval\ 
%and \minsize\ 
for the case when $k \geq L$ and $D > 0$ is NP-hard. 
%(proof in full version \cite{fullversion}). %Appendix~\ref{app:thm:opt-NP-hard})
\end{theorem}
\begin{proof}
The reduction is from the problem of finding a minimum vertex cover in a tri-partite graph %(called the \emph{VC-TPG} problem) 
$G$ with partitions $(X, Y, Z)$, which has shown to be NP-hard in \cite{LlewellynTT93}. The goal in this problem %in  VC-TPG 
is to decide if the input graph has a vertex cover of size $\leq M$, \ie, a subset of vertices $T \subseteq X \cup Y \cup Z$ and  $|S| \leq M$ such that any edge in $G$ has at least one endpoint in $T$\footnote{\cite{GolabKLSS15} gives a reduction from the tri-partite vertex cover problem for \emph{size-constrained weighted set cover} (given weights on the subsets, a size constraint $k$, a coverage fraction $s$, the goal is to return up to $k$ sets that together contain at least $sn$ elements and whose sum of weights is \emph{minimal}). In contrast, in our setting, the weights are assigned on elements (not on subsets); the goal is to select at most $k$ subsets with \emph{maximum} value with the distance and other restrictions,  that cover top-$L$ original  elements; and we show NP-hardness of the decision problem even if the elements are unweighted.}. First we give the proof for {\bf \maxval}. Suppose $G$ has $N_e$ edges and $N_v$ vertices.
\par
Given such a graph $G$ and a bound $M$, we construct a database instance $S$ as follows. There are three attributes $A_X, A_Y, A_Z$.
% and three special values of these attributes $x_0, y_0, z_0$ that do not belong to the set of vertices in $G$. 
(i) \emph{Top-$L$ tuples from the edges of $G$:} Any edge of the form $(x, y), x \in X, y \in Y$ forms two tuples $(x, y, Z^1_{xy})$ and $(x, y, Z^2_{xy})$, where $Z^1_{xy}, Z^2_{xy}$ are two unique values of the $A_Z$ attribute for the edge $(x, y)$ in $G$. Similarly, an edge $(y, z), y\in Y, z \in Z$ forms two tuples $(X^1_{yz}, y, z)$, $(X^2_{yz}, y, z)$, and an edge $(x, z), x \in X, z \in Z$ forms a tuple $(x, Y^1_{xz}, z)$. $(x, Y^2_{xz}, z)$. The weights of these tuples are 1. (ii) \emph{Redundant  tuples from the  vertices of $G$:} For any vertex $x \in A_X$ in $G$, create a redundant tuple $(x, \gamma^2_{x}, \gamma^3_{x})$. These redundant tuples have weight 0. Similarly, form redundant tuples for vertices $y \in A_Y$: $(\gamma^1_{y}, y, \gamma^3_{y})$, and for vertices $z \in A_Z$: $(\gamma^1_{z}, \gamma^2_{z}, z)$ with weight 0. 
(iii) \emph{More redundant  tuples:} For each $Z^1_{xy}$, form $N_r = 2*N_e *N_v$ redundant tuples of the form $(-, -, Z^1_{xy})$ with weight 0, where the positions with $-$ are filled with unique attribute values. Similarly, form redundant $N$ tuples for each of $Z^2_{xy}, Y^1_{xz}, Y^2_{xz}, X^1_{xz}$, and $X^2_{yz}$, placing these attribute values in their corresponding positions.
\par
 We set $k = M$, $L = 2 \times N_e$, $D = 3$. Note that only the tuples from the edges of the form $(x, y, Z^1_{xy})$, $(x, y, Z^2_{xy})$, $(x, Y^1_{xz}, z)$, $(x, Y^2_{xz}, z)$, $(X^1_{yz}, y, z)$, $(X^2_{yz}, y, z)$ form the top-$L$ original elements and have to be covered. We claim that $G$ has a vertex cover of size $\leq M$ if and only if $S$ has a solution, a set of clusters $SC$, of value $\geq \frac{2N_e}{2N_e + M}$, where $N_e$ is the number of edges in $G$.
\par
(only if) Suppose $G$ has a vertex cover $T$ of size $M' \leq M$. For any $x \in X \cap T$, choose the cluster $(x, *, *)$ in $SC$; similarly for $y \in Y \cap T$ and $z \in Z \cap T$, choose the clusters $(*, y, *)$ or $(*, *, z)$ respectively in $SC$. These clusters have mutual distance $= 3$, are incomparable, have size $\leq M$, and cover all top-$L$ elements. Each such cluster also covers a redundant element (with $\gamma$ attribute value) of value 0. Therefore, the value of the solution is $\frac{2N_e \times 1 + M' \times 0}{2N_e + M'} \geq \frac{2N_e}{2N_e + M}$. 
\par
(if) Suppose $S$ has a solution $SC$ of value $\geq \frac{2N_e}{2N_e + M}$. Without loss of generality, any cluster in $SC$ covers at least one of the top-$L$ elements with value 1, otherwise it can be discarded without increasing the value of the average or size/distance/coverage of the solution (the redundant elements have value 0 and can only reduce the average). Also note that the trivial solution $(*,*, *)$ cannot be chosen, since it has value $\frac{2N_e + 0}{2N_e + N_v + 2N_e * N_r} $ $\leq \frac{2N_e}{2N_e + 2N_e * N_r}$ $= \frac{1}{1 + N_r} = \frac{1}{1 + 2N_e N_v}$, which is strictly less than the assumed value of $SC$ $\geq \frac{2N_e}{2N_e + M}$. 
\par
\emph{(A) None of the chosen clusters in $SC$ can be of the form  $(*, *, Z^1_{xy})$ (similarly for $Z^2_{xy}, Y^1_{xz}, Y^2_{xz}, X^1_{xz}$).} Suppose one cluster in $SC$ is $(*, *, Z^1_{xy})$. Then it covers $N_r$ redundant tuples that are not covered by any other cluster in $SC$. Suppose $SC$ has $N'$ redundant tuples all together from all other clusters. Then the average value of $SC$ is $\frac{2N_e + 0}{2N_e + N' + N_r}$ $\leq \frac{2N_e}{2N_e + N_r}$ $= \frac{2N_e}{2N_e + 2N_e N_v}$ = $\frac{1}{1 + N_v}$, which is strictly less than $\frac{2N_e}{2N_e + M}$, the assumed value of $SC$ since $M \leq N_v$.
\par
Next we argue that each cluster in $SC$ can have exactly two $*$-s, combining with (A) above, must be of the form $(x, *, *), (*, y, *), (*,*, z)$.
\par
\emph{(B) The clusters in $SC$ cannot have zero $*$-s.} Suppose without loss of generality that for a top-$L$ tuple  $(x, y, Z^1_{xy})$, the singleton cluster  $(x, y, Z^1_{xy})$ has been chosen in $SC$. Due to the incomparability condition, none of $(x, y, *), (x, *, *), (*, y, *)$ can belong to $SC$. Hence, to cover the other top-$L$ tuple $(x, y, Z^2_{xy})$, one of $(x, y, Z^2_{xy}), (x, *, Z^2_{xy}), (*, y, Z^2_{xy})$ has to be chosen $(*, *, Z^2_{xy})$ cannot be chosen due to (A) above. However, these three clusters have distance 1, 2, 2 respectively from $(x, y, Z^1_{xy})$ violating the distance constraint $D = 3$.
\par
\emph{(C) The clusters in $SC$ cannot have one $*$-s.} (i) Suppose for a top-$L$ tuple  $(x, y, Z^1_{xy})$, the cluster  $(x, y, *)$ has been chosen in $SC$. Due to the incomparability condition, none of $(x, *, *), (*, y, *)$ can belong to $SC$, and any cluster with 1 or zero $*$ to cover top-$L$ tuples from edges of the form $(x, y')$ or $(x', y)$ in $G$ will have distance $\leq 2$ with $(x, y, *)$, violating $D = 3$. If $(x, *, Z^1_{xy})$ is chosen (same for $(*, y, Z^1_{xy})$), to cover the other top-$L$ tuple $(x, y, Z^2_{xy})$, one of $(x, y, Z^2_{xy}), (x, *, Z^2_{xy}), (*, y, Z^2_{xy})$ has to be chosen. Since the first two have distance 1 and 2 respectively from $(x, *, Z^1_{xy})$ violating the distance constraint $D = 3$, $(*, y, Z^2_{xy}$ must also belong to $SC$. However, $(x, *, Z^1_{xy})$ and $(*, y, Z^2_{xy})$ together rule out covering clusters for other top-$L$ tuples from edges of the form $(x, y')$ or $(x', y)$ in $G$, which will have distance $\leq 2$ from either of these two clusters, violating $D = 3$.
\par
Hence the clusters in $SC$ must be of the form $(x, *, *)$, $(*, y, *)$, or $(*, *, z)$, which corresponds to a vertex cover in $G$. Suppose there are $K$ clusters in $SC$. Then the value of $SC$ is $\frac{2N_e + 0}{2N_e + K} \geq \frac{2N_e}{2N_e + M}$ by our assumption, hence $K \leq M$, and $G$ has a vertex cover of size at most $M$.
\cut{
\par
The proof for \textbf{\minsize} is almost the same and simpler, where we argue that $G$ has a vertex cover of size $M$ if and only if $S$ has a feasible solution $SC$ covering $M$ redundant elements. 
}
\end{proof}

%\subsection{Proof of Theorem~\ref{thm:decision-nphard}}\label{app:thm:decision-nphard}
\begin{theorem}\label{thm:decision-nphard}
The decision version of whether a feasible non-trivial solution exists for the problem defined in Definition~\ref{def:problem} is NP-hard,even with three attributes, $D = 0$, $L = n$, and uniform weights of the elements \footnote{The top-$k$ original elements may not constitute an optimal solution for $D = 0$ when $L > k$.}. % Appendix~\ref{app:thm:decision-nphard})
\end{theorem}
\begin{proof}[Proof of Theorem~\ref{thm:decision-nphard}]
The reduction is again from the problem of finding a minimum vertex cover in a tri-partite graph %(called the \emph{VC-TPG} problem) 
$G$ with partitions $(X, Y, Z)$ (see the proof of Theorem~\ref{thm:opt-NP-hard}).
\cut{
 which has shown to be NP-hard in \cite{..}. The goal in this problem %in  VC-TPG 
is to decide if the input graph has a vertex cover of size $\leq M$, \ie, a subset of vertices $T \subseteq X \cup Y \cup Z$ and  $|S| \leq M$ such that any edge in $G$ has at least one endpoint in $T$.
}
\par
Given such a graph $G$ and a bound $M$, we construct a database instance $S$ as follows. There are three attributes $A_X, A_Y, A_Z$.
% and three special values of these attributes $x_0, y_0, z_0$ that do not belong to the set of vertices in $G$. 
Any edge of the form $(x, y), x \in X, y \in Y$ forms a tuple $(x, y, Z_{xy})$, where $Z_{xy}$ is a unique value of the $A_Z$ attribute for the edge $(x, y)$ in $G$. Similarly, an edge $(y, z), y\in Y, z \in Z$ forms a tuple $(X_{yz}, y, z)$, and an edge $(x, z), x \in X, z \in Z$ forms a tuple $(x, Y_{xz}, z)$. The total number of tuples in the relation $S$ is $n$ and each tuple has the same weight $1$. We set $k = M$, $L = $ the number of edges in $G$, and claim that $G$ has a vertex cover of size $\leq M$ if and only if $S$ has a non-trivial solution, a set of clusters $SC$, of size $\leq k$.
\par
(only if) Suppose $G$ has a vertex cover $T$ of size $\leq M$. For any $x \in X \cap T$, choose the cluster $(x, *, *)$ in $SC$; similarly for $y \in Y \cap T$ and $z \in Z \cap T$, choose the clusters $(*, y, *)$ or $(*, *, z)$ respectively in $SC$. Since $T$ is a vertex cover, $SC$ will cover all tuples in $S$ and has size $\leq M = k$.
\par
(if) Suppose $S$ has a non-trivial solution $SC$ of size $\leq k = M$. The clusters in $SC$ can have $*$ in zero, one, or two positions (all three positions cannot be $*$ since $SC$ is a non-trivial solution). We will argue that any cluster in $SC$ can be replaced by a cluster with two $*$ and a vertex from $G$ forming another feasible non-trivial solution without increasing the size of $SC$ (the size may decrease). Consider any tuple of the form $t = (x, y, Z_{xy})$ in $S$ (the other two cases $(x, Y_{xz}, z)$ and $(X_{yz}, y, z)$ follow similarly). If $t$ is covered by a cluster of the form $(x, y, Z_{xy})$, $(x, y, *)$, $(x, *, Z_{xy})$, $(*, y, Z_{xy})$, or, $(*, *, Z_{xy})$, such clusters cannot cover any other tuple in $S$ since $Z_{xy}$ is unique, so replace such clusters by $(x, *, *)$ or $(*, y, *)$. After this is repeated for all tuples in $S$, the only types of clusters remaining in $SC$ be of the form $(x, *, *)$, $(*, y, *)$, or $(*, *, z)$, which corresponds to a vertex cover in $G$, and the size of the solution has not increased.
\end{proof}

\cut{
\subsection{Connection with the Red-Blue Set Cover Problem}\label{app:red-blue}
The \emph{red-blue cover problem} is the following: The elements are marked with either red or blue colors, the goal is to find a collection of sets that cover all blue elements, while covering the minimum possible red elements. Carr et al.\cite{Carr+2000} showed that this problem cannot be approximated within $O(2^{\log^{1-\delta}N})$, where $\delta = 1/\log\log^cN$ for any constant $c < \frac{1}{2}$, where $N$ is the number of subsets (\ie, this problem is strictly harder than the set-cover problem that has a logarithmic approximation). This hardness holds even when every set contains one blue and two red elements. Although this does not give a hardness result for the $\FindClusters{-}\minsize(\ell, k, L)$ procedure due to the restricted semi-lattice structure, which has a similar goal of minimizing the number of redundant elements while covering all top-$L$ elements, this suggests that finding a good approximation is unlikely even for the \FindClusters\ procedure.
}

\subsection{Architecture and GUI}\label{sec:architecture}
%\subsection{Sever Side}
\begin{figure}[t]
\centering
\includegraphics[width=1\linewidth]{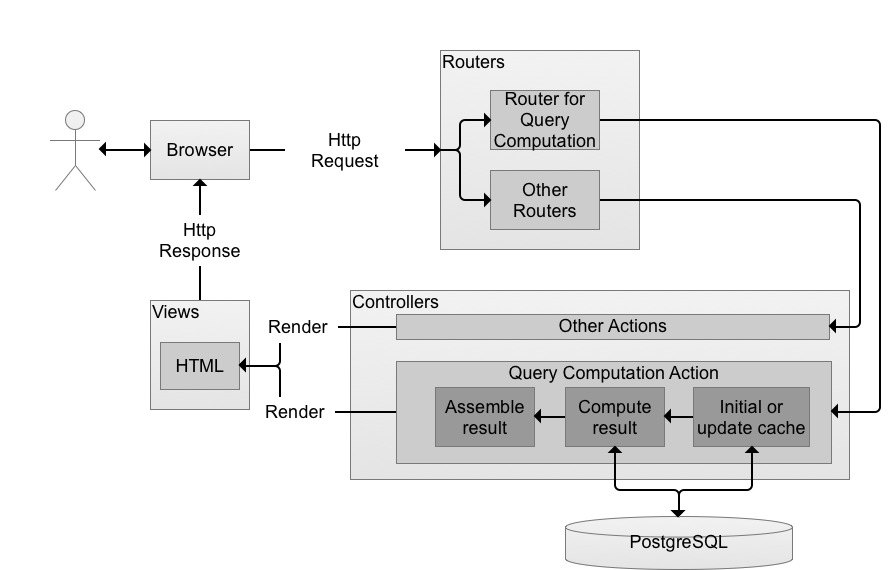}
\caption{Diagram of System Architecture.}
%\vspace{-5mm}
\label{fig:architecture}
\end{figure}
%\begin{figure}[h]
%\centering
%\includegraphics[width=0.9\linewidth]{figures/overview}\\
%\vspace{5mm}
%\includegraphics[width=0.9\linewidth]{figures/gui}
%\caption{A snapshot of the GUI, before and after the query is run and the clusters are displayed.}
%\label{fig:snapshot}
%\end{figure} 

%We developed a prototype with a graphical user interface (GUI) to help users interact with the solutions returned by our two-layered framework. 
%The architecture of the system is shown in Figure~\ref{fig:architecture}. 
%%that let users interactively query database and get cluster results. 
%The prototype is built using Java, Scala, and HTML/CSS/JavaScript as a web application based on Play Framwork 2.4.
%% and %The media used in the communication between 
%%the front-end and the back-end of the system communicate with JSON. 
%The system uses PostgreSQL as the relational DBMS storing the datasets and running the aggregate queries. The user explores the results of an aggregate query using a single-page interface on a web browser. %the part on the right is a working space. 
The architecture of the system is shown in Figure~\ref{fig:architecture}. 
The interface on the browser sends HTTP requests  to the server based on the actions of the %behaviors of users.
user.  For different requests, the server uses the router module to determine which action in the controller module will be performed. After an action processes the request, it  sends the result to the view module. Then the view module partially renders the page and sends it back to the browser.
\par
The interface consists of two parts: a \emph{tool panel} on the left, and a \emph{working space} on the right. When the user opens the GUI, it connects to the DBMS through the router and the controller modules.
The user can see all the databases in the DBMS, and get a preview of any table in any database using the tool panel to have an idea on the schema and data.
%
% and lists names of all the databases stored in the DBMS as a list on the tool panel. The user  selects a database that she wants to work on. Then all the tables in that selected database are listed at the bottom of the left panel. The user can select any table for preview, and the first five rows of that table are displayed on the working space to give the user an idea of the schema. Then the user types an SQL-like query, like the query $Q_k^*$ in Section~\ref{sec:div-cov}, in the input box called \emph{SQL input}. Apart from the {\tt LIMIT $k$} clause, where $k$ stands for the maximum number of output clusters in our framework, the query may contain {\tt COVERAGE $L$} and/or {\tt DISTANCE $D$} clauses. The $k, L, D$ parameters can be fed into the respective boxes in the working space as well. 
Then the user enters an aggregate SQL query (in the query box), the parameters
 %When the user enters a query and the $k, L, D$ parameters 
 $k, L, D$ (either in the query box or in their respective boxes), and clicks on \emph{`go'}.
 % the browser sends a request to the server. After the server processes the request, it sends a HTTP response back to the browser. 
 Then the \emph{SQL result} section  below the SQL input area is updated with the output. 
\par
The action for computing the output clusters has three stages. \emph{First,} it retrieves the original SQL result from PostgreSQL. Then it initializes or updates the cache, which stores some data structures that are frequently used while computing the clusters. % on the next stage. 
If the underlying aggregate query is new, the system will fully update the cache. If only parameters  $k, L$, or $D$ are changed, the system may partially update the cache. The system uses both memory and PostgreSQL as cache. \emph{Second,} based on the input parameters, the system selects an algorithm to compute and optimize the result. \emph{Third,} the system assembles the result as a JSON string and sends it back to the browser. Finally, the SQL result section on the GUI is updated without refreshing this page.
\subsection{Omitted Pseudocodes from Sections~\ref{sec:fixedorder}}\label{sec:app-fixedorder}

Algorithm~\ref{algo:fixedorder} describes details of the \fixedorder\ algorithm introduced in Section~\ref{sec:fixedorder}

\begin{algorithm}[ht]
{
\begin{algorithmic}[1]
\Require Size, coverage, and distance constraints $k, L, D$
%\Ensure Result set $\soln \subseteq \allclusters$, $\left|\soln\right|\leq k$
\State{$\soln = \emptyset$.}
\For{$i = 1$ to $L$}
\State{Let $t_i$ be the $i$-th top element.}
\If{$t \in \cov(\soln)$}
		\State{\emph{/* $t_i$ is already covered, do nothing. */}}
\ElsIf{$|\soln| < k$}
		\If{$t_i$ is at distance $\geq D$ from all clusters in $\soln$}
			\State{$\soln = \soln \cup \{t_i\}$.}
		\Else
			\State{Let $P_D$ be the pairs of clusters $(C, \{t_i\})$ where $C \in \soln$ such that the distance between $C, t_i$ is $< D$.}
			\State{Perform $\updatesolution(\soln, P_D)$.}
		\EndIf
\Else\emph{/* merge $t_i$ with one of the clusters in $\soln$ */}
		\State{Let $P$ be the pairs of clusters $(C, \{t_i\})$, $\forall C \in \soln$.}
		\State{Perform $\updatesolution(\soln, P)$.}
\EndIf
\EndFor
\State\Return $\soln$
\caption{The \fixedorder\ algorithm}
\label{algo:fixedorder}
\end{algorithmic}
}
 \end{algorithm}

%\section{Details from Section 9}\label{app:exp}
\subsection{Qualitative Evaluation with Other Related Approaches}\label{sec:expt-compare-others}
For the query in Example~\ref{eg:intro},  we here compare the results adapting approaches in three related papers \cite{JoglekarGP16, QinYC12, drosou2012disc} that are most relevant to this work (\ie, consider either summarization or diversification). 
Comparison with the MMR-based $\lambda$-parameterized diversification framework (no summarization) from  \cite{Vieira+ICDE2011} is also presented in this section. %can be found in the full version due to space limitation. 
We use the brute-force or an efficient algorithms from these papers using $k = 4$, $D = 2$, and $L = 10$. As expected, the objectives proposed in these papers do not serve the purpose of summarizing aggregate query answers that we study in this paper.
%Since our framework is not directly comparable with \cite{JoglekarGP16, drosou2012disc, QinYC12} for running time evaluation, we compare the results (for \maxval) qualitatively by using the brute-force or an efficient algorithms from these papers using $k = 4$, $D = 2$, and $L = 10$.

\subsubsection[Comparison with smart drill-down]{Comparison with smart drill-down \cite{JoglekarGP16}}\label{expt:smart-drill-down} 
 The goal in \cite{JoglekarGP16} is to find an ordered set of $k$ rules (clusters with $*$ in our framework) $R$ with maximum score  ${\tt score(R)  = \sum_{r \in R} MCount(r, R) \times W(r)}$, where the marginal count ${\tt MCount(r, R)}$ denotes the number of tuples covered by $r$ but not covered by preceding rules in  $R$, and $W(r)$ denotes the number of non-$*$ attributes in $r$. The two parameters focus on diversity (by preferring rules with high ${\tt MCount}$) and good-ness (by preferring more specific rules) of rules. To compare it with our framework with {\em relevance} (summarizes aggregate results where tuples have values), we update the scoring function as ${\tt score(R)  = \sum_{r \in R} MCount(r, R) \times W(r) \times \val(r)}$, where $\val$ denotes the average value of the elements covered by $r$ that are uncovered so far by previous rules. Our framework also considered {\em coverage} of top-$L$ elements, so we evaluated a greedy algorithm (shown to perform well in \cite{JoglekarGP16}) on two inputs: (i) all elements in the aggregate query results, and (ii) top-$L$ elements and obtained the following results:
 
 \begin{center}
 {\centering
{\scriptsize
\begin{tabular}{|c|c|c|c||c|}
\hline
{\tt hdecade} & {\tt agegrp} & {\tt gender} & {\tt occupation} & {\tt \bf avg score}   \\\hline\hline
%& {\tt decade} & {\tt range} & & & \\\hline
\rowcolor{Gray}
\multicolumn{5}{|c|}{{\bf Smart drill-down on top-10 elements}}\\\hline
* 	& 20s 	&  M 	& * 	& 3.47 \\\hline
* 	& 10s 	& M  	& Student 	& 3.63 \\\hline
1995 	& 30s 	& F  	& Educator 	& 3.70 \\\hline\hline
\rowcolor{Gray}
\multicolumn{5}{|c|}{{\bf Smart drill-down on all elements}}\\\hline
1995 	& *  	&  M 	& * 	& 3.18 \\\hline
* 	& 20s 	& M  	& *  	& 3.47 \\\hline
1995 	& * 	& F  	& * 	& 3.24 \\\hline
* 	& 10s 	& M  	& Student 	& 3.63 \\\hline
\end{tabular}\\
}
}
\end{center}

Comparing the above tables with Figure~\ref{fig:eg-intro-clusters} and \ref{fig:eg-intro-two-layer}, we see that the average score of the rules or clusters is much less than our output. Although the above rules capture {\tt (20, M)} that is prevalent in the top-10 results in Figure~\ref{fig:eg-intro-top-bottom}, it is not a characteristics of the top-valued elements; \eg, elements ranked 44, 46, 49 (in total 18 out of 50 tuples) satisfy this rule leading to a low score.
% due to these ``false positives''. 

\subsubsection[Comparison with diversified top-$k$]{Comparison with diversified top-$k$ \cite{QinYC12}}\label{expt:dtopk}
Given the elements with scores, using our terminology, the goal in \cite{QinYC12} is to find at most $k$ elements such that the distance between any two chosen elements is at least $D$ and the sum of scores is maximized. This notion considers {\em diversity} and {\em relevance}. To also include coverage like ours, we ran it on top-$L$ elements and got the following answers (by a brute-force implementation since the goal was qualitative evaluation). We show both actual value of the representative elements ({\tt score}) as well as the average value of elements within distance $D-1$ of these chosen elements ({\tt avg score}) below since intuitively the chosen elements cover these elements.

 \begin{center}
 {\centering
{\scriptsize
\begin{tabular}{|c|c|c|c|c||c|}
\hline
{\tt hdecade} & {\tt agegrp} & {\tt gender} & {\tt occu.} & {\tt score} & {\tt \bf avg score}   \\\hline\hline
\rowcolor{Gray}
\multicolumn{6}{|c|}{{\bf Diversified top-$k$ on top-10 elements}}\\\hline
1975 & 20s 	&  M 	& Student 	& 4.24 & 3.71 \\\hline
1980 	& 20s 	& M  	& Programmer 	& 4.13 & 3.77 \\\hline
1980	& 10s 	& M  	& Student	& 3.97 & 3.69 \\\hline
1995 	& 30s 	& F  	& Educator	& 3.70 & 3.52 \\\hline
\end{tabular}\\
}
}
\end{center}

Although \cite{drosou2012disc} optimizes for a different optimization goal, the above results illustrate that it does not perform well for our goal of summarization of top-valued elements with diversity. {\em First,} the average values of the ``clusters'' formed by the representative elements shown above are less than the values given by our approach in Figure~\ref{fig:eg-intro-clusters}. These chosen values include the original top-10 elements within distance 1, but also include many elements with low values within such circle. For instance, the first element above covers 6 elements (including itself) and covers the 28th element with value 3.12. {\em Second,} the representative of a cluster is one of the original elements, and we are not getting summarized common properties of the cluster using $*$-attribute values. %\red{what else}?

\cut{

1975	20	M	Student	4.23943661971831	avgVal:3.7097565996942751 (5 in total)
1980	20	M	Programmer	4.130434782608695	avgVal:3.7703056061603691 (6 in total)
1980	10	M	Student	3.9649122807017543	avgVal:3.6860677117440098 (5 in total)
1995	30	F	Educator	3.6984126984126986	avgVal:3.5225362725362725 (3 in total)

}
 
\subsubsection[Comparison with DisC diversity]{Comparison with DisC diversity \cite{drosou2012disc}}\label{expt:disc}
Given the distance parameter $D$ and a set of elements $P$, the goal of  \cite{drosou2012disc} is to find a {\em DisC diverse subset} $S^*$ of minimum size such that each element in $P$ is at most distance $D$ from some element in $S^*$, and no two elements in $S^*$ are within distance $D$ of each other.  Like \cite{QinYC12}, this diversity notion naturally includes summarization, since an element can be assigned to the {\em cluster} corresponding to a point in $S^*$ at distance $\leq D$. To also include the notion of {\em relevance} and {\em coverage of top-$L$}, we ran it too on top-$L$ elements (a brute-force implementation) and obtained the following results. Like \cite{QinYC12},  \cite{drosou2012disc} gives clusters with smaller scores than ours (.\eg, the first tuple ``covers'' eight elements where the last one is of rank 28 and has value 3.31), and do not exhibit the common properties by $*$ values. 
%These qualitative comparisons illustrate the effectiveness of our approach. % for summarizing top aggregate query results.

 \begin{center}
 {\centering
{\scriptsize
\begin{tabular}{|c|c|c|c|c||c|}
\hline
{\tt hdecade} & {\tt agegrp} & {\tt gender} & {\tt occu.} & score & {\tt \bf avg score}   \\\hline\hline
\rowcolor{Gray}
\multicolumn{6}{|c|}{{\bf DisC diversity on top-10 elements}}\\\hline
1980 & 20s 	&  M 	& Student 	& 3.91 & 3.81 \\\hline
1985 	& 10s 	& M  	& Student 	& 3.76 & 3.66 \\\hline
1995 	& 30s 	& F  	& Educator 	& 3.70 & 3.52 \\\hline
1985 	& 20s 	& M  	& Engineer	& 3.65 & 3.62 \\\hline
\end{tabular}\\
}
}
\end{center}

\subsubsection[Comparison with MMR-based $\lambda$-parameterized approach]{Comparison with MMR-based $\lambda$-parameterized approach \cite{Vieira+ICDE2011}}\label{expt:lambda}
Running an MMR-based $\lambda$-parameterized approach in \cite{Vieira+ICDE2011}, we obtain the following results for different values of $\lambda$. Note that this is a result diversification problem and not a result summarization problem, therefore does not include a coverage or summary of the top answers, or an average score. Below, higher $\lambda$ denotes higher diversity. 

 \begin{center}
 {\centering
{\scriptsize
\begin{tabular}{|c|c|c|c|c|}
\hline
{\tt hdecade} & {\tt agegrp} & {\tt gender} & {\tt occu.} & {\tt score}    \\\hline\hline
\rowcolor{Gray}
\multicolumn{5}{|c|}{{\bf $\lambda = 0$}}\\\hline
1975	& 20s & 	M	& Student& 	4.24 \\\hline
1980	& 20s	& M	& Programmer& 	4.13  \\\hline
1980    & 10s   & M& 	Student	& 3.96  \\\hline
1980 & 	20s& 	M	& Student& 	3.91  \\\hline
\rowcolor{Gray}
\multicolumn{5}{|c|}{{\bf $\lambda = 0.2, 0.5, 0.8$}}\\\hline
1975	& 20s& 	M& 	Student& 	4.24  \\\hline	
1980	& 20s& 	M& 	Programmer	& 4.13  \\\hline
1980	& 10s& 	M	& Student	& 3.96  \\\hline
1995& 	30s& 	F	& Educator& 	3.70  \\\hline
\rowcolor{Gray}
\multicolumn{5}{|c|}{{\bf $\lambda = 1.0$}}\\\hline
1985	& 20s	& M	& Programmer	& 3.86  \\\hline	
1980	& 20s	& M	& Engineer	& 3.83  \\\hline
1985	& 10s	& M	& Student	& 3.76  \\\hline
1995	& 30s	& F & 	Educator	& 3.70  \\\hline
\end{tabular}
}
}
\end{center}

\begin{figure}[ht]
\centering
%\includegraphics[width=0.9\linewidth]{figures/overview}\\
%\vspace{5mm}
\includegraphics[width=0.7\linewidth]{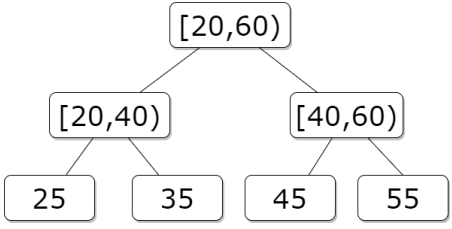}
\caption{Range tree example on age}
%\vspace{-0.5cm}
\label{fig:rangetreeage}
\end{figure} 

\begin{figure}[ht]
\centering
%\includegraphics[width=0.9\linewidth]{figures/overview}\\
%\vspace{5mm}
\includegraphics[width=0.9\linewidth]{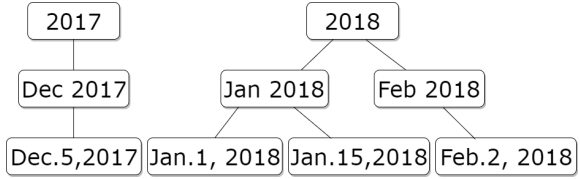}
\caption{Range tree example on date}
%\vspace{-0.5cm}
\label{fig:rangetreedate}
\end{figure} 

\begin{figure}[h!]
\includegraphics[width=0.9\linewidth]{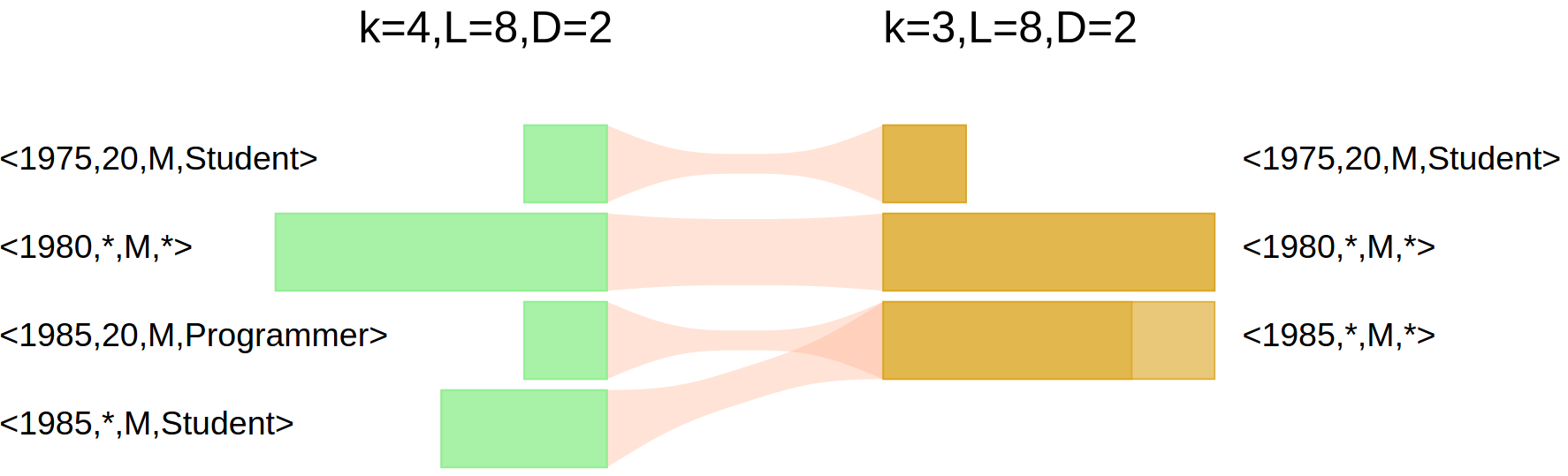}
%\vspace{-2ex}
\caption{Visualizing changes between two consecutive clustering in Example~\ref{eg:intro-clusters}: from $k=4$ to $k=3$.}
%\vspace{-0.5cm}
\label{fig:intro-vizcompare}
\end{figure} 

\subsection{Extension for range values}\label{app:range}
We introduced \emph{don't care} value ($*$) in Section~\ref{sec:prelim}. However, for numerical attributes (such as age) and text attributes (such as date) with given hierarchies,  ranges values can be an interesting option to present. Aiming at this goal, we set up a tree structure to represent the given hierarchy. Examples of trees for numerical attributes and text attributes are shown in Figure~\ref{fig:rangetreeage} and Figure~\ref{fig:rangetreedate}. Individual values of the attribute serve as leaves, and possible ranges serve as nodes of the tree structure. 
\par
To get the proper range required for building and calculating clusters, we just need to get the LCA of candidate leaves and nodes in the tree structure. The LCA node is unique for certain candidates and can be found directly in the tree. For example, for the example given in Figure~\ref{fig:rangetreeage}, to get the union of $[20,40)$ and $55$, we need to find the LCA node of them, which is $[20,60)$. A $log(n)$ algorithm\cite{harel1984fast} is available to find the LCA node, where $n$ is the total number of nodes. 
\par
Our framework can currently handle a given concept hierarchy in the form of a tree. How to automatically build such concept hierarchies will be an orthogonal future research direction.

\ansc{
\subsection{Details for Visualizing Successive Changes}\label{sec:visualization}
Figure~\ref{fig_sim} %and Figure~\ref{fig:newgui} 
shows how a user can inspect the clusters and the elements they contain in the form of tables in our framework.
However, for interactive exploration, it is also important that the user can see how the old solution changes to a new solution (if $k,  D, $ or $L$ is updated).
When the user explores the solution space updating an input parameter, in some scenarios the solution can change marginally, whereas in others it can change drastically. To help the user understand how two consecutive solutions compare with each other, our framework produces a visualization showing the old and new solution, and how the tuples in these clusters are redistributed (also the size of the cluster, the fraction of top-$L$ tuples contained in them, etc.). For example, Figure~\ref{fig:intro-vizcompare} shows that  if $k = 4$ in Example~\ref{eg:intro-clusters} is changed to  $k = 3$, then two of the clusters will merge to form the new solution. %(Section~\ref{sec:visualization}). 

\cut{
\begin{itemize}[leftmargin=*]
\itemsep0em
\item To achieve a \emph{clean} visualization, we formulate an optimization problem that minimizes the crossing of the bands showing flow of tuples from old to new clusters.
\item We give an optimal poly-time algorithm by reducing this problem to min-cost perfect matching in bipartite graphs. 
\end{itemize}
}

%Original viz section paragraph. The new one is merged from intro texts with this one.
%Figure~\ref{fig_sim} %and Figure~\ref{fig:newgui}  shows how a user can inspect the clusters and the elements they contain in the form of tables in our framework. However, for interactive exploration, it is also important that the user can see how the old solution changes to a new solution (if $k,  D, $ or $L$ is updated), since in some cases the change can be incremental, whereas in other cases, changing a parameter may give a complete new solution. 
To support this comparison, we display the successive solutions and their overlaps (Section~\ref{sec:comparison_gui}), and formulate it as an optimization problem for a clean display (Section~\ref{sec:comparison_opt}).

\subsubsection{Visualizing Changes}\label{sec:comparison_gui}
An example visualization is shown in Figure~\ref{fig:viz-clean}. Each box corresponds to a cluster in the solution. The left hand side boxes (in green) are result clusters for the previous run, and the right hand side boxes (in yellow) are clusters under the current parameters. The width of each box is proportional to  the number of tuples contained in the cluster. Clusters connected with bands or ribbons contain shared tuples. The thicker a band is in the middle, the more common tuples it contains \footnote{This visualization shows some similarity with SANKEY diagrams that are widely used in the field of energy and material flow management\cite{schmidt2008sankey}. 
% In some sense our visualization can be considered as an variation of SANKEY diagram since both 
% and provide flow from one set of objects to another. 
However, popular SANKEY diagram libraries (\eg, d3-sankey) focus more on managing placement among columns (horizontal positioning). For vertical positioning, they do multiple iterations to re-position objects to achieve a satisfying visualization. We do not need to consider multiple columns, and our visualization focuses on vertical positioning and ordering of objects}. The parts of the boxes in darker color correspond to the fraction of top-$L$ tuples contained in these clusters. Hovering the mouse over different regions shows details.

\begin{figure}[h!]
\includegraphics[width=0.9\linewidth]{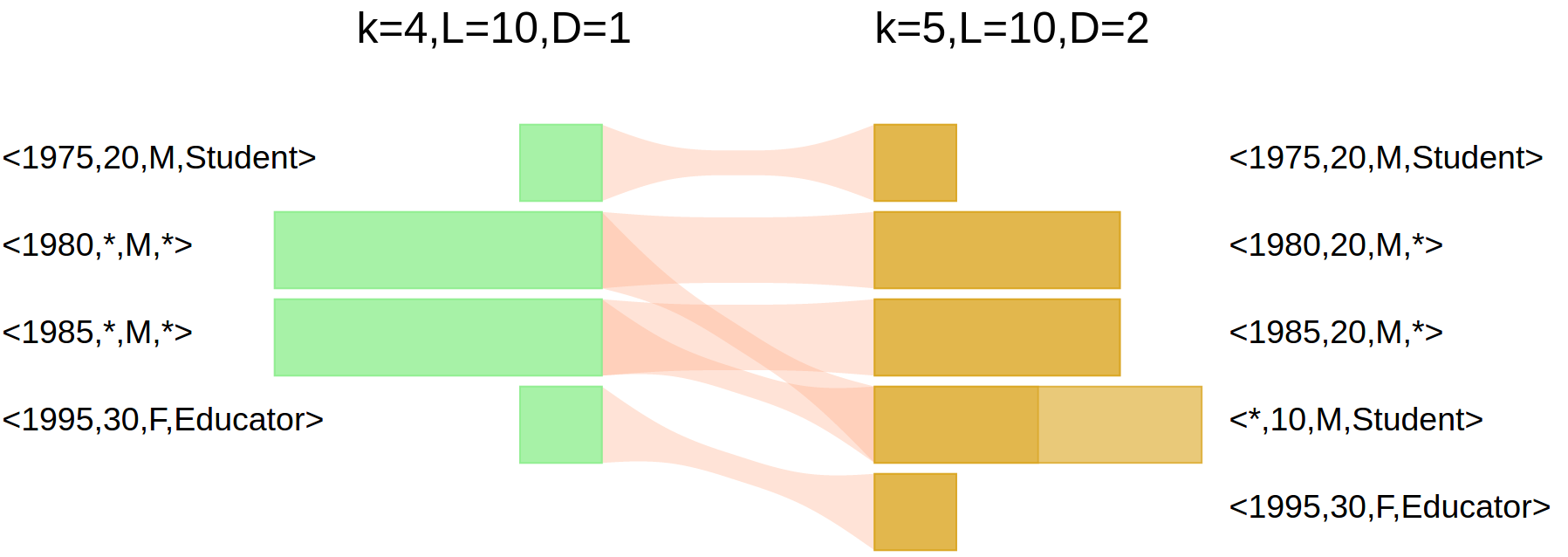}
%\vspace{-0.2cm}
\caption{\ansc{Solution comparison: clean visualization}}
%\vspace{-0.3cm}
\label{fig:viz-clean}
\end{figure} 

\begin{figure}[h!]
\includegraphics[width=0.9\linewidth]{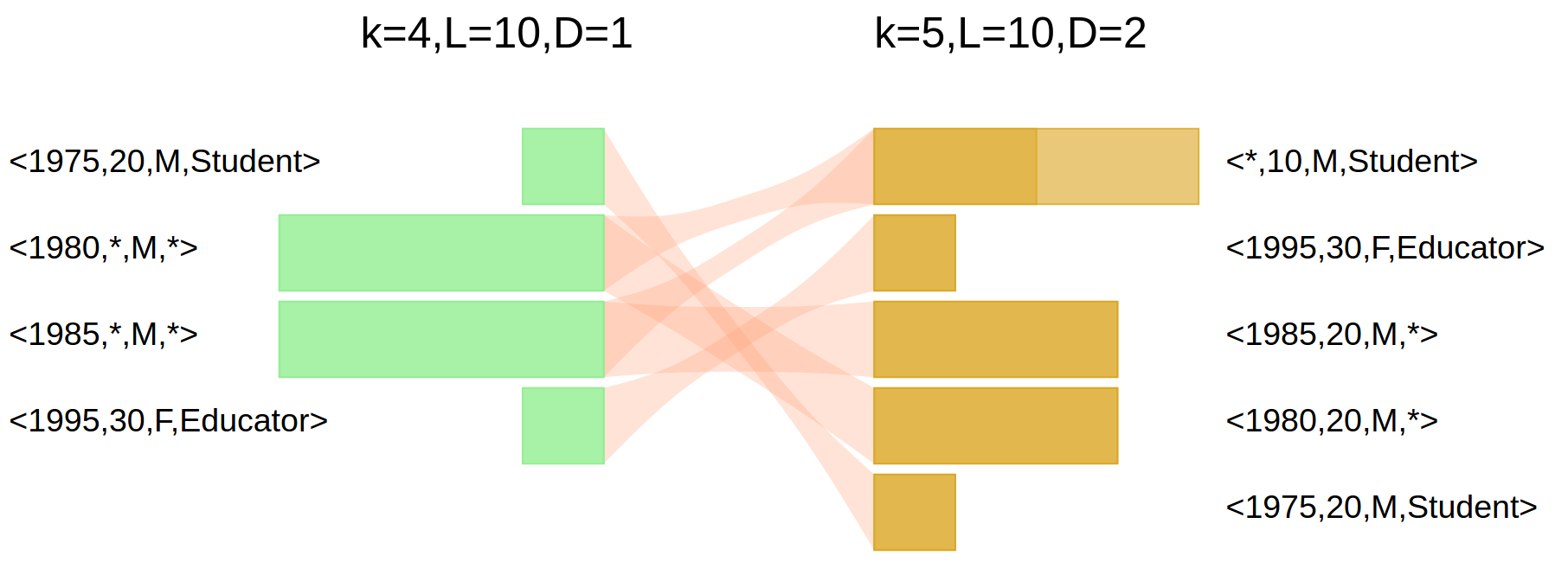}
%\vspace{-0.2cm}
\caption{\ansc{Solution comparison: cluttered visualization}}
%\vspace{-0.3cm}
\label{fig:viz-chaos}
\end{figure}

\subsubsection{Optimizing Cluster Placement}\label{sec:comparison_opt}
A careful ordering of the cluster boxes helps in displaying a cleaner visualization comparing two consecutive executions. For instance, Figure~\ref{fig:viz-chaos} is an example where placement of the clusters (boxes) leads to more \emph{crossing} of the bands, whereas  Figure~\ref{fig:viz-clean} shows a better placement. The placement can have more effect when the number of clusters (the value of $k$) is larger. Therefore, we formulate the cluster placement problem as an optimization problem that intends to minimize crossing of the bands as follows.
%As a result, it is meaningful to align the clusters carefully for formulating as a placement problem in the next subsection. 

%\subsubsection{Definition}
\textbf{Optimization problem.~~} Let $\soln_a$ and $\soln_b$ denote the old and new set of clusters respectively in two consecutive runs containing clusters $\soln_a = \big\{ c_{a1},c_{a2},...,c_{am} \big\}$ and $\soln_b = \big\{ c_{b1},c_{b2},...,c_{bn} \big\}$. Let $m_{ij}$ denote the number of shared tuples between clusters $c_{ai}$ and $c_{bj}$, and $M$ denote the set of all such $m_{ij}$ values.
\cut{
The number of tuples in each cluster in $S_a$ is $L_a=\big\{ l_{a1},l_{a2},...,l_{am}\big\}$ and $L_b=\big\{ l_{b1},l_{b2},...,l_{bn} \big\}$.
The mapping between two cluster sets is built in the form of $m\times n$ matrix $M$, where $m_{ij}$ shows the number of shared tuples between cluster $c_{ai}$ and cluster $c_{bj}$. 
}
\cut{
The visualization has two columns of boxes representing two cluster sets - previous set on the left and new set on the right -  with bands from left to right connecting both sides. All clusters (boxes) share the same unit height $1$.
} 
We assume that the first clusters in both sides are placed at the same vertical position, and there are no gaps or overlaps  between two adjacent clusters on either side. Consider two orderings of clusters on the left hand side and right hand side respectively  in terms of their starting positions $P_a = {p_{a1},p_{a2},..,p_{am}}$ (a permutation of $[0, m-1]$), and $P_b = {p_{b1},p_{b2},...,p_{bn}} $ (a permutation of $[0, n-1]$). We define a weighted \emph{earth mover's distance}\cite{rubner2000earth} $d_{ij}$ between one left cluster $c_{ai}$ and one right cluster $c_{bj}$ to evaluate the amount of crossing due to a single band (from $c_{ai}$ to $c_{bj}$) as:
\[ d_{ij} = m_{ij}\times|p_{ai}-p_{bj}|\]
Since $\soln_a$ is the previous cluster set, $P_a$ is fixed and given. The goal of the optimization problem for this visualization is to output a good ordering $P_b$ for $\soln_b$, and is formulated as follows:
%The value of $p_{ai}$ denotes the starting vertical position of $c_{ai}$ in $S_a$ and $p_{bj}$ denotes the position of $c_{bj}$ in $\soln_b$. 

\cut{
For the ordering, we have
\[ \forall  o_{ai}\in O_a, o_{bj} \in O_b, o_{ai}\in N^0, o_{bi}\in N^0\]
\[ \forall o_{ai}\in O_a, \forall o_{bj} \in O_b, 0\leq o_{ai}<m, 0\leq o_{bj}<n\]
\[ \forall i'\neq i, o_{ai}\neq o_{ai'}\]
\[ \forall j'\neq j, o_{bj}\neq o_{bj'}\]
}

\begin{definition}\label{def:opt-place}
\textbf{Optimization for placement of clusters.~~} Given old and new clusters $\soln_a$, $\soln_b$, %$L_a$, $L_b$, 
their overlaps $M$, and ordering $P_a$ of the clusters on the left hand side, find an ordering $P_b$ of the clusters on the right hand side that minimizes
$D=\sum_{i=1}^{m}\sum_{j=1}^{n}d_{ij}$.
\end{definition}

\textbf{Optimal solution using bipartite matching.~}
The above optimization problem can be reduced to the minimum cost perfect matching problem in a complete bipartite graph as follows. We form a weighted complete bipartite graph $G(U \cup V, E)$, where the $n$ nodes in $U$ correspond to the $n$ clusters in $\soln_b$, and the $n$ nodes in $V$ correspond to the positions $1 \cdots n$. An edge $(u, v)$ denotes the possibility when cluster $c_{bu}$ is placed in position $v \in [1, n]$. The weight of the edge $(u, v)$ is the cost $\sum_{i=1}^md_{iu}$ = $\sum_{i=1}^m (m_{iu} \times |p_{ai} - (v-1)|)$ (if $c_{bu}$ is placed in position $v$, there will be $v-1$ clusters before it, and its position will be $v-1$), \ie, the total contribution of cluster $c_{bu}$ in the optimization objective in Definition~\ref{def:opt-place} if it is placed in position $v \in [1, n]$. Since $\soln_a$ and $P_a$ are given and fixed, this weight can be computed in polynomial time as a precomputation step for each cluster in $\soln_b$ and each position. A matching gives a positioning $P_b$ on the clusters in $\soln_b$, and the minimum cost perfect matching, which has a polynomial time algorithm\cite{geomans2009}, gives an optimal solution to our optimization problem. %The solution is accurate and the approximation ratio is 1.
\par
We also studied an alternative formulation of the above optimization problem, where instead of the \emph{width} of the boxes being proportional to the number of tuples in clusters, the \emph{height} is proportional to the number of tuples, \ie, the starting positions $P_a$ and $P_b$ are no longer permutations of $[0, m-1], [0, n-1]$ but also depend on the height of the individual clusters. However, we found that this variant is NP-hard by a reduction from the \emph{earliness-tardiness job scheduling problem}\cite{garey1988one}. The proof and details of this alternative formulation are deferred to the extended version of this paper.

\subsubsection{Performance of Comparison Visualization}
\par
 We tested the running time for calculating and generating the visualization for both the applied algorithm~\cite{geomans2009} and the brute-force algorithm under $k=10, L=15,20$ and $D=2$ in Movielens dataset with $N=2087$. In this test, both algorithms have similar figure drawing time (~20ms) since they have identical data for figure generation (both of them get the optimal answer), but difference between the calculation time is enormous---the bipartite matching algorithm takes less then 10ms while brute-force takes more than 2s.
\par
The quality of visualizations produced by bipartite matching is shown in Figure~\ref{fig:vizexpts-distancecompare} and Figure~\ref{fig:vizexpts-crosscompare}, as the ``matched visualization,'' in comparison to the ``default visualization.''  For the default visualization, we use the sequences of clusters as returned by successive runs of the clustering algorithm, where clusters for both sides are ordered by value. Parameter sets are $D=2, (k,(L_1,L_2))=(5,(8,10)), (10,(15,20))$ and $(20,(30,40))$ where $L_1$ and $L_2$ are $L$s for the two answer sets.  Figure~\ref{fig:vizexpts-distancecompare} shows that bipartite matching is very effective in reducing the ``clutter'' in visualization, as measured by our distance metric in Section~\ref{sec:comparison_opt} (note that the distances are generally not comparable among different $k$s).  In addition to this metric, we also counted the number of crossings among bands (connections between left and right clusters) and plot the result in Figure~\ref{fig:vizexpts-crosscompare}.  It is clear that our approach also succeeds in cutting down the amount of crossings. 

\begin{figure}[t]
\centering
\begin{minipage}[t]{.47\textwidth}
  \centering
 %\vspace*{\fill}
\subfloat[{\scriptsize Total distance between matched visualizations and default visualizations}]{
\includegraphics[width=0.42\linewidth]{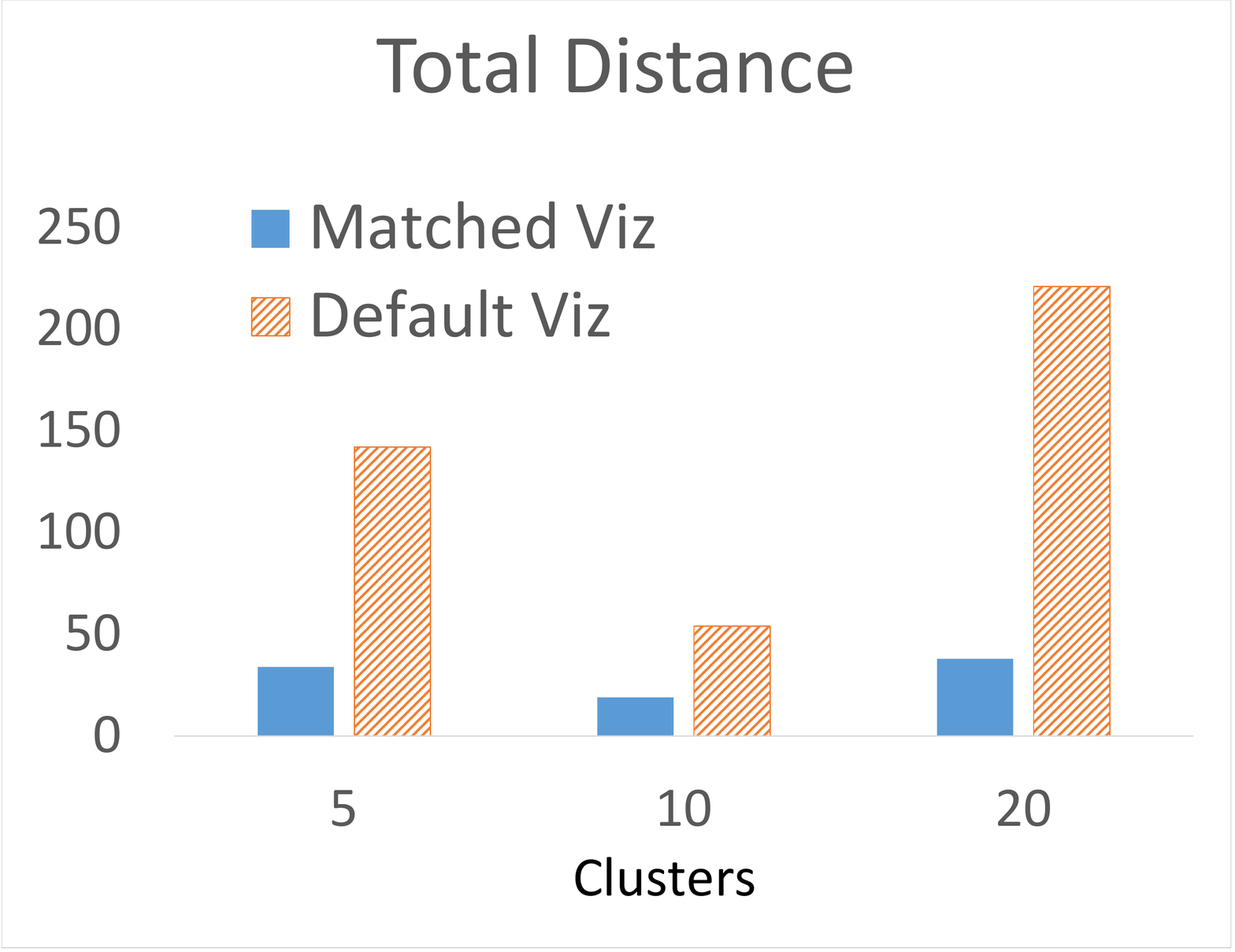}~~~~
\label{fig:vizexpts-distancecompare}}
\hspace{0.02\linewidth}
\subfloat[{\scriptsize Amount of crossings among bands between matched visualizations and default visualizations}]{
\includegraphics[width=0.42\linewidth]{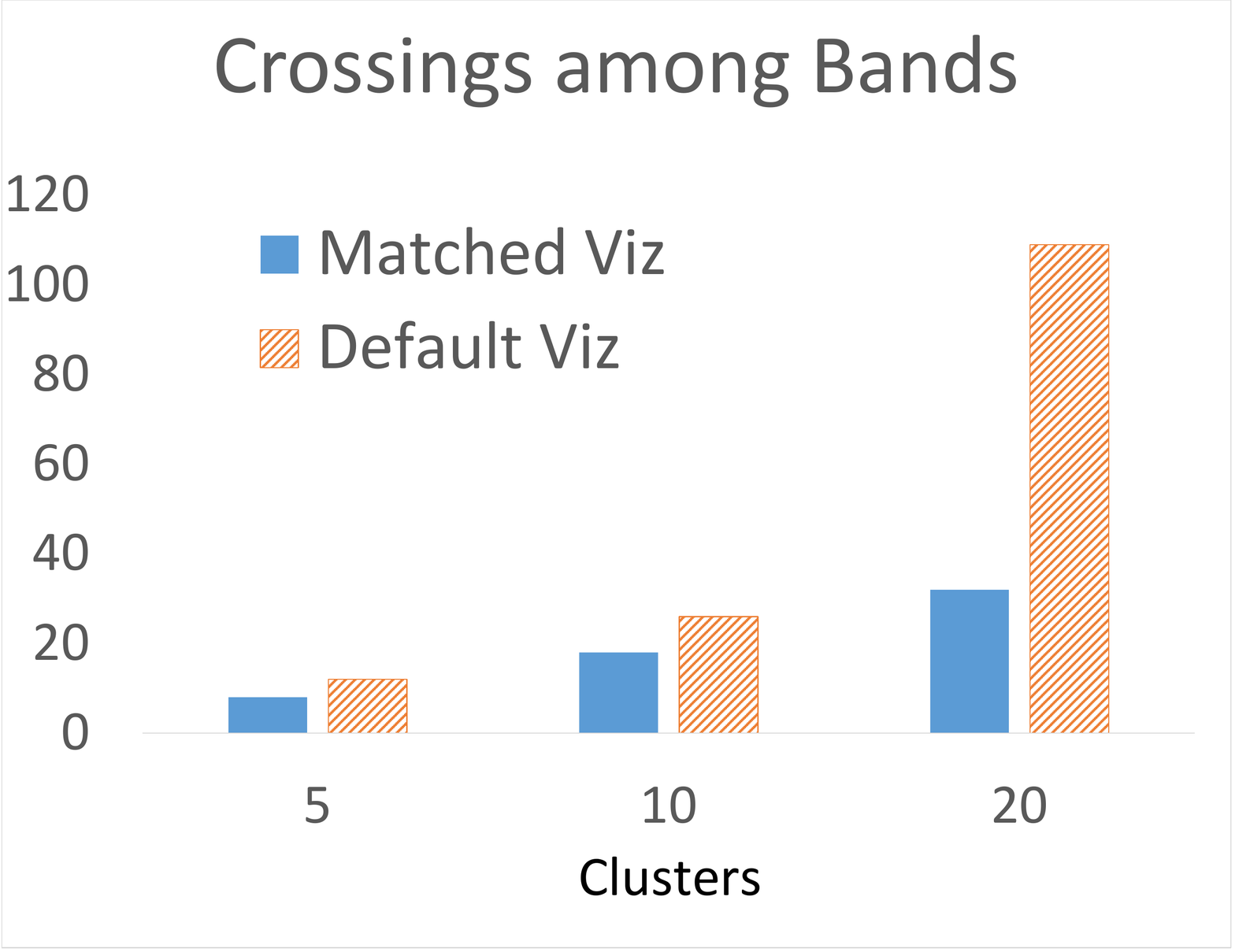}~~~~
\label{fig:vizexpts-crosscompare}}
\vspace{-0.3cm}
\end{minipage}\hfill
\caption{Experiment on Comparison Visualizations}
\vspace{-0.2cm}
\label{fig:vizexpts}
\end{figure}

\subsubsection{Informal User Study}

As part of the informal user study we conducted at SIGMOD 2018 mentioned in Section~\ref{sec:user-study}, the satisfactory towards this comparison view is shown as below.
\begin{center}\small
  \begin{tabular}{r|c|c|c|c}
    Did you find the
    & Yes, very & Yes & Not that & Not\\
    visualizations helpful?
    & much & & much & at all\\\hline
    For parameter selection & $4$ & $13$ & $1$ & $0$\\
    For comparing old/new clusters & $7$ & $11$ & $0$ & $0$
  \end{tabular}
\end{center}

It suggests that the comparison view is helpful and useful in real life scenario. One constructive suggestion upon this view is that it would be better if the user can view tuples inside each cluster and band after clicking.

}

\begin{table*}
  \centering\small\setlength\fboxsep{1pt}
  \begin{tabular}{|cr|cc|cc|cc|}
    \hline
    \multirow{2}{*}{} & \multirow{2}{*}{Task group}
        & \multicolumn{2}{c|}{Varying-method} & \multicolumn{2}{c|}{Varying-$k$} & \multicolumn{2}{c|}{Varying-$D$}\\
        & & Decision tree & Our method & $k=5$ & $k=10$ & $D=1$ & $D=3$\\
    \hline
    \multirow{3}{*}{Patterns-only}
        & Time/question & $28.8\pm7.0$ & $28.2\pm4.9$ & \fbox{$19.9\pm8.1$} & $24.8\pm7.0$ & $11.8\pm2.5$ & \fbox{$9.7\pm3.4$}\\
        & T-accuracy & $0.833\pm0.204$ &$0.833\pm0.118$ & $0.750\pm0.144$ & \fbox{$0.833\pm0.118$} & $0.750\pm0.083$ & $0.792\pm0.138$\\
        & TH-accuracy   & $0.708\pm0.072$ & \fbox{$0.917\pm0.083$} & $0.583\pm0.144$ & \fbox{$0.750\pm0.186$} & $0.792\pm0.072$ & $0.833\pm0.167$\\
    \hline
    \multirow{3}{*}{Memory-only}
        & Time/question & $9.8\pm3.0$ & $10.3\pm2.1$ & $13.8\pm4.2$ & \fbox{$11.5\pm5.3$} & $7.5\pm0.3$ & \fbox{$5.5\pm3.9$} \\
        & T-accuracy   & $0.542\pm0.072$ & \fbox{$0.708\pm0.138$} & \fbox{$0.667\pm0.204$}
        & $0.583\pm0.186$ & $0.667\pm0.204$ & $0.625\pm0.138$\\
        & TH-accuracy   & $0.625\pm0.217$ & \fbox{$0.917\pm0.083$} & $0.708\pm0.138$ & $0.667\pm0.138$ & $0.792\pm0.217$ & \fbox{$0.875\pm0.138$}\\
    \hline
    \multirow{3}{*}{Patterns+members}
        & Time/question & \fbox{$23.6\pm5.4$} & $25.3\pm1.6$ & \fbox{$17.1\pm3.1$} & $31.7\pm4.4$ & $13.5\pm2.4$ & $13.2\pm1.9$\\
        & T-accuracy   & $0.875\pm0.217$ & \fbox{$0.938\pm0.063$} & $0.969\pm0.054$ & $0.969\pm0.054$ & $0.906\pm0.104$ & $0.938\pm0.063$\\
        & TH-accuracy   & $0.750\pm0.088$ & \fbox{$0.844\pm0.136$} & $0.938\pm0.063$ & $0.97\pm0.054$ & $0.844\pm0.054$ & \fbox{$0.969\pm0.054$}\\
    \hline
    \multicolumn{2}{|r|}{Overall preferred} & $12.5\%$ & \fbox{$87.5\%$} & $43.8\%$ & \fbox{$56.2\%$} & $37.5\%$& \fbox{$62.5\%$}\\
    \hline
  \end{tabular}
  \caption{\label{tbl:user-study-vary-order}Summary of results from the user
    study when varying-method group goes first.  Times are in seconds, and accuracies are between $0$ and
    $1$; we report average and standard deviation over all subjects.
    Better performances (shorter times and higher accuracies) and
    stronger preferences are highlighted with box enclosures, unless
    the advantage is too small.}
\end{table*}

\subsection{Aggregate Queries Used in Experiments}~\label{app:exp-queries}
The aggregate queries we use for the MovieLens dataset has the following form: %related experiments
\newsavebox\sqlexp
\begin{lrbox}{\sqlexp}\begin{minipage}{\textwidth}
\lstset{language=SQL, basicstyle=\ttfamily, deletekeywords={year,month,action},tabsize=2}
\begin{lstlisting}[mathescape]
SELECT $\langle$ grouping attributes$\rangle$, avg(rating) as $\val$ 
FROM RatingTable 
GROUP BY $\langle$ grouping attributes$\rangle$
HAVING count(*) > 50 
ORDER BY $\val$ DESC
\end{lstlisting}
\end{minipage}\end{lrbox}
\resizebox{0.85\textwidth}{!}{\usebox\sqlexp}

The aggregate queries we use for TPC-DS related experiments %are shown as follows.
share the same form with the example given above.

\begin{lrbox}{\sqlexp}\begin{minipage}{\textwidth}
\lstset{language=SQL, basicstyle=\ttfamily, deletekeywords={year,month,action},tabsize=2}
\begin{lstlisting}[mathescape]
SELECT $\langle$ grouping attributes$\rangle$, cast(avg(net_profit) 
        as int) as $\val$ 
FROM store_sales
GROUP BY $\langle$ grouping attributes$\rangle$
HAVING count(*) > 10 
ORDER BY $\val$ DESC
\end{lstlisting}
\end{minipage}\end{lrbox}
\resizebox{0.85\textwidth}{!}{\usebox\sqlexp}

\ansa{
\subsection{Detailed User Study Setup}\label{sec:app-userstudy}

\textbf{Dataset and queries.~~}
All data are drawn from the \emph{MovieLens} \emph{RatingTable} as
described in Section~\ref{sec:experiments}.  Queries are based on the
same aggregate query template introduced therein, with an additional
\texttt{WHERE} condition and variations in query constants and
group-by attributes across user tasks; %see~\cite{fullversion} for details.

\textbf{Adapted decision tree.~~}
As discussed in Section~\ref{sec:related}, no existing method suits
our problem setting.  After exploring various possibilities, we
decided to adapt the method of decision
trees~\cite{quinlan1986induction} as it offers the closest match with
our application scenarios.  The structure of a decision tree naturally
induces summaries of top-$L$ tuples in the form of predicates, which
are easier for users to interpret than other classifiers.  It is also
discriminative, as opposed to simply running clustering algorithms
over the top-$L$ tuples while ignoring low-value tuples.  Finally, it
is possible to control the complexity of the tree.  We used the
standard decision tree implementation provided by Python's
\texttt{scikit-learn} package~\cite{pedregosa2011scikit}; given $k$,
the maximum number of clusters to produce, we tune the height
parameter of the decision tree such that the number of ``positive''
leaf nodes (wherein top-$L$ tuples are the majority) as close as
possible to, but no greater than, $k$.

Note that the cluster patterns under this approach can be more complex
than ours, as they may involve non-equality comparisons and negations.
This additional complexity increases the discriminative power, but
makes the patterns more difficult for users to interpret and
internalize---a hypothesis that we shall test with our study.

\textbf{Tasks.~~}
Each study subject is asked to carry out three groups of tasks: the
\emph{varying-method group}, \emph{varying-$k$ group}, and
\emph{varying-$D$ group}.  The third first group is designed to
compare our approach and decision trees.  The last two are designed to
evaluate the utility of making parameters $k$ and $D$ in our approach
specifiable by users.\footnote{\reva{We do not evaluate the utility of
    making $L$ user-specifiable, as it should be evident that what
    ``top'' tuples mean depends on the situation---e.g., a small $L$
    means the user is interested in characterizing really high-valued
    tuples, while a larger $L$ may mean the user is interested in
    tuples whose values are ``good enough.''}}  To account for the
possible effect of users learning and getting better with our
approach, we sequence the task groups differently among study
subjects---half of them go through the sequence (varying-method,
varying-$k$, varying-$D$), while the remaining half go through
(varying-$k$, varying-$D$, varying-method).

All tasks within one group are based on the same aggregate query.
Before beginning the task group, we familiarize the subject with the
aggregate query and result as well as the tasks to perform; we also
show all query result tuples in a table, with top $L$ tuples
highlighted for convenience.  Then, we remove the table of all query
result tuples, and give the subject a series of questions, organized
into three sections, in order.  Each question asks the subject to
classify a given tuple, whose value is hidden, into one of three
categories: ``top'' (the tuple is among the top $L$ of all query
result tuples), ``high'' (the tuple has value above or at the average
among all tuples, but is outside the top $L$), and ``low'' (the tuple
has below-average value).  The three sections are based on the same
``working set'' of clusters, but differ in what information the
subject can access when answering questions:
\begin{itemize}[leftmargin=*]
\itemsep0em
\item \emph{Patterns-only}, 6 questions: When answering these
  questions, the subject can see the clusters and their associated
  patterns, but not the membership within clusters or the table of all
  query result tuples.  This section is designed to test how well the
  cluster patterns help users understand the data.  The 6 tuples to be
  classified are chosen randomly and evenly across the top, high, and
  low categories, and are ordered randomly.  We do not reveal to the
  subject how these tuples are distributed among the three categories,
  as with questions in other sections below.
\item \emph{Memory-only}, 6 questions: The subject can see neither the
  clusters or the table of all query result tuples; all questions must
  be answered from memory.  This section is designed to test the
  extent to which users can internalize the insights learned from the
  cluster patterns for later use.  The 6 tuples are chosen in the same
  way as in the patterns-only section, but we ensure that they are
  distinct from those chosen before.
\item \emph{Patterns+members}, 8 questions: The subject can see the
  clusters, their associated patterns, as well as the result tuples
  they cover; but the table of all query result tuples remains
  inaccessible.  This section is designed to test how our full-fledged
  cluster UI can help user explore data.  The 8 tuples are chosen and
  reordered randomly from the 12 tuples used in the previous two
  sections, such that 4/2/2 are from the top/high/low categories,
  respectively.
\end{itemize}

After these three sections are done, to conclude the task group, we
present two sets of clusters side-by-side: one is the working set that
the subject has been using, and the other one is obtained under a
different setting (but for the same aggregate query and same $L$) for
comparison.  We then ask the subject to choose which set of clusters
would be preferred for the tasks just performed.
\begin{itemize}[leftmargin=*]
\itemsep0em
\item For a \emph{varying-method} task group, the clusters to compare
  are produced by our approach (using \hybrid) and by the method of
  decision trees, under the same $k$ setting ($D$ does not apply to
  decision trees).
\item For a \emph{varying-$k$} task group, the clusters to compare are
  produced by our approach under two different $k$ settings, while
  other parameters remain the same.
\item For a \emph{varying-$D$} task group, the clusters to compare are
  produced by our approach under two different $D$ settings, while
  other parameters remain the same.
\end{itemize}

\textbf{Participants and assignment of tasks.~~}
There are 16 participants in total.  14 of them are graduate students
at Duke University (12 in computer science and 2 others), while the
remaining 2 are Duke undergraduates.  They have varying degrees of
knowledge about databases and SQL language, but all have some prior
experience working with tabular data and are capable of handling all
tasks in our user study.

Recall that each of the three task groups compares two sets of
clusters.  While every subject sees both sets at the end of the task
group, the questions earlier in the task group are based on one
working set chosen between the two.  There are a total of $2^3=8$
possibilities for assigning working sets to the three task groups.  We
assign two subjects to each of these $8$ possibilities.  As discussed
earlier, to account for the learning effect, we make one of these
subjects go through the sequence (varying-method, varying-$k$,
varying-$D$) and the other (varying-$k$, varying-$D$, varying-method).
Finally, we ensure that tuples used in our questions appear equal
number of times over tasks across all subjects.

\textbf{Metrics.~~}
We record the time it takes for each subject to complete each of the
three sections in each of the three task groups.  We evaluate the
accuracy of answers to the questions using the standard accuracy
measure of $\smash{\frac{TP+TN}{TP+FP+FN+TN}}$ based on confusion
matrices~\cite{fawcett2006introduction}, but we define two variants:
\emph{T-accuracy} focuses on the ability to discern the top tuples
from the rest, where ``positive'' means being in top $L$;
\emph{TH-accuracy} focuses on the ability to discern the top and high
tuples from the low ones, where ``positive'' means being in either top
or high category.
}

\ansa{
\subsection{Detailed Analysis for Learning Effect in User Study}
We give some brief conclusions on learning effect in Section~\ref{sec:user-study} in the main paper. In this subsection, we present the quantitative result within one experimental sequence (varying-method first, then varying-$k$ and varying-$D$) in Table~\ref{tbl:user-study-vary-order}. Our conclusion is that the learning effect does not have a huge impact on the leadership within each task group.

The learning effect takes place with the time for each question - since users in Table~\ref{tbl:user-study-vary-order} work with varying-method first, it takes ~20\% more time for patterns-only questions; when it comes to varying-$D$, since users in this group are already familiar with tasks, it takes slightly less time in all three tasks.
}

\subsection{Additional Related Work}\label{sec:more-related}
%\textbf {Other work on result diversification, summarization, and exploration.}
In addition to the papers discussed in Section~\ref{sec:related}, diversification of query results has been extensively studied in the literature for both query answering in databases and other applications \cite{Carbonell+1998, agrawal2009diversifying, GollapudiD2009, Ziegler+WWW2005, Yu+EDBT2009, DBLP:conf/edbt/GkorgkasVDN15, Xin+2006, FanWW13, ZhuGGA07, RaviRT94, Borodin+2012, DBLP:conf/kdd/AbbassiMT13, Tao09, DengFan2014, Vieira+ICDE2011, QinYC12, drosou2012disc}. 
%Examples include diversifying search results \cite{Agrawal+2009}, ranking and diversifying answers in recommendation systems \cite{Ziegler+WWW2005, Yu+EDBT2009, DBLP:conf/edbt/GkorgkasVDN15}, 
%extracting redundancy aware top-k patterns \cite{Xin+2006}, text summarization \cite{Carbonell+1998}, and diversified top graph pattern matching \cite{FanWW13}.
One of the formalisms to capture both \emph{diversity and relevance} of a resultset is to balance these two objectives using a trade-off parameter $\lambda$ specified by the user. This approach, called \emph{MMR (Maximal Marginal Relevance)} aims to reduce redundancy while maintaining relevance of the chosen outputs for the input query, and was first used for re-ranking retrieved documents and in selecting appropriate passage for text summarization \cite{Carbonell+1998}.  Gollapudi and Sharma \cite{GollapudiD2009} studied three variants of the objective function (max-sum, max-min, mono)  %with the trade-off parameter $\lambda$  
based on the MMR criterion. 
%For instance,  given a set of elements $S$, size constraint $k$, a value function $w$ for elements, and a distance function $d$ for pairs of elements,  the max-sum diversification problem aims to select a subset $S' \subseteq S$ of size $k = |S'|$ that maximizes $f(S') = (k-1) \sum_{u \in S'}w(S') + 2\lambda \sum_{u, v \in S'} d(u, v)$ (here $k-1$ and 2 are used for balancing the different number of elements in the first component for relevance and the second component for diversity). 
%They also proposed an axiomatic framework for the diversification problem with several natural axioms, and argued that no objectives can satisfy all at the same time. 
Deng and Fan  \cite{DengFan2014} studied the data complexity and combined complexity for these problems.
%the three objective functions in \cite{GollapudiD2009} for different classes of queries. 
%The DivDB paper\cite{vieira2011divdb} published in VLDB in 2011 uses formula to compute score of the final result to evaluate relevance and diversity. 
%Vieira et al. \cite{Vieira+ICDE2011} studied various existing and new algorithms for the max-sum objective experimentally (with a factor of $(1 - \lambda)$ in the first component for relevance). 
Vieira et al. \cite{Vieira+ICDE2011} conducted an experimental study of existing and new algorithms for the max-sum objective defined in \cite{GollapudiD2009} with some small modifications. 
\cut{
In particular, the objective in \cite{Vieira+ICDE2011} is as follows: given a set of elements $S$, trade-off parameter $\lambda$, size parameter $k$, weight function $w$, and distance function $d$, choose a subset $S' \subseteq S$, such that $|S'| = k$, and $S'$ maximizes $(k-1)(1 - \lambda) \sum_{t \in S'} w(t) + 2\lambda \sum_{t, t' \in S'}d(t, t')$.
}
% and proposed new algorithms in a general diversification framework for diversifying aggregate query answers taking into account their relevance to the input query. 
Fraternali et al. \cite{Fraternali+2012} studied this objective for diversification of objects in a low-dimensional vector space. 
%The parameter $\lambda$ combines  components from two domains that may have different meanings and distributions, making it difficult for the users to choose a good value of $\lambda$.
% This paper uses a parameter $\lambda$ to balance between relevance and diversity of the returned answers $S$ with the following form: $(1- \lambda) TotalScore(S) + \lambda TotalDistance(S)$. 
%Although the parameter $\lambda$  intends to balance between relevance and diversity of the answers, (a) the results are not summarized by representatives, and (b) $\lambda$ combines two different measures that have different meanings and domains, so even with normalization, adding these two measures may not be meaningful. Therefore, it is not suitable for our goal of achieving summarization and diversity including relevance and coverage. 
We compare results from Vieira et al. \cite{Vieira+ICDE2011} with our work in Appendix~\ref{expt:lambda}.\\

%the limitations of these approaches are that (a) it might be difficult for the users to select a good value of $\lambda$, 
%%%%%%%%
%%%%%%%%%%
The \emph{diversity} criterion has been studied algorithmically as the \emph{facility dispersion problem} \cite{RaviRT94}. 
\cut{
place $k$ facilities on $N$ nodes such that some function of the distances between the facilities (max-min or max-avg) is maximized. Minimizing the distance function between facilities nodes gets the standard \emph{$k$-center} or \emph{facility location} problems that handle the \emph{coverage} criterion. Borodin et al.
}
\cite{Borodin+2012} studied the \emph{max-sum dispersion problem} and the \emph{max-sum diversification problem} (as in \cite{GollapudiD2009}) %(choose $p$ points $T$ that maximizes $f(T) + \lambda \times$ total pairwise distances in $T$), and obtained a 2-approximation 
when the value of a subset of elements $w(S')$  is given by a monotone submodular function. Abbassi et al. \cite{DBLP:conf/kdd/AbbassiMT13} studied the diversity maximization of a set of points under matroid constraints. \emph{Diverse skyline} \cite{Tao09} is another related direction.\\
\cut{
where given a size constraint $k$, the goal is to select at most $k$ skyline points that best describe the whole skyline such that each points in the skyline has a close representative point among the chosen points.\\
}
%With a similar objective like DisC diversity \cite{drosou2012disc}, 
Zhu et al. \cite{ZhuGGA07} proposed a ranking algorithm with applications in  text summarization and social network analysis. %called \emph{GRASSHOPPER} 
\cut{
that uses random walk on a graph, prefers elements that are similar to many other items and cover many different groups of elements, and can incorporate a pre-specified ranking as a prior knowledge. 
}
\cut{
They discuss applications of this approach in text summarization and social network analysis. Here the focus is on re-ranking all elements such that the top elements are different from each other and give a broad coverage of the whole element set. \red{Therefore, this approach may not be suitable for database queries when the user is interested in the top elements according to their values.} 
}
Vee et al. \cite{VeeSSBA08} studied the problem of computing diverse query results for non-aggregate queries in online shopping applications.
\cut{
, where the goal is to return a representative diverse set of top answers (with minimum total pairwise similarity) from all the tuples that satisfy the selection conditions entered by the user in a non-aggregate query (\eg, the user can search for different \emph{Honda cars} or different \emph{2007 Honda civic cars}). \cite{VeeSSBA08} assumes that there is a total \emph{diversity ordering} on the attributes, and aims to choose a set that minimizes the sum of \emph{a similarity measure} between all pairs of elements with respect to all possible prefixes of this diversity ordering (the diversity ordering ensures that the higher ordered attributes are varied first before varying the less preferred attributes).  
}
 In the area of Information Retrieval (IR), Zheng et al. \cite{Zheng2012} studied search result diversification. using  $\lambda$-parameterized MMR objective function  \cite{GollapudiD2009, Vieira+ICDE2011}, but their ``diversity score'' is defined as the sum (over possible topics) of product of importance of a subtopic to the input query and how much a document covers this topic.\\ %\red{This is unrelated to summarizing and diversifying answers to SQL queries.  } 
%%%
 Chen and Li \cite{ChenLi2007} considered the problem of categorizing query answers using clusters on a navigational tree by exploiting the query history of the users when different users have diverse preferences. Other approaches include relational data summarization \cite{Zaharioudakis:2000} and Web table search  taking into account schema/instance diversity, table popularity, and redundancy \cite{Nguyen+15}.
For result summarization and exploration in databases, Gebaly et al. \cite{ElGebaly:2014} considered summarization of attributes using $*$ values to find factors that affect a binary (non-aggregate) attribute. Sarawagi explored (\eg, \cite{Sarawagi00}) sophisticated OLAP operators for helping the user visit unvisited interesting parts in a data cube.\\
%%%

\end{sloppypar}

\end{document}